\newif\iftechreport
\techreportfalse
\newif\ifappendix
\appendixfalse

\documentclass[sigplan,nonacm]{acmart}

\usepackage{marvosym}
\usepackage{mathpartir}
\usepackage{xcolor}
\usepackage{listings}
\usepackage{xspace}
\usepackage{hyperref}
\usepackage{flushend}
\usepackage[inline]{enumitem}
\usepackage{placeins}
\usepackage{stmaryrd}
\usepackage{subcaption}
\usepackage{adjustbox}
\usepackage{pgfplots}
\usepackage{layouts}
\pgfplotsset{width=6.2cm, compat=newest}
\usepackage{algorithm}
\usepackage[noend]{algpseudocode}
\usepackage[normalem]{ulem}
\usepackage{cancel}
\usepackage{wrapfig}

\usepackage[most]{tcolorbox}
\usepackage{cleveref}

\usepackage{thmtools}

\makeatletter
\@ifundefined{newcounteralias}{}{%
  \renewcommand\thmt@autorefsetup{%
    \@xa\def
      \csname\thmt@envname autorefname\@xa\endcsname
      \@xa{\thmt@thmname}%
  }%
}
\makeatother

\usepackage{tikz}
\usetikzlibrary{automata, arrows, positioning, calc}
\usetikzlibrary{arrows.meta}
\tikzset{node distance=4cm, 
  every state/.style={ 
    semithick,fill=gray!10},
  initial text={},     
  double distance=2pt, 
  every edge/.style={  
    draw,
    ->,>=stealth,     
    auto,semithick}}

\pgfdeclarelayer{background}
\pgfsetlayers{background,main}

\newtheorem{theorem}{Theorem}

\theoremstyle{definition}
\newtheorem{definition}{Definition}

\theoremstyle{remark}

\usepackage{expkv-cs}
\usepackage{xparse}

\newcommand\lprp[1]{\left(#1\right)}

\mathchardef\mhyphen="2D

\makeatletter

\newcommand{\checknextarg}{\@ifnextchar\bgroup{\gobblenextarg}{}}
\newcommand{\gobblenextarg}[1]{ \join #1\@ifnextchar\bgroup{\gobblenextarg}{}}

\def\bracketed@cmd@primes{}
\def\bracketed@cmd@subs{}
\def\bracketed@cmd@base{}
\def\bracketed@cmd@left{}
\def\bracketed@cmd@right{}

\def\bracketed@cmd@addsub#1{%
  \ifx\bracketed@cmd@subs\empty\else\g@addto@macro\bracketed@cmd@subs{,}\fi
  \g@addto@macro\bracketed@cmd@subs{#1}%
}
\def\bracketed@cmd@scan{%
  \@ifnextchar'{\bracketed@cmd@grabprime}{%
    \@ifnextchar_{\bracketed@cmd@grabsub}{%
      \bracketed@cmd@print
    }%
  }%
}
\def\bracketed@cmd@grabprime'{%
  \g@addto@macro\bracketed@cmd@primes{\prime}%
  \bracketed@cmd@scan
}
\def\bracketed@cmd@grabsub_#1{%
  \bracketed@cmd@addsub{#1}%
  \bracketed@cmd@scan
}
\def\bracketed@cmd@print{%
  \bracketed@cmd@left
  \bracketed@cmd@base
  \ifx\bracketed@cmd@primes\empty\else^{\bracketed@cmd@primes}\fi
  \ifx\bracketed@cmd@subs\empty\else_{\bracketed@cmd@subs}\fi
  \bracketed@cmd@right
}

\NewDocumentCommand{\NewCustomBracketedCommand}{mmmmm}{%
  \NewDocumentCommand{#1}{o}{%
    \def\bracketed@cmd@primes{}%
    \def\bracketed@cmd@left{#2}%
    \def\bracketed@cmd@right{#3}%
    \def\bracketed@cmd@base{#4}%
    \def\bracketed@cmd@subs{#5}%
    \IfNoValueF{##1}{\bracketed@cmd@addsub{##1}}%
    \bracketed@cmd@scan
  }%
}

\NewDocumentCommand{\NewBracketedCommand}{mmm}{%
  \NewCustomBracketedCommand{#1}{\lfloor}{\rceil}{#2}{#3}%
}

\NewDocumentCommand{\NewFullyBracketedCommand}{mmm}{%
  \NewCustomBracketedCommand{#1}{[}{]}{#2}{#3}%
}

\NewDocumentCommand{\NewUnbracketedCommand}{mmm}{
  \NewCustomBracketedCommand{#1}{}{}{#2}{#3}
}

\makeatother

\newcommand*{\ddefeq}{\mathrel{\vcenter{\baselineskip0.5ex \lineskiplimit0pt
                     \hbox{\scriptsize.}\hbox{\scriptsize.}}}%
                      \mathrel{\vcenter{\baselineskip0.5ex \lineskiplimit0pt
                     \hbox{\scriptsize.}\hbox{\scriptsize.}}}
                     =\ }

\newcommand{\rulename}[1]{{\ensuremath{\TirNameStyle{#1}}}}

\newcommand\public{\ensuremath{\bot}\xspace}
\newcommand\private{\ensuremath{\top}\xspace}
\newcommand\type[2]{\langle#1,#2\rangle}
\ekvcSplit\ty{t=\tau, l=l}{\type{#1}{#2}}
\newcommand\tc{\vdash}
\newcommand\itt{\code{Int_{32}}}
\newcommand\ione{\code{Int_1}}
\newcommand\ptr[1]{#1 *}
\newcommand\var[1]{#1}
\newcommand\fname[1]{#1}
\newcommand\bLabel[1]{\code{label}\ #1}
\newcommand\main{\code{main}}
\newcommand\argi{\alpha_i}
\newcommand\arglist{\overline{\alpha}}
\newcommand\philist{\overline{\phi}}
\newcommand\parami{\rho_i}
\newcommand\paramlist{\overline{\rho}}
\newcommand\nil{\code{null}}

\newcommand\Nop{\code{nop}}
\newcommand\Asgn[2]{#1 = #2}
\newcommand\Add[3]{#1 = \code{add} #2, #3}

\newcommand\Alloca[2]{#1 = \code{alloca} #2}
\newcommand\Malloc[2]{#1 = \code{malloc} #2}
\newcommand\Free[1]{\code{free} #1}
\newcommand\Load[3]{#1 = \code{load} #2, #3}
\newcommand\Store[2]{\code{store} #1, #2}
\newcommand\Seq[2]{#1 ; #2}
\newcommand\Br[1]{\code{br}\ #1}
\newcommand\BrCond[3]{\code{br}\ #1, #2, #3}

\newcommand\PhiNode[2]{#1 = \phi #2}
\newcommand\FCall[3]{#1 = \code{call}\ #2 (#3)}
\newcommand\FDef[3]{#1 (#2) \{#3\}}

\newcommand\FHead[3]{#1 (#2)}
\newcommand\Ret[1]{\code{ret}\ #1}
\newcommand\GEP[4]{#1 = \code{gep}\ #2, #3, #4}
\newcommand\Declassify[3]{#1 = \code{declassify}\ #2, #3}
\newcommand\tmin{\ensuremath{{\textsf{min}_\textsf{T}}}\xspace}
\newcommand\emin{\ensuremath{{\textsf{min}_\textsf{TCB}}}\xspace}
\newcommand\objboth{\ensuremath{{\textsf{min}_\textsf{T+TCB}}}\xspace}
\newcommand\objdata{\ensuremath{{\textsf{min}_\textsf{data}}}\xspace}

\newcommand\EncInit[2]{\code{}}
\newcommand\EncKill[1]{\code{}}
\newcommand\EncEnter[2]{\code{}}
\newcommand\EncExit[2]{\code{}}

\newcommand\Var{\ensuremath{\textit{Vars}}\xspace}
\newcommand\Fcn{\ensuremath{\textit{FNames}}\xspace}
\newcommand\Block{\ensuremath{\textit{BId}}\xspace}
\newcommand\func{\ensuremath{\textit{fn}}\xspace}

\newcommand\Vspace{V}
\NewBracketedCommand{\bVspace}{\Vspace}{}
\NewFullyBracketedCommand{\fbVspace}{\Vspace}{}
\newcommand\TVspace{\ensuremath{\mathcal{V}}}
\NewBracketedCommand{\bTVspace}{\TVspace}{}
\NewFullyBracketedCommand{\fbTVspace}{\TVspace}{}

\newcommand\Sspace{\Sigma}
\newcommand\Stspace{\sigma}
\NewBracketedCommand{\bStspace}{\Stspace}{}
\NewFullyBracketedCommand{\fbStspace}{\Stspace}{}
\newcommand\TSspace{\mathcal{E}}
\newcommand\TStspace{\sigma}
\NewBracketedCommand{\bTStspace}{\TStspace}{}
\NewFullyBracketedCommand{\fbTStspace}{\TStspace}{}
\newcommand\Hspace{H}
\newcommand\THspace{\ensuremath{\mathcal{H}}}
\newcommand\Fspace{\Fenv}
\newcommand\TFspace{\TFenv}
\newcommand\Xspace{\Xi}
\newcommand\TXspace{\Theta}

\NewDocumentCommand\Mspace{t'}{\IfBooleanTF{#1}{\Sspace' \cup \Hspace'}{\Sspace \cup \Hspace}}
\NewDocumentCommand\TMspace{t'}{\IfBooleanTF{#1}{\TSspace' \cup \THspace'}{\TSspace \cup \THspace}}

\NewUnbracketedCommand{\prev}{b}{p}
\NewUnbracketedCommand{\cur}{b}{c}
\NewUnbracketedCommand{\trg}{b}{t}
\newcommand\body{\textit{body}}

\NewBracketedCommand{\bloc}{m}{}
\NewBracketedCommand{\bconst}{c}{}
\NewBracketedCommand{\bdonst}{d}{} 
\NewBracketedCommand{\bkappa}{\kappa}{}
\NewBracketedCommand{\bbody}{\!\textit{\body}}{}
\NewBracketedCommand{\bbl}{\text{BlockLabel}}{}
\NewBracketedCommand{\blab}{b}{}
\NewBracketedCommand{\bprev}{b}{p}
\NewBracketedCommand{\bcur}{b}{c}

\NewFullyBracketedCommand{\fbloc}{m}{}
\NewFullyBracketedCommand{\fbconst}{c}{}
\NewFullyBracketedCommand{\fbkappa}{\kappa}{}
\NewFullyBracketedCommand{\fbbody}{\!\body}{}
\NewFullyBracketedCommand{\fbbl}{\text{BlockLabel}}{}
\NewFullyBracketedCommand{\fblab}{b}{}
\NewFullyBracketedCommand{\fbprev}{b}{p}
\NewFullyBracketedCommand{\fbcur}{b}{c}

\newcommand\configtwo[2]{\ensuremath{\langle #1, #2 \rangle}\xspace}
\newcommand\configthree[3]{\ensuremath{\langle #1, #2, #3 \rangle}\xspace}

\newcommand\configfive[5]{\ensuremath{\langle #1, #2, #3, #4, #5 \rangle}\xspace}

\newcommand\configseven[7]{\ensuremath{\langle #1, #2, #3, #4, #5, #6, #7 \rangle}\xspace}
\newcommand\configeight[8]{\ensuremath{\langle #1, #2, #3, #4, #5, #6, #7, #8 \rangle}\xspace}
\newcommand\confignine[9]{\ensuremath{\langle #1, #2, #3, #4, #5, #6, #7, #8, #9 \rangle}\xspace}

\ekvcSplit\scfgmem{
  s=\Sspace,
  h=\Hspace
} {\ensuremath{\langle #1, #2 \rangle}\xspace}

\ekvcSplit\tcfgmem{
  s=\TSspace,
  h=\THspace
} {\ensuremath{\langle #1, #2 \rangle}\xspace}

\ekvcSplit\scfgvx{
  v=\Vspace,
  x=\Xspace
} {\ensuremath{\langle #1, #2 \rangle}\xspace}

\ekvcSplit\bscfgvx{
  v=\bVspace,
  x=\Xspace,
} {\ensuremath{\langle #1, #2 \rangle}\xspace}

\ekvcSplit\tcfgvx{
  v=\TVspace,
  x=\TXspace
} {\ensuremath{\langle #1, #2 \rangle}\xspace}

\ekvcSplit\scfge{
	v=\Vspace,
	e=e
} {\ensuremath{\langle #1, #2 \rangle }\xspace}%

\ekvcSplit\bscfge{
	v=\bVspace,
	e=e
} {\ensuremath{\langle #1, #2 \rangle }\xspace}%

\ekvcSplit\tcfge{
	v=\TVspace,
	e=e
} {\ensuremath{\langle #1, #2 \rangle }\xspace}%

\ekvcSplit\btcfge{
	v=\bTVspace,
	e=e
} {\ensuremath{\langle #1, #2 \rangle }\xspace}%

\ekvcSplit\scfgi{
	h=\Hspace,
	s=\Sspace,
	v=\Vspace,
  prev=\prev,
	i=\Nop
} {\ensuremath{\langle #1, #2, #3, #4, #5 \rangle}\xspace}%

\ekvcSplit\bscfgi{
	h=\Hspace,
	s=\Sspace,
	v=\bVspace,
  prev=\bprev,
	i=\Nop
} {\ensuremath{\langle #1, #2, #3, #4, #5 \rangle}\xspace}%

\ekvcSplit\tcfgi{
	h=\THspace,
	s=\TSspace,
	v=\TVspace,
	e=\Elive,
	prev=\prev,
	mode=\mu,
	i=\Nop
} {\ensuremath{\langle #1, #2, #3, #4, #5, #6, #7 \rangle}\xspace}%

\ekvcSplit\btcfgi{
	h=\THspace,
	s=\TSspace,
	v=\bTVspace,
	e=\Elive,
	prev=\bprev,
	mode=\mu,
	i=\Nop
} {\ensuremath{\langle #1, #2, #3, #4, #5, #6, #7 \rangle}\xspace}%

\ekvcSplit\scfgb{
	h=\Hspace,
	s=\Sspace,
	v=\Vspace,
  x=\Xspace,
	prev=\prev,
	cur=\cur,
	body=\Nop
} {\ensuremath{\langle #1, #2, #3, #4, #5, #6, #7 \rangle}\xspace}%

\ekvcSplit\scfgbn{
	h=\Hspace,
	s=\Sspace,
	v=\Vspace,
  x=\Xspace,
	prev=\prev,
	cur=\cur,
	body=\Nop,
  pow={~}
} {\ensuremath{\langle #1^#8, #2^#8, #3^#8, #4^#8, #5^#8, #6^#8, #7^#8 \rangle}\xspace}%

\ekvcSplit\bscfgb{
	h=\Hspace,
	s=\Sspace,
	v=\bVspace,
	x=\Xspace,
	prev=\bprev,
	cur=\bcur,
  trg=\trg,
	body=\Nop
} {\ensuremath{\langle #1, #2, #3, #4, #5, #6, #7, #8 \rangle}\xspace}

\ekvcSplit\bscfgbn{
	h=\Hspace,
	s=\Sspace,
	v=\bVspace,
  x=\Xspace,
	prev=\prev,
	cur=\cur,
  trg=\trg,
	body=\Nop,
  pow={~}
} {\ensuremath{\langle #1^#9, #2^#9, #3^#9, #4^#9, #5^#9, #6^#9, #7^#9, #8^#9 \rangle}\xspace}%

\ekvcHash\tcfgb
{
	h=\THspace,
	s=\TSspace,
	v=\TVspace,
	e=\Elive,
	x=\TXspace,
	prev=\prev,
	cur=\cur,
	mode=\mu,
	body=\Nop
} {\langle \ekvcValue{h}{#1}, \ekvcValue{s}{#1}, \ekvcValue{v}{#1}, \ekvcValue{e}{#1}, \ekvcValue{x}{#1},
              \ekvcValue{prev}{#1}, \ekvcValue{cur}{#1}, \ekvcValue{mode}{#1}, \ekvcValue{body}{#1}
              \rangle}%

\ekvcHash\tcfgbn
{
	h=\THspace,
	s=\TSspace,
	v=\TVspace,
	e=\Elive,
	x=\TXspace,
	prev=\prev,
	cur=\cur,
	mode=\mu,
	body=\Nop,
  pow={~},
} {\langle \ekvcValue{h}{#1}^\ekvcValue{pow}{#1}, \ekvcValue{s}{#1}^\ekvcValue{pow}{#1}, \ekvcValue{v}{#1}^\ekvcValue{pow}{#1}, \ekvcValue{e}{#1}^\ekvcValue{pow}{#1}, \ekvcValue{x}{#1}^\ekvcValue{pow}{#1},
              \ekvcValue{prev}{#1}^\ekvcValue{pow}{#1}, \ekvcValue{cur}{#1}^\ekvcValue{pow}{#1}, \ekvcValue{mode}{#1}^\ekvcValue{pow}{#1}, \ekvcValue{body}{#1}^\ekvcValue{pow}{#1}
              \rangle}%

\ekvcHash\btcfgb
{
	h=\THspace,
	s=\TSspace,
	v=\bTVspace,
	e=\Elive,
	x=\TXspace,
	prev=\bprev,
	cur=\bcur,
  trg=\trg,
	mode=\mu,
	body=\Nop
}
	{\langle \ekvcValue{h}{#1}, \ekvcValue{s}{#1}, \ekvcValue{v}{#1}, \ekvcValue{e}{#1}, \ekvcValue{x}{#1},
	              \ekvcValue{prev}{#1}, \ekvcValue{cur}{#1}, \ekvcValue{trg}{#1}, \ekvcValue{mode}{#1}, \ekvcValue{body}{#1}
	              \rangle}%

\ekvcSplit\scfgp{
	h=\Hspace,
	v=\Vspace,
	p=p
} {\ensuremath{\langle #1, #2, #3 \rangle}\xspace}%

\ekvcSplit\tcfgp{
	h=\THspace,
	v=\TVspace,
	p=p
} {\ensuremath{\langle #1, #2, #3 \rangle}\xspace}%

\ekvcSplit\scfgres{
  h=\Hspace',
  c=c
} {\ensuremath{\langle #1, #2 \rangle}\xspace}

\ekvcSplit\tcfgres{
  h=\THspace',
  c=c
} {\ensuremath{\langle #1, #2 \rangle}\xspace}

\ekvcSplit\scfgevent{
  h=\Hspace,
  s=\Sspace,
} {\ensuremath{\langle #1, #2 \rangle}\xspace}

\ekvcSplit\eventsub{
  h=\Hspace,
  s=\Sspace,
  sub={},
} {\ensuremath{\langle #1_{#3} , #2_{#3} \rangle}\xspace}
\ekvcSplit\eventpow{
  h=\Hspace,
  s=\Sspace,
  pow={},
} {\ensuremath{\langle #1^{#3} , #2^{#3} \rangle}\xspace}
\ekvcSplit\eventsubpow{
  h=\Hspace,
  s=\Sspace,
  sub={},
  pow={},
} {\ensuremath{\langle #1_{#3}^{#4} , #2_{#3}^{#4} \rangle}\xspace}

\ekvcSplit\scfgx{
  v=\Vspace,
  prev=\prev,
  cur=\cur,
  r=x,
  body=body
} {\ensuremath{\langle #1, #2, #3, #4, #5 \rangle}\xspace}

\ekvcSplit\bscfgx{
    v=\bVspace,
    prev=\bprev,
    cur=\bcur,
    trg=\trg,
    r=x,
    body=\bbody
} {\ensuremath{\langle #1, #2, #3, #4, #5, #6 \rangle}\xspace}

\ekvcSplit\tcfgxcall
	{
		kind=call,
		v=\TVspace,
		prev=b_p,
		cur=b_c,
		mode=\mu,
		x=x,
		body=body
	}
	{\configseven{#1}{#2}{#3}{#4}{#5}{#6}{#7}}

\ekvcSplit\tcfgxenclave
	{
		prev=b_p,
		cur=b_c,
		mode=\mu,
		body=body
	}
	{\configfive{\eenter}{#1}{#2}{#3}{#4}}

\ekvcSplit\btcfgxcall
	{
		kind=call,
		v=\TVspace,
		prev=b_p,
		cur=b_c,
		mode=\mu,
		trg = b_t,
		x=x,
		body=body
	}
	{\configeight{#1}{#2}{#3}{#4}{#5}{#6}{#7}{#8}}

\ekvcSplit\btcfgxenclave
	{
		prev=b_p,
		cur=b_c,
		mode=\mu,
		trg = b_t,
		body=body
	}
	{\configfive{\eenter}{#1}{#2}{#3}{#4}{#5}}

\ekvcSplit\llblock
{
  w=4.0cm,
  pos={below right=of b2},
  x=0.0cm,
  y=0.0cm,
  id=b,
  lab=b,
  pc={\pc=l_f},
  code=\Nop,
} {
  \node[rectangle split, rectangle split parts=2, draw, text width=#1, #2, xshift=#3, yshift=#4]
    (#5) {
      $#6 : #7$
    \nodepart{two}
      #8
    }
}
\ekvcSplit\llenclaveblock
{
  w=4.0cm,
  pos={below right=of b2},
  x=0.0cm,
  y=0.0cm,
  id=b,
  lab=b,
  pc={\pc=l_f},
  code=\Nop,
} {
  \node[rectangle split, rectangle split parts=2, rectangle split part fill={green!10, white},draw, text width=#1, #2, xshift=#3, yshift=#4]
    (#5) {
      $#6 : #7$
    \nodepart{two}
      #8
    }
}
\ekvcSplit\llfnheader
{
  f=f,
  id=f,
  r=\ty{},
  l=l_f,
  args=\ty{}\var{x},
} {
  \node[rectangle, draw]
  (#2) {
    \begin{tabular}{| c || c |}
      \hline
      \multicolumn{2}{|l|}{\textsc{Function} $#1$} \\
      \hline
      Type: $#3$ & Args: $#5$ \\
      \hline
    \end{tabular}
  }
}

\ekvcSplit\regblock
{
  id=name,
  text=text,
  color=black,
  x=0.0cm,
  y=0.0cm,
  args={},
} {
  \node[rectangle,draw={#3}, xshift=#4,yshift=#5,#6]
  (#1) {
    \footnotesize
    \code{#2}
  }
}

\newcommand\concat{\mathbin{+\!+}}

\newcommand\high[1]{[#1]}
\newcommand\higher[2]{#1\!\uparrow^{#2}}
\newcommand\declat[2]{\ensuremath{{#1}_{\downarrow_{#2}}}}
\newcommand\restrict[2]{\left. #1 \right|_{#2}}

\newcommand\mergepoint[1]{\ensuremath{\mathrm{ipdom}(#1)}}

\newcommand\fresh[1]{\ensuremath{\mathrm{fresh}(#1)}}
\newcommand\flatten[1]{\ensuremath{\mathrm{flatten}(#1)}}
\newcommand\region[1]{\ensuremath{\mathrm{region}(#1)}}

\newcommand\raises{\triangleright}

\newcommand\sstepe{\longrightarrow_{e}}
\newcommand\sstepi{\longrightarrow_{i}}
\newcommand\sstepb{\longrightarrow_{b}}
\newcommand\sstepbtr{\longrightarrow^{\mathrm{tr}}_{b}}

\newcommand\strstepb{\longrightarrow_{b}^\ast}
\newcommand\sstepp{\longrightarrow_{p}}
\newcommand\ssteptr{\longrightarrow_{\text{tr}}}
\newcommand\bsstepe{\longrightarrow_{[e]}}
\newcommand\bsstepi{\longrightarrow_{[i]}}
\newcommand\bsstepb{\longrightarrow_{[b]}}
\newcommand\bstrstepb{\longrightarrow_{[b]}^*}
\newcommand\bsstepp{\longrightarrow_{[p]}}
\newcommand\bsstepbtr{\longrightarrow^{\mathrm{tr}}_{[b]}}
\newcommand\bsstepptr{\longrightarrow^{\mathrm{tr}}_{[p]}}

\newcommand\tstepe{\longrightarrow_{e}}
\newcommand\tstepi{\longrightarrow_{i}}
\newcommand\tstepb{\longrightarrow_{b}}
\newcommand\ttrstepb{\longrightarrow_{b}^\ast}
\newcommand\tstepp{\longrightarrow_{p}}
\newcommand\tsteptr{\longrightarrow_{\text{tr}}}
\newcommand\btstepe{\longrightarrow_{[e]}}
\newcommand\btstepi{\longrightarrow_{[i]}}
\newcommand\btstepb{\longrightarrow_{[b]}}
\newcommand\bttrstepb{\longrightarrow_{[b]}^*}

\newcommand\flowsto{\mathrel{\sqsubseteq}}
\newcommand\consalt{\mathbin{::}}
\newcommand\nflowsto{\mathrel{\nsqsubseteq}}
\newcommand\join{\mathrel{\sqcup}}

\newcommand\lequiv{\approx_\bot}
\newcommand\sequiv{\simeq^\Menv_l}
\newcommand\biglequiv[2]{\mathrm{eqv}_\bot\lprp{\begin{array}{c}#1 \\ #2\end{array}}}
\newcommand\mequiv{\approx_M}
\newcommand\isbr[1]{\mathfrak{Br}(#1)}
\newcommand\isub[1]{\mathfrak{Ub}(#1)}

\newcommand\srctypecheckcom[2]{#1 \vdash #2}

\newcommand\TypeEnv{\Gamma}
\newcommand\TTypeEnv{\ensuremath{\mathcal{G}}}

\newcommand\Benv{\beta}
\newcommand\TBenv{\ensuremath{\mathcal{B}}}
\newcommand\Fenv{F}
\newcommand\TFenv{\ensuremath{\mathcal{F}}}
\newcommand\Menv{\ensuremath{\mathsf{S}}\xspace}
\newcommand\TMenv{\ensuremath{\mathcal{S}}}
\newcommand\pc{\mathrm{pc}}

\newcommand\lf{l_f}

\newcommand\transrccontext{\TypeEnv}

\newcommand\trantrgcontext{\TTypeEnv}

\newcommand\transtypingcontext{\transrccontext, \trantrgcontext}
\newcommand\transtypingcontextbody{\transrccontext, \pc, \trantrgcontext}
\renewcommand\lg{l_g}
\newcommand\lh{l_h}

\newcommand\typingcontextconfig{\Menv, \Benv, \TypeEnv, \pc, \lf}
\newcommand\trgtypingcontextcom{\TTypeEnv, \pc, \mu, oc, b_c, l_{f}}
\newcommand\trgtypingcontextconfig{\TMenv, \TTypeEnv, \pc, \mu, oc, b_c, \lf}

\ekvcSplit\stctxe{g=\TypeEnv}{#1}
\ekvcSplit\stctxi{g=\TypeEnv, pc=\pc, cur=\cur, lf=\lf}{#1,#2,#3,#4}
\ekvcSplit\stctxb{g=\TypeEnv, pc=\pc, cur=\cur, lf=\lf}{#1,#2,#3,#4}
\ekvcSplit\stctxf{g=\TypeEnv}{#1}
\ekvcSplit\stctxp{g=\TypeEnv}{#1}

\ekvcSplit\ttctxe{g=\TTypeEnv}{#1}
\ekvcSplit\ttctxi{g=\TTypeEnv,mu=\mu,oc=oc,pc=\pc,cur=\cur,lf=\lf}{#1,#2,#3,#4,#5,#6}
\ekvcSplit\ttctxb{g=\TTypeEnv,mu=\mu,oc=oc,pc=\pc,cur=\cur,lf=\lf}{#1,#2,#3,#4,#5,#6}
\ekvcSplit\ttctxf{g=\TTypeEnv}{#1}
\ekvcSplit\ttctxp{g=\TTypeEnv}{#1}

\newcommand\Bmode[1]{ \TBenv(#1)}
\newcommand\modepdom[1]{\code{mode\_exit}(#1)}
\newcommand\translateSrcInt[2]{#1 \rightsquigarrow_f #2}
\newcommand\translateSrc[2]{#1 \rightsquigarrow #2}
\newcommand\translateSrcHeap[3]{#1 \vdash #2 \rightsquigarrow #3}

\newcommand\translatePostProcess[2]{#1 \rightsquigarrow #2}
\newcommand\withconstraints[2]{\langle #1, #2 \rangle}
\newcommand\mode{\mu}
\newcommand\constr{\ensuremath{\mathbb{C}}}

\newcommand\modereg[1]{\code{mode}(#1)}

\newcommand\sirenmemory[1]{(\THspace_0\cup\THspace_{#1}\cup\TSspace_0\cup\TSspace_{#1})}

\newcommand\TT{\mathcal{T}}

\ekvcSplit\tcfgevent{
  h=\THspace,
  s=\TSspace,
} {\ensuremath{\langle #1, #2 \rangle}\xspace}
\ekvcSplit\tcfgeventsub{
  h=\THspace,
  s=\TSspace,
  sub={},
} {\ensuremath{\langle #1_{#3} , #2_{#3} \rangle}\xspace}
\ekvcSplit\tcfgeventpow{
  h=\THspace,
  s=\TSspace,
  pow={},
} {\ensuremath{\langle #1^{#3} , #2^{#3} \rangle}\xspace}
\ekvcSplit\tcfgeventsubpow{
  h=\THspace,
  s=\TSspace,
  sub={},
  pow={},
} {\ensuremath{\langle #1_{#3}^{#4} , #2_{#3}^{#4} \rangle}\xspace}

\newcommand\Ttrace{\mathcal{T}}

\newcommand\hequiv{\simeq^\TMenv_h}

\newcommand\ocstat{\ensuremath{o}}
\newcommand\ocyes{{\small \ensuremath{\Upsilon}}}
\newcommand\ocno{{\small \ensuremath{\eta}}}
\newcommand\Elive{\mathcal{L}} 

\newcommand\oret{\code{oret}}
\newcommand\ocall{\code{ocall}}
\newcommand\eenter{\code{eenter}}
\newcommand\eexit{\code{eexit}}
\newcommand\ekill{\code{ekill}}
\newcommand\ecreate{\code{ecreate}}
\newcommand\preserve{\code{preserve}}

\newcommand\ORet[1]{\oret #1}
\newcommand\OCall[3]{#1 = \ocall #2 (#3)}
\newcommand\EnclaveEnter[1]{\eenter\ #1}
\newcommand\EnclaveExit{\eexit}
\newcommand\EnclaveKill[1]{\ekill\ #1}
\newcommand\EnclaveCreate[2]{\ecreate\ #1, #2}
\newcommand\Preserve[2]{#1 = \preserve~#2}

\newcommand\TgtMalloc[3]{\code{#1 = \code{malloc} #2, #3}}

\newcommand\tgtptr[2]{\langle\ptr{#1}, #2\rangle}
\newcommand\todo[1]{\textcolor{red}{TODO: #1}\xspace}

\newcommand\sam[1]{\colorbox{black}{\textcolor{yellow}{\textbf{#1}}}\xspace}

\newcommand\new[1]{\textcolor{blue}{#1}}

\newcommand\code[1]{{\sf #1}}

\newcommand\dom[1]{\mathrm{dom}(#1)}

\newcommand\infertool{{Inferno}\xspace}
\newcommand\srclang{{SIR}\xspace}
\DeclareRobustCommand\targetlang{{SIREN}\xspace}
\newcommand\targetlangminus{\ensuremath{\text{SIREN}^{-}}\xspace}

\newcommand\memory{\ensuremath{M}\xspace}

\newcommand\Loc{\ensuremath{\textit{Loc}}\xspace}

\newcommand\lattacker{\ensuremath{l}-attacker\xspace}
\newcommand\pubattacker{\ensuremath{\bot}-attacker\xspace}
\newcommand\eattacker{\ensuremath{e}-attacker\xspace}

\newcommand\event{\ensuremath{\varepsilon}\xspace}

\newcommand\emptytrace\varepsilon

\lstnewenvironment{numsnippet}{
\lstset{frame=none,xleftmargin=1em,xrightmargin=1em,backgroundcolor=\color{white},numbers=left,numbersep=3pt}
}{
}
\lstnewenvironment{snippet}{
\lstset{frame=none,xleftmargin=1em,xrightmargin=1em,backgroundcolor=\color{white}}
}{
}

\newcommand{\splitacronym}{\splitr}
\newcommand{\splitr}{SPLITR\xspace}

\newcommand\srcsmallsteptrace[3]{\ensuremath{ #1 \longrightarrow_{tr} #2 \triangleright #3}\xspace}

\newcommand\srcsyntax{\begin{figure*}[t]
  \scalebox{0.83}{
  \begin{minipage}{0.35\textwidth}
  \centering
  \begin{align*}
    x,y,z \in&\ \Var \\
    f,g,h \in&\ \Fcn \\
    b \in&\ \Block \\
    c,d \in&\ \mathbb{Z} \\
    \mathit{Const} \quad \kappa \ddefeq& m \mid c \\
    \iota \ddefeq& 0 \mid \nil \\
    e \ddefeq& \var{x} \mid c \\
    \mathit{Sec\ Label} \quad l \ddefeq& \public \mid \private \\  
    \mathit{Prog\ Counter} \quad \pc \ddefeq& l \mid \declat{l}{b} \\
     & \mid pc_1 \join pc_2 \\
    \mathit{Ground}\ \mathit{type} \quad \tau \ddefeq& \mathrm{Int_1} \mid \mathrm{Int_{32}} \mid \ptr{\tau} \\
    \mathit{Type} \ddefeq& \type{\tau}{l} \\
    \arglist \ddefeq& \cdot \mid \ty{} e, \arglist \\
    \philist \ddefeq& \cdot \mid [b, e], \philist \\
    \paramlist \ddefeq& \cdot \mid \type{\tau}{l}\var{x}, \paramlist \\
  \end{align*}
  \end{minipage}
}\scalebox{0.83}{
  \begin{minipage}{0.39\textwidth}
  \centering
  \begin{align*}
    i \ddefeq& \Add{\var{x}}{\type{\tau}{l} e_1}{e_2} \mid \Alloca{\var{x}}{\type{\tau}{l}} \\
          & \mid \Malloc{\var{x}}{\type{\tau}{l}} \mid \Free{\type{\tau*}{l}\var{x}} \\
          & \mid \Load{\var{x}}{\type{\tau}{l_1}}{\type{\ptr{\tau}}{l_2} \var{y}} \\
          & \mid \Store{\type{\tau}{l_1} e}{\type{\ptr{\tau}}{l_2} \var{y}} \mid \PhiNode{\var{x}}{\type{\tau}{l}}{\philist} \\
          & \mid \GEP{\var{x}}{\type{\tau_1}{l_1}}{\type{\ptr{\tau_1}}{l_1}y}{\type{\tau_2}{l_2}e} \\
          & \mid \FCall{\var{x}}{\type{\tau}{\lf}\fname{f}}{\arglist} \\
          & \mid \Declassify{x}{\type{\tau}{l}}{e} \\
          & \mid \colorbox{gray!50}{\ensuremath{\Asgn{\var{x}}{\type{\tau}{l}e}}} \mid \colorbox{gray!50}{\Nop} \\
    t \ddefeq& \Br{\var{b}} \mid \BrCond{\type{\ione}{l} e}{\var{b_1}}{\var{b_2}} \\
          & \mid \Ret{\type{\tau}{l} e} \\
    \mathit{body} \ddefeq& t \mid \Seq{i}{\mathit{body}}\\
    \mathit{block}\ddefeq& b:\mathit{body}\\
    \func \ddefeq& \FDef{\type{\tau}{\lf} \fname{f}}{\paramlist}{\overline{\mathit{block}}}\\
    p \ddefeq& \overline{\func} \\
  \end{align*}
  \end{minipage}
}\scalebox{0.83}{
  \begin{minipage}{.4\textwidth}
  \centering
  \begin{align*}
     m \in&\ \Loc \\
    \Menv :&\ m \mapsto \mathit{Type} \\
    \TypeEnv :&\ x \mapsto \mathit{Type} \\
    \Benv :&\ b \mapsto \pc \\
    \Vspace :&\ \Var \rightarrow \mathit{Const} \\
    \Fspace :&\ \Fcn \cup \Block \rightarrow \func \\
    \Xspace :&\ \varnothing \mid \scfgx{}::\Xspace \\
\Loc \ddefeq&\ \Loc_{\Hspace} \cup \Loc_{\Sspace} \\
    \Loc_{\Hspace} \ddefeq&\ \bigcup_l \Loc_{\Hspace,l} \\
    \Loc_{\Sspace} \ddefeq&\ \bigcup_{i,l} \Loc_{\Stspace_i, l} \\
    \mathit{Memory} \quad \memory :&\ \Loc \rightarrow \mathit{Const} \\
    \mathit{Heap}\ \Hspace :&\ \Loc_{H} \rightarrow \mathit{Const} \\
    \mathit{Stack}\ \Stspace :&\ \Loc_{\Stspace, \bot} \cup \Loc_{\Stspace, \top} \rightarrow \mathit{Const} \\
    \Sspace :&\ \varnothing \mid \Stspace :: \Sspace
  \end{align*}
  \end{minipage}
  }

  \caption{\srclang Syntax}
  \label{fig:srcsyntax}
  \vspace{-10pt}
\end{figure*}}

\newcommand\tealbot{\textcolor{teal}{\bot}}
\newcommand\redbot{\textcolor{red}{\bot}}
\newcommand\bluetop{\textcolor{blue}{\top}}
\newcommand\redc{\ty{t=\ione,l=\redbot}\textcolor{red}{c_1}}
\newcommand\bluec{\ty{t=\ione,l=\bluetop}\textcolor{blue}{c_2}}

\newcommand\blocklabeldiagram{
  \begin{tikzpicture}
    \llfnheader{r=\type{\itt}{\textcolor{teal}{\bot}},args={\redc, \bluec}};
    \llblock{w=5.4cm, pos={below=of f}, y=3.6cm, id=entry, lab=\code{entry},
      pc={\tealbot},
      code=$\BrCond{\redc}{\code{while.end}}{\code{while.cond}}$
    };
    \llblock{w=5.4cm, pos={below=of entry}, y=3.6cm, id=wc, lab=\code{while.cond},
      pc={\tealbot \join \declat{\redbot}{\code{while.end}} \join \declat{\bluetop}{\code{while.end}}},
      code=$\BrCond{\bluec}{\code{while.end}}{\code{while.body}}$
    };
    \llblock{w=5.4cm, pos={below=of wc}, y=3.6cm, id=wb, lab=\code{while.body},
      pc={\tealbot \join \declat{\redbot}{\code{while.end}} \join \declat{\bluetop}{\code{while.end}}},
      code=$\Br{\code{while.cond}}$
    };
    \llblock{w=3.6cm, pos={below=of wb},x=-1.7cm,y=3.6cm, id=we, lab=\code{while.end},
      pc={\tealbot},
      code=$\Br{\code{return}}$
    };
    \llblock{w=2.4cm, pos={below=of wb},x=2.0cm,y=3.6cm, id=ret, lab=\code{return},
      pc={\tealbot},
      code=$\Ret{\ty{t=\itt,l=\tealbot}0}$
    };
    \draw[->, =>stealth, auto, semithick]
      (f.south) -- (entry.north);
    \begin{scope}[ transform canvas={xshift=1.5cm}]
      \draw[->, =>stealth, auto, semithick]
        (entry) -- (wc);
      \draw[->, =>stealth, auto, semithick]
        (wc)  -- (wb);
    \end{scope}
    \draw[->, =>stealth, auto, semithick]
      (entry.270) |- +(-2.8,-0.2) -| (we.163);
    \draw[->, =>stealth, auto, semithick]
      (wc.270) |- +(-2,-0.2) -| (we.160);
    \draw[->, =>stealth, auto, semithick]
      (we.270) |- +(2,-0.2) -| +(2.15,0.2) |- (ret.170);
    \draw[->, =>stealth, auto, semithick]
      (wb.200) |- +(3.2,-0.2) -| +(4.5,2.9) -| (wc.13);
  \end{tikzpicture}
}

\ekvcSplit\srceconst{c={\scfge{e=c} \sstepe c}}{
  \inferrule*[Lab=E-Const]
  {~}
  {#1}
}

\ekvcSplit\srcevar{p={\Vspace(x) = \kappa},c={ \scfge{e=\var{x}} \sstepe \kappa }}{
  \inferrule*[Lab=E-Var]
  {#1}
  {#2}
}

\ekvcSplit\srceasgn{
  p1={\configtwo{\Vspace}{e} \sstepe \kappa},
  pb={\Vspace' = \Vspace[x\mapsto \kappa]},
  c={\scfgi{i=\Asgn{\var{x}}{\type{\tau}{l} e}} \sstepi  \scfgi{v=\Vspace'}}
}{
  \inferrule*[Lab=E-Asgn]
  {#1 \\ #2}
  {#3}
}

\ekvcSplit\srceadd{
  p1={\configtwo{\Vspace}{e_i} \sstepe c_i},
  p2={\Vspace' = \Vspace[x\mapsto c_1+c_2]},
  c={\scfgi{i=\Add{\var{x}}{\type{\tau}{l}  e_1}{e_2}} \sstepi  \scfgi{v=\Vspace'}}}{
  \inferrule*[Lab=E-Add]
  {#1 \\ #2}
  {#3}
}

\ekvcSplit\srcealloca{
  p1={m = \mathrm{next}(\Menv,\Sspace,\ty{}) \\ \Menv(m) = \ty{}},
  p2={\Stspace(m) = \varnothing},
  p3={\Vspace'=\Vspace[x\mapsto m]},
  p4={\Stspace' = \Stspace[m\mapsto\iota]},
  c={\scfgi{s=\Stspace::\Sspace, i=\Alloca{\var{x}}{\type{\tau}{l}}} \sstepi  \scfgi{s=\Stspace'::\Sspace, v=\Vspace'}}}{
  \inferrule*[Lab=E-Alloca]
  {#1 \\ #2 \\\\ #3 \\ #4}
  {#5}
}

\ekvcSplit\srcemalloc{
  p1={m = \mathrm{next}(\Menv,\Hspace,\ty{}) \\ \Menv(m) = \ty{}},
  p2={\Hspace(m) = \varnothing},
  p3={\Vspace'=\Vspace[x\mapsto m]},
  p4={\Hspace' = \Hspace[m\mapsto\iota]},
  c={\scfgi{i=\Malloc{\var{x}}{\type{\tau}{l}}} \sstepi
      \scfgi{h=\Hspace', v=\Vspace'}}
} {
  \inferrule*[Lab=E-Malloc]
  {#1 \\ #2 \\\\ #3 \\ #4}
  {#5}
}

\ekvcSplit\srcefree{
  p1={\Vspace(y) = m},
  p2={m \in \Loc_{\Hspace}},
  p3={\Hspace(m) \ne \varnothing},
  p4={\Hspace' = \Hspace[m\mapsto\varnothing]},
  c={\scfgi{i=\Free{\type{\tau*}{l}\var{y}}} \sstepi
    \scfgi{h=\Hspace'}}}{
  \inferrule*[Lab=E-Free]
  {#1 \\ #2 \\\\ #3 \\ #4}
  {#5}
}

\ekvcSplit\srceload{
  p1={\Vspace(y) = m},
  p2={m\ne\nil},
  p3={\Mspace(m) = \kappa},
  p4={\mathrm{isptr}(\tau) \implies \kappa \in \region{m}},
  p5={\Vspace'=\Vspace[x\mapsto \kappa]},
  c={\scfgi{i=\Load{\var{x}}{\type{\tau}{l_1}}{\type{\ptr{\tau}}{l_2}\var{y}}} \sstepi
    \scfgi{v=\Vspace'}}}{
  \inferrule*[Lab=E-Load]
  {#1 \\ #2 \\ #3 \\\\ #4 \\ #5}
  {#6}
}

\ekvcSplit\srcestore{
  p1={\configtwo{\Vspace}{e} \sstepe \kappa},
  p2={\Vspace(y) = m},
  p3={\Menv(m) = \ty{t=\tau'*,l=l_2} \implies \kappa \in \region{m}},
  p4={m\ne0},
  p5={\Mspace(m)\neq \varnothing},
  p6={\Mspace' = \Mspace(m\mapsto \kappa)},
  c={\scfgi{i=\Store{\type{\tau}{l_1} e}{\type{\ptr{\tau}}{l_2} y}} \sstepi
   \scfgi{h=\Hspace',s=\Sspace'}}}{
  \inferrule*[Lab=E-Store]
  {#1 \\ #2 \\ #3 \\\\ #4 \\ #5 \\ #6}
  {#7}
}

\ekvcSplit\srcephi{
  p1={\exists i\in [1, n]. b_i = \prev},
  p2={\scfge{e=e_i} \sstepe \kappa},
  p3={\Vspace' = \Vspace[x\mapsto\kappa]},
  c={\scfgi{i=\PhiNode{\var{x}}{\type{\tau}{l}}{[b_1, e_1] \dots [b_n,e_n]}} \sstepi  
   \scfgi{v=\Vspace',i=\Nop}}
} {
  \inferrule*[Lab=E-Phi]
  {#1 \\ #2 \\ #3}
  {#4}
}

\ekvcSplit\srcegep{
  p1={\Vspace(y) = m},
  p2={\scfge{} \sstepe c},
  p3={m' = m + c},
  p4={m' \in \region{m}},
  p5={\Vspace' = \Vspace[x \mapsto m']},
  c={\scfgi{i=\GEP{\var{x}}{\ty{t=\tau_1,l=l_1}}{\ty{t=\tau_1*,l=l_1}\var{y}}{\ty{t=\tau_2,l=l_2}e}} \sstepi
    \scfgi{v=\Vspace'}}}{
  \inferrule*[Lab=E-GEP]
  {#1 \\ #2 \\ #3 \\\\ #4 \\ #5}
  {#6}
}

\ekvcSplit\srcenop{
  p1={~},
  c={\scfgb{body=\Seq{\Nop}{body}} \sstepb  \scfgb{body=body}}}{
  \inferrule*[Lab=E-Nop]
  {#1}
  {#2}
}

\ekvcSplit\srceinstruction{
  p1={\scfgi{i=i} \sstepi \scfgi{h=\Hspace', s=\Sspace', v=\Vspace',i=i'}},
  c={\scfgb{body=\Seq{i}{body}} \sstepb  \scfgb{h=\Hspace', s=\Sspace', v=\Vspace', body=\Seq{i'}{body}}}}{
  \inferrule*[Lab=E-Instruction]
  {#1}
  {#2}
}

\ekvcSplit\srcebrconditional{
  p1={\configtwo{\Vspace}{e} \sstepe d},
  p2={d\in\{0,1\}},
  c={\scfgb{body=\BrCond{\type{\ione}{l} e}{\var{b_0}}{\var{b_1}}} \sstepb \\\\ \scfgb{body=\Br{\var{b_d}}}}}{
  \inferrule*[Lab=E-Br-Conditional]
  {#1 \\ #2}
  {#3}
}

\ekvcSplit\srcebrunconditional{
  p1={\Fspace(\cur) = \FDef{\ty{} f}{\paramlist}{\ldots; b: body; \ldots}},
  p2={\prev'=\cur},
  p3={\cur' = b},
  c={\scfgb{body=\Br{\var{b}}} \sstepb  \scfgb{prev=\prev', cur=\cur', body=body}}}{
  \inferrule*[Lab=E-Br-Unconditional]
  {#1 \\ #2 \\ #3}
  {#4}
}

\ekvcSplit\srceret{
  p1={\configtwo{\Vspace}{e} \sstepe \kappa},
  p2={\Xspace = \scfgx{v=\Vspace',prev=\prev',cur=\cur'}::\Xspace'},
  p3={\Vspace'' = \Vspace'[x\mapsto\kappa]},
  c={\scfgb{s=\Stspace::\Sspace,body=\Ret{\type{\tau}{l} e}} \sstepb  \scfgb{v=\Vspace'', x=\Xspace', prev=\prev', cur=\cur', body=\body}}
} {
  \inferrule*[Lab=E-Ret]
  {#1 \\ #2 \\ #3}
  {#4}
}

\ekvcSplit\srcecall{
  p1={\Fspace(g) = \type{\tau}{l_g} g(\paramlist) \{\overline{b_i: body_1}\}},
  p2={\configtwo{\Vspace}{\argi} \sstepe \kappa_i},
  p3={\Vspace' = \{\parami : \kappa_i\}},
  p4={\fresh{\Stspace}},
  p5={\Xspace' = \scfgx{}::\Xspace},
  c={\scfgb{body=\Seq{\FCall{\var{x}}{\type{\tau}{l_g} g}{\arglist}}{body}} \sstepb \\\\\scfgb{s=\Stspace::\Sspace, v=\Vspace', prev=\varnothing, cur=b_1, x=\Xspace',body=body_1}}}{
  \inferrule*[Lab=E-Call]
  {#1 \\ #2 \\ #3 \\\\ #4 \\ #5}
  {#6}
}

\ekvcSplit\srceprogram{
  p1={\forall i.\ \func_i = \FDef{\type{\tau}{l} f_i}{\paramlist}{\overline{block}}},
  p2={\Fspace_f=\{f_i :\func_i \mid \forall i\}},
  p3={\Fspace_b= \{b : \func_i \mid \forall i,\ \forall b \in \func_i\}},
  p4={\Fspace = \Fspace_f \cup \Fspace_b},
  p5={\func_0 = \FDef{\type{\itt}{\bot} \main}{}{\overline{block}_{\main}}},
  p6={\forall j.\  x_{j} \in V},
  p7={\scfgb{s=\varnothing, x=\varnothing, prev=\varnothing, cur=\varnothing, body=\Seq{\FCall{\var{x}}{\type{\itt}{\public} \main}{}}{\Ret{\var{x}}}}
    \strstepb  
    \scfgb{h=\Hspace', s=\varnothing, v=\Vspace', x=\varnothing, prev=\varnothing, cur=\varnothing, body=\Ret{x}}},
  c={\langle H, V, \func_0, \dots, \func_n \rangle \sstepp \configtwo{\Hspace'}{\Vspace'(x)}}}{
  \inferrule*[Lab=E-Program]
  {#1 \\ #5 \\ #7 }
  {#8}
}

\ekvcSplit\srcstrefl{
  p1={~},
  c={\type{\tau}{l} \flowsto \type{\tau}{l}}}{
  \inferrule*[Lab=ST-Refl]
  {#1}
  {#2}
}

\ekvcSplit\srcstint{
  p1={\tau\in\{\ione, \itt\}},
  p2={l_1 \flowsto l_2},
  c={\type{\tau}{l_1} \flowsto \type{\tau}{l_2}}}{
  \inferrule*[Lab=ST-Int]
  {#1 \\ #2}
  {#3}
}

\ekvcSplit\flowrefl{
  c={l \flowsto l}
} {
  \inferrule*[Lab=Flow-Refl]
  {~}
  {#1}
}

\ekvcSplit\floword{
  c={\bot \flowsto \top}}{
  \inferrule*[Lab=Flow-Ord]
  {~}
  {#1}
}

\ekvcSplit\flowtrans{
  p1={l_1\flowsto l_2},
  p2={l_2\flowsto l_3},
  c={l_1 \flowsto l_3}}{
  \inferrule*[Lab=Flow-Trans]
  {#1 \\ #2}
  {#3}
}

\ekvcSplit\flowdecl{
  p1={l_1 \flowsto l_2},
  c={\declat{l_1}{b} \flowsto l_2}}{
  \inferrule*[Lab=Flow-Decl]
  {#1}
  {#2}
}

\ekvcSplit\flowrdecl{
  p1={~},
  c={\bot \flowsto \declat{l}{b}}
}{
  \inferrule*[Lab=Flow-RDecl]
  {#1}
  {#2}
}

\ekvcSplit\flowasym{
  p1={l_1\flowsto l_2},
  p2={l_2\flowsto l_1},
  c={l_1=l_2}
} {
  \inferrule*[Lab=Flow-Asym]
  {#1 \\ #2}
  {#3}
}

\ekvcSplit\flowstar{
  p1={~},
  c={\declat{l}{\bigstar}=l}
} {
  \inferrule*[Lab=Flow-Star]
  {#1}
  {#2}
}

\ekvcSplit\flowljoin{
  p1={\pc_1 \flowsto \pc_3},
  p2={\pc_2 \flowsto \pc_3},
  c={\pc_1 \join \pc_2 \flowsto_b \pc_3}}{
  \inferrule*[Lab=BFlow-LJoin]
  {#1 \\\\ #2}
  {#3}
}

\ekvcSplit\flowrjoin{
  p1={\pc_1 \flowsto \pc_2\lor \pc_1\flowsto_b\pc_3},
  c={\pc_1 \flowsto \pc_2 \join \pc_3}}{
  \inferrule*[Lab=Flow-RJoin]
  {#1}
  {#2}
}

\ekvcSplit\bflowsuper{
  p1={l_1 \flowsto l_2},
  c={l_1 \flowsto_{b} l_2}}{
  \inferrule*[Lab=BFlow-Super]
  {#1}
  {#2}
}

\ekvcSplit\bflowhigh{
  p1={l_1 \flowsto l_2},
  p2={b_1 \neq b_2},
  c={\declat{l_1}{b_1} \flowsto_{b_2} l_2}}{
  \inferrule*[Lab=BFlow-High]
  {#1 \\ #2}
  {#3}
}

\ekvcSplit\bflowdecl{
  p1={~},
  c={\declat{l}{b} \flowsto_{b} \bot}}{
  \inferrule*[Lab=BFlow-Decl]
  {#1}
  {#2}
}

\ekvcSplit\bflowljoin{
  p1={\pc_1 \flowsto_b \pc_3},
  p2={\pc_2 \flowsto_b \pc_3},
  c={\pc_1 \join \pc_2 \flowsto_b \pc_3}}{
  \inferrule*[Lab=BFlow-LJoin]
  {#1 \\\\ #2}
  {#3}
}

\ekvcSplit\bflowrjoin{
  p1={\pc_1 \flowsto_b \pc_2\lor \pc_1\flowsto_b\pc_3},
  c={\pc_1 \flowsto_b \pc_2 \join \pc_3}}{
  \inferrule*[Lab=BFlow-RJoin]
  {#1}
  {#2}
}

\ekvcSplit\srctnull{
  c={\tc \nil : \ty{t=\tau*}}
}
{
  \inferrule*[Lab=T-Null]
  {~}
  {#1}
}

\ekvcSplit\srctvar{
  p1={\TypeEnv(x) = \type{\tau}{l}},
  c={\stctxe{} \tc  \var{x} : \type{\tau}{l}}}{
  \inferrule*[Lab=T-Var]
  {#1}
  {#2}
}

\ekvcSplit\srctsub{
  p1={\stctxe{} \tc e : \type{\tau}{l_1}},
  p2={\type{\tau}{l_1} \flowsto \type{\tau}{l_2}},
  c={\stctxe{} \tc  e : \type{\tau}{l_2}}}{
  \inferrule*[Lab=T-Sub]
  {#1 \\ #2}
  {#3}
}

\ekvcSplit\srctnop{
  p1={~},
  c={\stctxb{} \tc{\Nop}}}{
  \inferrule*[Lab=T-Nop]
  {#1}
  {#2}
}

\ekvcSplit\srctasgn{
  p1={\TypeEnv(x) = \type{\tau}{l}},
  p2={\TypeEnv \tc e : \type{\tau}{l}},
  p3={\pc \flowsto l},
  c={\stctxb{} \tc{\Asgn{\var{x}}{\type{\tau}{l}e}}}}{
  \inferrule*[Lab=T-Asgn]
  {#1 \\ #2 \\ #3}
  {#4}
}

\ekvcSplit\srctadd{
  p1={\TypeEnv(x) = \type{\tau}{l}},
  p2={\forall i.\ \TypeEnv \vdash e_i : \type{\tau}{l}},
  p3={\pc\flowsto l},
  c={\stctxb{} \tc{\Add{\var{x}}{\type{\tau}{l} e_1}{e_2}}}}{
  \inferrule*[Lab=T-Add]
  {#1 \\ #2 \\ #3}
  {#4}
}

\ekvcSplit\srctalloca{
  p1={\TypeEnv(x) = \type{\tau*}{l}},
  p2={\pc \flowsto l},
  c={\stctxb{} \tc{\Alloca{\var{x}}{\type{\tau}{l}}}}}{
  \inferrule*[Lab=T-Alloca]
  {#1 \\ #2}
  {#3}
}

\ekvcSplit\srctmalloc{
  p1={\TypeEnv(x) = \type{\tau*}{l}},
  p2={\pc \flowsto l},
  c={\stctxb{} \tc{\Malloc{\var{x}}{\type{\tau}{l}}}}}{
  \inferrule*[Lab=T-Malloc]
  {#1 \\ #2}
  {#3}
}

\ekvcSplit\srctfree{
  p1={\TypeEnv(y) = \type{\tau*}{l}},
  p2={\pc \flowsto l},
  c={\stctxb{} \tc{\Free{\type{\tau*}{l}}{\var{y}}}}}{
  \inferrule*[Lab=T-Free]
  {#1 \\ #2}
  {#3}
}

\ekvcSplit\srctstore{
  p1={\TypeEnv \tc e : \type{\tau}{l_1}},
  p2={\TypeEnv \tc \var{y} : \type{\tau*}{l_2}},
  p3={\type{\tau}{l_1} \flowsto \type{\tau}{l_2}},
  p4={\pc \flowsto l_2},
  c={\stctxb{} \tc{\Store{\type{\tau}{l_1}e}{\type{\tau*}{l_2}\var{y}}}}}{
  \inferrule*[Lab=T-Store]
  {#1 \\ #2 \\\\ #3 \\ #4}
  {#5}
}

\ekvcSplit\srctload{
  p1={\TypeEnv(x)=\type{\tau}{l_1}},
  p2={\TypeEnv \tc \var{y} : \type{\tau*}{l_2}},
  p3={\type{\tau}{l_2} \flowsto \type{\tau}{l_1}},
  p4={\pc \flowsto l_1},
  c={\stctxb{} \tc{\Load{\var{x}}{\type{\tau}{l_1}}{\type{\tau*}{l_2}\var{y}}}}}{
  \inferrule*[Lab=T-Load]
  {#1 \\ #2 \\\\ #3 \\ #4}
  {#5}
}

\ekvcSplit\srctphi{
  p1={\TypeEnv(x) = \type{\tau}{l}},
  p2={\forall i.\ \TypeEnv \tc e_i : \type{\tau}{l} },
  p3={\forall i.\ \Benv(b_i) \flowsto l},
  p4={\pc\flowsto l},
  c={\stctxb{} \tc{\PhiNode{\var{x}}{\type{\tau}{l}}{[b_1, e_1] \dots [b_n,e_n]}}}}{
  \inferrule*[Lab=T-Phi]
  {#1 \\ #2 \\\\ #3 \\ #4}
  {#5}
}

\ekvcSplit\srctgep{
  p1={\TypeEnv(x) = \TypeEnv(y) = \type{\tau_1*}{l_1}},
  p2={\TypeEnv \tc e : \type{\tau_2}{l_2}},
  p5={\tau_2 \in \{\itt, \ione\}},
  p3={l_2 \flowsto l_1},
  p4={\pc\flowsto l_1},
  c={\stctxb{} \tc
	                  {\GEP{\var{x}}
                         {\type{\tau_1}{l_1}}
                         {\type{\tau_1*}{l_1}\var{y}}
                         {\type{\tau_2}{l_2} e}}}}{
  \inferrule*[Lab=T-GEP]
  {#1 \\ #2 \\\\ #3 \\ #4 \\ #5}
  {#6}
}

\ekvcSplit\srctbrunconditional{
  p1={\pc \flowsto_b \Benv(b)},
  c={\stctxb{} \tc{\Br{\var{b}}}}}{
  \inferrule*[Lab=T-Br-Unconditional]
  {#1}
  {#2}
}

\ekvcSplit\srctbrconditional{
  p1={\TypeEnv \tc e : \type{\ione}{l}},
  p2={\trg = \mergepoint{\cur}},
  p3={\forall i.\ \pc \join \declat{l}{\trg} \flowsto_{b_i} \Benv(b_i)},
  c={\stctxb{}\tc  \BrCond{\type{\ione}{l}e}{\var{b_0}}{\var{b_1}}}}{
  \inferrule*[Lab=T-Br-Conditional]
  {#1 \\ #2 \\\\ #3}
  {#4}
}

\ekvcSplit\srctbody{
  p1={\stctxb{} \tc{i}},
  p2={\stctxb{} \tc{body}},
  c={\stctxb{} \tc{\Seq{i}{body}}}}{
  \inferrule*[Lab=T-Body]
  {#1 \\\\ #2}
  {#3}
}

\ekvcSplit\srctret{
  p1={\TypeEnv \tc e : \type{\tau}{l_f}},
  c={\stctxb{} \tc{\Ret{\type{\tau}{l_f} e}}}}{
  \inferrule*[Lab=T-Ret]
  {#1}
  {#2}
}

\ekvcSplit\srctcall{
  p1={\Fenv(g) = \FHead{\type{\tau}{l_g} \fname{g}}{\type{\tau_1}{l_1}x_1, \type{\tau_2}{l_2}x_2, \dots, \type{\tau_n}{l_n}x_n}{}\{\cdots\}},
  p2={\stctxf{} \tc g},
  p3={\TypeEnv(x) = \type{\tau}{\lg}},
  p4={\pc \flowsto \lg},
  p5={\forall i.\ \TypeEnv \tc e_i : \type{\tau_i}{l_i}},
  c={\stctxb{} \tc
                    {\FCall{\var{x}}
                           {\type{\tau}{\lg} \fname{g}}
                           {\type{\tau_1}{l'_1}e_1, \type{\tau_2}{l'_2}e_2\dots, \type{\tau_n}{l'_n}e_n}
                    }}}{
  \inferrule*[Lab=T-Call]
  {#1 \\ #2 \\\\ #3 \\ #4 \\ #5}
  {#6}
}

\ekvcSplit\srctfunction{
  p1={\Fenv(f) = \FHead{\type{\tau}{\lf} \fname{f}}{\paramlist}{}\{\overline{b_i: body_i}\}},
  p2={\forall i.\ \pc_i = \Benv_f(b_i)},
  p3={\forall i.\ \lf \flowsto \pc_i},
  p4={\forall i.\ \stctxb{cur=b_i, pc=\pc_i} \tc body_i},
  c={ \stctxf{} \tc  f}}{
  \inferrule*[Lab=T-Function]
  {#1 \\\\ #2 \\ #3 \\\\ #4}
  {#5}
}

\ekvcSplit\srctprogram{
  p1={\Fspace(f_0)=\FDef{\type{\itt}{\public}\main}{\arglist}{\overline{block}}},
  p2={\Gamma_i = \Gamma\mid_i},
  p3={\forall i.\ \stctxf{g=\Gamma_{i}} \tc f_i},
  c={ \stctxp{} \tc  f_0, f_1, \dots,f_n}}{
  \inferrule*[Lab=T-Program]
  {#1 \\\\ #2 \\ #3}
  {#4}
}

\ekvcSplit\srctmemoryconfig{
  p1={\forall m\in\Loc.\ \Mspace(m) = \varnothing \lor \Menv \tc \Mspace(m) : \Menv(m)},
  c={\tc \configtwo{\Sspace}{\Hspace}}}{
  \inferrule*[Lab=T-Memory-Config]
  {#1}
  {#2}
}

\ekvcSplit\srctvariables{
  p1={\forall x\in\Var.\ \Menv \tc \Vspace(x) : \TypeEnv(x)},
  c={\stctxe{} \tc  \Vspace}}{
  \inferrule*[Lab=T-Variables]
  {#1}
  {#2}
}

\ekvcSplit\srctfunctions{
  p1={\forall f\in \Fcn.\ \tc f},
  c={ \tc  \Fspace}}{
  \inferrule*[Lab=T-Functions]
  {#1}
  {#2}
}

\ekvcSplit\srctcallstack{
  p1={\Fenv(\cur) = \FHead{\type{\tau}{\lh} \fname{h}}{\type{\tau_1}{l_1}x_1, \type{\tau_2}{l_2}x_2, \dots, \type{\tau_n}{l_n}x_n}{}},
  p2={\pc = \Benv_h(\cur)},
  p3={\TypeEnv_h \tc \Vspace},
  p4={\TypeEnv_h(x) = \type{\tau}{\lf}},
  p5={\Benv_h, \TypeEnv_h, \pc, l_h \tc body},
  p6={\lh\flowsto\pc\flowsto\lf},
  p7={l_h \tc \Xspace},
  c={\lf \tc  \scfgx{}::\Xspace}}{
  \inferrule*[Lab=T-Callstack]
  {#1 \\ #2 \\ #3 \\ #4 \\\\ #5 \\ #6 \\ #7}
  {#8}
}

\ekvcSplit\srctcallstackbase{
  p1={\Xspace = \varnothing},
  c={\Menv, \lf \tc  \Xspace}}{
  \inferrule*[Lab=T-Callstack-Base]
  {#1}
  {#2}
}

\ekvcSplit\srctinstructionconfig{
  p1={\Menv \tc \configtwo{\Sspace}{\Hspace}},
  p2={\Menv, \TypeEnv \tc \Vspace},
  p3={\stctxb{} \tc{i}},
  c={\typingcontextconfig \tc  \scfgi{i=i}}}{
  \inferrule*[Lab=T-Instruction-Config]
  {#1 \\ #2 \\ #3}
  {#4}
}

\ekvcSplit\srctbodyconfig{
  p1={\tc \configtwo{\Sspace}{\Hspace}},
  p2={\TypeEnv \tc \Vspace},
  p3={\lf \tc \Xspace},
  p4={\stctxb{} \tc{body}},
  p5={ \tc \Fspace},
  c={\typingcontextconfig \tc  \scfgb{body=body}}}{
  \inferrule*[Lab=T-Body-Config]
  {#1 \\ #2 \\ #3 \\\\ #4 \\ #5}
  {#6}
}

\ekvcSplit\srctexpressionconfig{
  p1={\stctxe{} \tc \Vspace},
  p2={\Menv, \TypeEnv \tc e : \ty{}},
  c={\Menv,\TypeEnv \tc  \scfge{}}
} {
  \inferrule*[Lab=T-Expression-Config]
  {#1 \\ #2}
  {#3}
}

\ekvcSplit\srctprogramconfig{
  p1={ \tc p},
  p2={\Menv \tc \Hspace},
  p3={\Menv, \TypeEnv \tc \Vspace},
  c={\Menv, \TypeEnv \tc  \scfgp{}}}{
  \inferrule*[Lab=T-Program-Config]
  {#1 \\ #2 \\ #3}
  {#4}
}

\ekvcSplit\srctracenostep{
  c={\scfgb{body=\body} \sstepbtr \scfgb{body=\body} \raises \configtwo{\Hspace}{\Sspace}}
} {
  \inferrule*[Lab=Trace-NoStep]
  {~}
  {#1}
}

\ekvcSplit\srctracenextstep{
  p1={\scfgbn{body=\body, pow=0} \sstepb \scfgbn{body=\body, pow=1}},
  p2={\scfgbn{body=\body, pow=1} \sstepbtr \scfgbn{body=\body, pow=n} \raises \event^1 \consalt \cdots \consalt \event^n},
  p3={\forall i \in [0,n].\ \event^i = \eventpow{pow=i}},
  c={\scfgbn{body=\body, pow=0} \sstepbtr \scfgbn{body=\body, pow=n} \raises \event^0 \consalt \event^1 \consalt \cdots \consalt \event^n}
} {
  \inferrule*[Lab=Trace-NextStep]
  {#1 \\ #2 \\ #3}
  {#4}
}

\ekvcSplit\srceqvevent{
  p1={\Hspace_1\sequiv\Hspace_2},
  p2={\Sspace_1\sequiv\Sspace_2},
  c={\eventsub{sub=1}\sequiv\eventsub{sub=2}}
} {
  \inferrule*[Lab=Eqv-Event]
  {#1 \\ #2}
  {#3}
}

\ekvcSplit\srceqvtracelockstep{
  p1={T_1 = \eventsub{sub=1} :: T_1'},
  p2={T_2 = \eventsub{sub=2} :: T_2'},
  p3={\eventsub{sub=1} \sequiv \eventsub{sub=2}},
  p4={T'_1 \sequiv T'_2}, 
  c={T_1\sequiv T_2}
} {
  \inferrule*[Lab=Eqv-Trace-LockStep]
  {#1 \\ #2 \\\\ #3 \\ #4}
  {#5}
}

\ekvcSplit\srceqvtracehalfstep{
  p1={T_1 = \eventsub{sub=1} :: T'_1},
  p2={T_2 = \eventsub{sub=2} :: T'_2},
  p3={\eventsub{sub=1} \sequiv \eventsub{sub=2}},
  p4={T'_1 \sequiv T_2}, 
  c={T_1\sequiv T_2}
} {
  \inferrule*[Lab=Eqv-Trace-HalfStep]
  {#1 \\ #2 \\\\ #3 \\ #4}
  {#5}
}

\ekvcSplit\srceqvtracesymm{
  p1={T_2 \sequiv T_1}, 
  c={T_1\sequiv T_2}
} {
  \inferrule*[Lab=Eqv-Trace-Symm]
  {#1}
  {#2}
}

\ekvcSplit\srceqvtraceempty{
  p1={T_1 = T_2 = \varnothing},
  c={T_1\sequiv T_2}
} {
  \inferrule*[Lab=Eqv-Trace-Empty]
  {#1}
  {#2}
}

\ekvcSplit\brsrcarithmetichigh{
  p1={~},
  c={[c_1] + [c_2] = c_1 + [c_2] = [c_1] + c_2 = [c_1 + c_2]}}{
  \inferrule*[Lab={[Arithmetic-High]}]
  {#1}
  {#2}
}

\ekvcSplit\brsrcelevationa{
  p1={l\nflowsto\bot},
  c={\higher{\kappa}{l} = [\kappa]}}{
  \inferrule*[Lab={[Elevation-1]}]
  {#1}
  {#2}
}

\ekvcSplit\brsrcelevationb{
  p1={l\flowsto\bot},
  c={\higher{\kappa}{l} = \kappa}}{
  \inferrule*[Lab={[Elevation-2]}]
  {#1}
  {#2}
}

\ekvcSplit\brsrcelevationc{
  p1={~},
  c={\higher{[\kappa]}{l} = [\kappa]}}{
  \inferrule*[Lab={[Elevation-3]}]
  {#1}
  {#2}
}

\ekvcSplit\brsrcmemlookuphigh{
  p1={~},
  c={\Mspace(\fbloc) = \Mspace(m)}}{
  \inferrule*[Lab={[Mem-Lookup-High]}]
  {#1}
  {#2}
}

\ekvcSplit\brsrceconst{
  p1={~},
  c={\scfge{e=\bconst} \bsstepe \bconst}}{
  \inferrule*[Lab={[E-Const]}]
  {#1}
  {#2}
}

\ekvcSplit\brsrcevar{
  p1={\Vspace(x) = \bkappa},
  c={\scfge{e=\var{x}} \bsstepe \bkappa}}{
  \inferrule*[Lab={[E-Var]}]
  {#1}
  {#2}
}

\ekvcSplit\brsrceasgn{
  p1={\scfge{e=e} \bsstepe \bkappa},
  p2={\Vspace' = \Vspace[x\mapsto \higher{\bkappa}{l}]},
  c={\bscfgi{i=\Asgn{\var{x}}{\type{\tau}{l} e}} \bsstepi 
    \bscfgi{v=\Vspace'}}}{
  \inferrule*[Lab={[E-Asgn]}]
  {#1 \\ #2}
  {#3}
}

\ekvcSplit\brsrceadd{
  p1={\scfge{e=e_i} \bsstepe \bconst_i},
  p2={\Vspace' = \Vspace[x\mapsto \higher{(\bconst_1+\bconst_2)}{l}]},
  c={\bscfgi{i=\Add{\var{x}}{\type{\tau}{l} e_1}{e_2}} \bsstepi
    \bscfgi{v=\Vspace'}}}{
  \inferrule*[Lab={[E-Add]}]
  {#1 \\ #2}
  {#3}
}

\ekvcSplit\brsrcealloca{
  p1={m = \mathrm{next}(\Menv,\Sspace,\ty{}) \\ \Menv(m) = \ty{}},
  p2={\bStspace(m) = \varnothing},
  p3={\Vspace'=\Vspace[x\mapsto \higher{m}{l}]},
  p4={\bStspace' = \bStspace{} [m\mapsto\higher{\iota}{l}]},
  c={\bscfgi{s=\bStspace::\Sspace, i=\Alloca{\var{x}}{\type{\tau}{l}}} \bsstepi
    \bscfgi{s=\bStspace'::\Sspace, v=\Vspace'}}}{
  \inferrule*[Lab={[E-Alloca]}]
  {#1 \\ #2 \\\\ #3 \\ #4}
  {#5}
}

\ekvcSplit\brsrcemalloc{
  p1={m = \mathrm{next}(\Menv,\Hspace,\ty{}) \\ \Menv(m) = \ty{}},
  p2={\Hspace(m) = \varnothing},
  p3={\Vspace'=\Vspace[x\mapsto \higher{m}{l}]},
  p4={\Hspace' = \Hspace[m\mapsto\higher{\iota}{l}]},
  c={\bscfgi{i=\Malloc{\var{x}}{\type{\tau}{l}}} \bsstepi
    \bscfgi{h=\Hspace', v=\Vspace'}}}{
  \inferrule*[Lab={[E-Malloc]}]
  {#1 \\ #2 \\\\ #3 \\ #4}
  {#5}
}

\ekvcSplit\brsrcefree{
  p1={\Vspace(y) = \bloc},
  p2={\bloc \in \Loc_{\Hspace}},
  p3={\Hspace(\bloc) \ne \varnothing},
  p4={\Hspace' = \Hspace[\bloc\mapsto\varnothing]},
  c={\bscfgi{i=\Free{\type{\tau*}{l}\var{y}}} \bsstepi
    \bscfgi{h=\Hspace'}}}{
  \inferrule*[Lab={[E-Free]}]
  {#1 \\ #2 \\\\ #3 \\ #4}
  {#5}
}

\ekvcSplit\brsrceload{
  p1={\Vspace(y) = \bloc},
  p2={m\ne\nil},
  p3={\Mspace(\bloc) = \bkappa},
  p4={\mathrm{isptr}(\tau) \implies \kappa\in\region{m}},
  p5={\Vspace'=\Vspace[x\mapsto \higher{\bkappa}{l_1}]},
  c={\bscfgi{i=\Load{\var{x}}{\type{\tau}{l_1}}{\type{\ptr{\tau}}{l_2}\var{y}}} \bsstepi
    \bscfgi{v=\Vspace'}}}{
  \inferrule*[Lab={[E-Load]}]
  {#1 \\ #2 \\ #3 \\\\ #4 \\ #5}
  {#6}
}

\ekvcSplit\brsrcestore{
  p1={\scfge{e=e} \bsstepe \bkappa},
  p2={\Vspace(y) = \bloc},
  p3={\Menv(m) = \ty{t=\tau'*,l=l_2} \implies \kappa \in \region{m}},
  p4={\bloc\ne0},
  p5={\Mspace(\bloc) \neq \varnothing},
  p6={\Mspace' = \Mspace(\bloc\mapsto\higher{\bkappa}{l_2})},
  c={\bscfgi{i=\Store{\type{\tau}{l_1} e}{\type{\ptr{\tau}}{l_2} y}} \bsstepi
    \bscfgi{h=\Hspace', s=\Sspace'}}}{
  \inferrule*[Lab={[E-Store]}]
  {#1 \\ #2 \\ #3 \\\\ #4 \\ #5 \\ #6}
  {#7}
}

\ekvcSplit\brsrcephi{
  p1={\exists i\in [1, n]. b_i = \bprev},
  p2={\bscfge{e=e_i} \bsstepe \bkappa},
  p3={\bVspace' = \bVspace{[x\mapsto\higher{\bkappa}{l}]}},
  c={\bscfgi{i=\PhiNode{\var{x}}{\type{\tau}{l}}{[b_1, e_1] \dots [b_n,e_n]}} \sstepi  
   \bscfgi{v=\bVspace',i=\Nop}}
} {
  \inferrule*[Lab={[E-Phi]}]
  {#1 \\ #2 \\ #3}
  {#4}
}

\ekvcSplit\brsrcegep{
  p1={\Vspace(y) = \bloc},
  p2={\scfge{e=e} \bsstepe \bconst},
  p3={\bloc' =  \bloc + \bconst},
  p4={\bloc' \in \region{\bloc}},
  p5={\Vspace' = \Vspace[x \mapsto \bloc']},
  c={\bscfgi{i=\GEP{\var{x}}{\ty{t=\tau_1,l=l_1}}{\ty{t=\tau_1*,l=l_1}\var{y}}{\ty{t=\tau_2,l=l_2}e}} \bsstepi
    \bscfgi{v=\Vspace'}}}{
  \inferrule*[Lab={[E-GEP]}]
  {#1 \\ #2 \\ #3 \\\\ #4 \\ #5}
  {#6}
}

\ekvcSplit\brsrceinstructionlow{
  p1={\bscfgi{i=i} \bsstepi
    \bscfgi{h=\Hspace', s=\Sspace', v=\Vspace', i=i'}},
  c={\bscfgb{cur=\cur, trg=\varnothing, body=\Seq{i}{body}} \bsstepb
    \bscfgb{h=\Hspace', s=\Sspace', v=\Vspace', cur=\cur, trg=\varnothing, body=\Seq{i'}{body}}}}{
  \inferrule*[Lab={[E-Instruction-Low]}]
  {#1}
  {#2}
}

\ekvcSplit\brsrceinstructionhigh{
  p1={\bscfgi{i=i} \bsstepi
    \bscfgi{h=\Hspace', s=\Sspace', v=\Vspace', i=i'}},
  c={\bscfgb{cur=\fbcur, body=[\Seq{i}{body}]} \bsstepb
    \bscfgb{h=\Hspace', s=\Sspace', v=\Vspace', cur=\fbcur, body=[\Seq{i'}{body}]}}}{
  \inferrule*[Lab={[E-Instruction-High]}]
  {#1}
  {#2}
}

\ekvcSplit\brsrcenoplow{
  p1={~},
  c={\bscfgb{cur=\cur, trg=\varnothing, body=\Seq{\Nop}{body}} \bsstepb
    \bscfgb{cur=\cur, trg=\varnothing, body=body}}}{
  \inferrule*[Lab={[E-Nop-Low]}]
  {#1}
  {#2}
}

\ekvcSplit\brsrcenophigh{
  p1={~},
  c={\bscfgb{cur=\fbcur, body=[\Seq{\Nop}{body}]} \bsstepb
    \bscfgb{cur=\fbcur, body=\fbbody}}}{
  \inferrule*[Lab={[E-Nop-High]}]
  {#1}
  {#2}
}

\ekvcSplit\brsrcebrconditionallow{
  p1={\scfge{e=e} \bsstepe \bdonst},
  p2={\higher{\bdonst}{l}\in\{0,1\}},
  c={\bscfgb{cur=\cur, trg=\varnothing, body=\BrCond{\type{\ione}{l} e}{\var{b_0}}{\var{b_1}}} \bsstepb
    \bscfgb{cur=\cur, trg=\varnothing, body=\Br{\var{b_d}}}}}{
  \inferrule*[Lab={[E-Br-Conditional-Low]}]
  {#1 \\ #2}
  {#3}
}

\ekvcSplit\brsrcebrconditionalraise{
  p1={\scfge{e=e} \bsstepe \bdonst},
  p2={\higher{\bdonst}{l}\in\{[0],[1]\}},
  p3={\trg = \mergepoint{\cur}},
  c={\bscfgb{cur=\cur, trg=\varnothing, body=\BrCond{\type{\ione}{l} e}{\var{b_0}}{\var{b_1}}} \bsstepb
    \bscfgb{cur=\cur, body=[{\Br{\var{b_d}}}]}}}{
  \inferrule*[Lab={[E-Br-Conditional-Raise]}]
  {#1 \\ #2 \\ #3}
  {#4}
}

\ekvcSplit\brsrcebrconditionalhigh{
  p1={\scfge{e=e} \bsstepe \bdonst},
  c={\bscfgb{cur=\fbcur, body=[\BrCond{\type{\ione}{l} e}{\var{b_0}}{\var{b_1}}]} \bsstepb
    \bscfgb{cur=\fbcur, body=[\Br{\var{b_d}}]}}}{
  \inferrule*[Lab={[E-Br-Conditional-High]}]
  {#1}
  {#2}
}

\ekvcSplit\brsrcebrunconditionalhigh{
  p1={\Fspace(\cur) = \FDef{\ty{l=\lf} f}{\paramlist}{\ldots; b: body; \ldots}},
  p2={\fbprev'=\fbcur},
  p3={b \ne \trg},
  c={\bscfgb{cur=\fbcur, body=[\Br{ \var{b}}]} \bsstepb
    \bscfgb{prev=\fbprev', cur=[b], body=\fbbody}}}{
  \inferrule*[Lab={[E-Br-Unconditional-High]}]
  {#1 \\ #2 \\ #3}
  {#4}
}

\ekvcSplit\brsrcebrunconditionallower{
  p1={\Fspace(\cur) = \FDef{\ty{l=\lf} f}{\paramlist}{\ldots; b: body; \ldots}},
  p2={\prev'=\fbcur},
  p3={b = \trg},
  c={\bscfgb{cur=\fbcur, body=[\Br{ \var{b}}]} \bsstepb
    \bscfgb{prev=\prev', cur=b, trg=\varnothing, body=body}}}{
  \inferrule*[Lab={[E-Br-Unconditional-Lower]}]
  {#1 \\ #2 \\ #3}
  {#4}
}

\ekvcSplit\brsrcebrunconditionallow{
  p1={\Fspace(\cur) = \FDef{\ty{l=\lf} f}{\paramlist}{\ldots; b: body; \ldots}},
  p2={\prev' = \cur},
  c={\bscfgb{cur=\cur, trg=\varnothing, body=\Br{ \var{b}}} \bsstepb
    \bscfgb{prev=\prev', cur=b, trg=\varnothing, body=body}}}{
  \inferrule*[Lab={[E-Br-Unconditional-Low]}]
  {#1 \\ #2}
  {#3}
}

\ekvcSplit\brsrcecalllow{
  p1={\Fspace(g) = \type{\tau}{l_g} g(\ty{l=l_1,t=\tau_1}x_1, \dots, \ty{l=l_n,t=\tau_n}x_n) \{b_1:body_1; \cdots;b_m:body_m\}},
  p2={\scfge{e=e_i} \bsstepe \bkappa_i},
  p3={\Vspace' = \{x_i : \higher{\bkappa_i}{l_i}\}},
  p4={\fresh{\Stspace}},
  p5={\Xspace' = \bscfgx{v=\Vspace,cur=\cur,trg=\varnothing,body=\body}::\Xspace},
  c={\bscfgb{cur=\cur, trg=\varnothing, body=\Seq{\FCall{\var{x}}{\type{\tau}{l} f}{\ty{t=\tau_1,l=l'_1}e_1, \dots, \ty{t=\tau_n,l=l'_n}e_n}}{body}} \bsstepb \\
    \bscfgb{s=\Stspace::\Sspace, v=\Vspace',
                   prev=\varnothing, cur=b_1, x=\Xspace', trg=\varnothing, body=body_1}}}{
  \inferrule*[Lab={[E-Call-Low]}]
  {#1 \\ #2 \\ #3 \\ #4 \\\\ #5}
  {#6}
}

\ekvcSplit\brsrcecallhigh{
  p1={\Fspace(g) = \type{\tau}{l_g} g(\ty{l=l_1,t=\tau_1}x_1, \dots, \ty{l=l_n,t=\tau_n}x_n) \{b_1:body_1; \cdots;b_m:body_m\}},
  p2={\scfge{e=e_i} \bsstepe \bkappa_i},
  p3={\fbVspace' = \{x_i : \higher{\bkappa_i}{l_i}\}},
  p4={\fresh{\fbStspace}},
  p5={\Xspace' = \bscfgx{v=\bVspace,cur=\fbcur,body=\fbbody}::\Xspace},
  c={\bscfgb{cur=\fbcur,
                   body=[\Seq{\FCall{\var{x}}
                                    {\type{\tau}{l} f}
                                    {\arglist}}
                             {body}]} \bsstepb  \\
    \bscfgb{s=\fbStspace::\Sspace, v=\fbVspace',
                   prev=\varnothing, cur=[b_1], x=\Xspace', trg=\bigstar, body=[body_1]}}}{
  \inferrule*[Lab={[E-Call-High]}]
  {#1 \\ #2 \\ #3 \\ #4 \\\\ #5}
  {#6}
}

\ekvcSplit\brsrceretlow{
  p1={\scfge{v=\Vspace, e=e} \bsstepe \bkappa},
  p2={\Xspace = \bscfgx{v=\Vspace', cur=\cur', trg=\varnothing, body=\body}::\Xspace'},
  p3={\Vspace''=\Vspace'[x\mapsto\higher{\bkappa}{\lf}]},
  c={\bscfgb{s=\Stspace::\Sspace,trg=\varnothing, body=\Ret{\type{\tau}{\lf} e}} \bsstepb
    \bscfgb{v=\Vspace'', x=\Xspace', prev=\bprev', cur=\cur', trg=\varnothing, body=body}}
} {
  \inferrule*[Lab={[E-Ret-Low]}]
  {#1 \\ #2 \\ #3}
  {#4}
}

\ekvcSplit\brsrcerethigh{
  p1={\scfge{v=\fbVspace, e=e} \bsstepe \bkappa},
  p2={\Xspace = \bscfgx{v=\bVspace',cur=\fbcur', prev=\bprev', trg=\trg', body=\fbbody}::\Xspace'},
  p3={\bVspace''=\bVspace'[x\mapsto\higher{\bkappa}{\lf}]},
  c={\bscfgb{s=\fbStspace::\Sspace, prev=\fbprev, cur=\fbcur, trg=\bigstar, v=\fbVspace, body=[\Ret{\type{\tau}{\lf} e}]} \bsstepb
    \bscfgb{v=\bVspace'', cur=\fbcur', x=\Xspace', prev=\bprev', trg=\trg', body=\fbbody}}}{
  \inferrule*[Lab={[E-Ret-High]}]
  {#1 \\ #2 \\ #3}
  {#4}
}

\ekvcSplit\brsrceprogram{
  p1={\forall i.\ func_i = \FDef{\type{\tau}{l} f_i}{\paramlist}{b_1 : body_1; \cdots; b_n : body_n}},
  p2={\Fspace_f=\{f_i :func_i \mid \forall i\}},
  p3={\Fspace_b= \{b : func_i \mid \forall i,\ \forall b \in func_i\}},
  p4={\Fspace = \Fspace_f \cup \Fspace_b},
  p5={\forall i.\ \func_i = \FDef{\ty{}f_i}{\paramlist}{\overline{block}}},
  p6={\func_0 = \FDef{\type{\itt}{\bot} \main}{}{\overline{block}_{\main}}},
  p7={\bscfgb{s=\varnothing, x=\varnothing, prev=\varnothing, cur=\varnothing, trg=\varnothing, body=\Seq{\FCall{\var{x}}{\type{\itt}{\public} \main}{}}{\Ret{\var{x}}}}
    \bstrstepb  
    \bscfgb{h=\Hspace', s=\varnothing, v=\Vspace', x=\varnothing, prev=\varnothing, cur=\varnothing, trg=\varnothing, body=\Ret{c}}},
  c={\langle H, V, func_0, \dots, func_N \rangle \bsstepp \configtwo{\Hspace'}{c}}}{
  \inferrule*[Lab={[E-Program]}]
  {#5 \\ #6 \\ #7}
  {#8}
}

\ekvcSplit\brsrceprogramtrace{
  p1={\forall i.\ func_i = \FDef{\type{\tau}{l} f_i}{\paramlist}{b_1 : body_1; \cdots; b_n : body_n}},
  p2={\Fspace_f=\{f_i :func_i \mid \forall i\}},
  p3={\Fspace_b= \{b : func_i \mid \forall i,\ \forall b \in func_i\}},
  p4={\Fspace = \Fspace_f \cup \Fspace_b},
  p5={\forall i.\ \func_i = \ty{}f_i(\paramlist)\{\overline{block}\}},
  p6={\func_0 = \ty{t=\itt,l=\bot}\main()\{\overline{block}_\main\}},
  p7={\scfgb{s=\varnothing, x=\varnothing, prev=\varnothing, cur=\varnothing, body=\Seq{\FCall{\var{x}}{\type{\itt}{\public} \main}{}}{\Ret{\var{x}}}}
    \bsstepbtr  
    \scfgb{h=\Hspace', s=\varnothing, v=\Vspace', x=\varnothing, prev=\varnothing, cur=\varnothing, body=\Ret{x}} \raises T},
  c={\langle H, V, \func_0, \dots, \func_N \rangle \bsstepptr  
    \configtwo{\Hspace'}{\Vspace'(x)} \raises T}}{
  \inferrule*[Lab={[E-Program-Trace]}]
  {#5 \\ #6 \\ #7}
  {#8}
}

\ekvcSplit\brsrcstrefl{
  p1={~},
  c={\type{\tau}{l} \flowsto \type{\tau}{l}}}{
  \inferrule*[Lab={[ST-Refl]}]
  {#1}
  {#2}
}

\ekvcSplit\brsrcstint{
  p1={\tau\in\{\ione, \itt\}},
  p2={l_1 \flowsto l_2},
  c={\type{\tau}{l_1} \flowsto \type{\tau}{l_2}}}{
  \inferrule*[Lab={[ST-Int]}]
  {#1 \\ #2}
  {#3}
}

\ekvcSplit\brsrctflowsa{
  p1={l_1 \flowsto l_2},
  c={\declat{l_1}{b_1} \flowsto_{b_2} l_2}}{
  \inferrule*[Lab={[T-Flows-1]}]
  {#1}
  {#2}
}

\ekvcSplit\brsrctflowsb{
  p1={~},
  c={\declat{l}{b} \flowsto_{b} \bot}}{
  \inferrule*[Lab={[T-Flows-2]}]
  {#1}
  {#2}
}

\ekvcSplit\brsrctflowsc{
  p1={l_1 \flowsto l_2},
  c={\declat{l_1}{b} \flowsto l_2}}{
  \inferrule*[Lab={[T-Flows-3]}]
  {#1}
  {#2}
}

\ekvcSplit\brsrctnullhigh{
  p1={l\nflowsto\bot},
  c={\tc [\nil] : \ty{t=\tau*}}
}
{
  \inferrule*[Lab={[T-Null-High]}]
  {#1}
  {#2}
}

\ekvcSplit\brsrctnulllow{
  c={\tc \nil : \ty{t=\tau*,l=\bot}}
}
{
  \inferrule*[Lab={[T-Null-Low]}]
  {~}
  {#1}
}

\ekvcSplit\brsrctvar{
  p1={\TypeEnv(x) = \type{\tau}{l}},
  c={\stctxe{} \tc  \var{x} : \type{\tau}{l}}}{
  \inferrule*[Lab={[T-Var]}]
  {#1}
  {#2}
}

\ekvcSplit\brsrctsub{
  p1={\stctxe{} \tc e : \type{\tau}{l_1}},
  p2={\type{\tau}{l_1} \flowsto \type{\tau}{l_2}},
  c={\stctxe{} \tc  e : \type{\tau}{l_2}}}{
  \inferrule*[Lab={[T-Sub]}]
  {#1 \\ #2}
  {#3}
}

\ekvcSplit\brsrctnop{
  p1={~},
  c={\stctxb{} \tc{\Nop}}}{
  \inferrule*[Lab={[T-Nop]}]
  {#1}
  {#2}
}

\ekvcSplit\brsrctasgn{
  p1={\TypeEnv(x) = \type{\tau}{l}},
  p2={\TypeEnv \tc e : \type{\tau}{l}},
  p3={\pc \flowsto l},
  c={\stctxb{} \tc{\Asgn{\var{x}}{\type{\tau}{l}e}}}}{
  \inferrule*[Lab={[T-Asgn]}]
  {#1 \\ #2 \\ #3}
  {#4}
}

\ekvcSplit\brsrctadd{
  p1={\TypeEnv(x) = \type{\tau}{l}},
  p2={\forall i.\ \TypeEnv \vdash e_i : \type{\tau}{l}},
  p3={\pc\flowsto l},
  c={\stctxb{} \tc{\Add{\var{x}}{\type{\tau}{l}e_1}{e_2}}}}{
  \inferrule*[Lab={[T-Add]}]
  {#1 \\ #2 \\ #3}
  {#4}
}

\ekvcSplit\brsrctalloca{
  p1={\TypeEnv(x) = \type{\tau*}{l}},
  p2={\pc \flowsto l},
  c={\stctxb{} \tc{\Alloca{\var{x}}{\type{\tau}{l}}}}}{
  \inferrule*[Lab={[T-Alloca]}]
  {#1 \\ #2}
  {#3}
}

\ekvcSplit\brsrctmalloc{
  p1={\TypeEnv(x) = \type{\tau*}{l}},
  p2={\pc \flowsto l},
  c={\stctxb{} \tc{\Malloc{\var{x}}{\type{\tau}{l}}}}}{
  \inferrule*[Lab={[T-Malloc]}]
  {#1 \\ #2}
  {#3}
}

\ekvcSplit\brsrctfree{
  p1={\TypeEnv(y) = \type{\tau*}{l}},
  p2={\pc \flowsto l},
  c={\stctxb{} \tc{\Free{\type{\tau*}{l}}{\var{y}}}}}{
  \inferrule*[Lab={[T-Free]}]
  {#1 \\ #2}
  {#3}
}

\ekvcSplit\brsrctstore{
  p1={\TypeEnv \tc e : \type{\tau}{l_1}},
  p2={\TypeEnv \tc \var{y} : \type{\tau*}{l_2}},
  p3={l_1 \flowsto l_2},
  p4={\pc \flowsto l_2},
  c={\stctxb{} \tc
                    {\Store{\type{\tau}{l_1}e}{\type{\tau*}{l_2}\var{y}}}}}{
  \inferrule*[Lab={[T-Store]}]
  {#1 \\ #2 \\\\ #3 \\ #4}
  {#5}
}

\ekvcSplit\brsrctload{
  p1={\TypeEnv(x)=\type{\tau}{l_1}},
  p2={\TypeEnv \tc \var{y} : \type{\tau*}{l_2}},
  p3={\type{\tau}{l_2} \flowsto \type{\tau}{l_1}},
  p4={\pc \flowsto l_1},
  c={\stctxb{} \tc
                    {\Load{\var{x}}{\type{\tau}{l_1}}{\type{\tau*}{l_2}\var{y}}}}}{
  \inferrule*[Lab={[T-Load]}]
  {#1 \\ #2 \\\\ #3 \\ #4}
  {#5}
}

\ekvcSplit\brsrctphi{
  p1={\TypeEnv(x) = \type{\tau}{l}},
  p2={\forall i.\ \TypeEnv \tc e_i : \type{\tau}{l} },
  p3={\forall i.\ \Benv(b_i) \flowsto l},
  p4={\pc\flowsto l},
  c={\stctxb{} \tc
                    {\PhiNode{\var{x}}{\type{\tau}{l}}{[b_1, e_1] \dots [b_n,e_n]}}}}{
  \inferrule*[Lab={[T-Phi]}]
  {#1 \\ #2 \\\\ #3 \\ #4}
  {#5}
}

\ekvcSplit\brsrctgep{
  p1={\TypeEnv(x) = \TypeEnv(y) = \type{\tau_1*}{l_1}},
  p2={\TypeEnv \tc e : \type{\tau_2}{l_2}},
  p3={l_2 \flowsto l_1},
  p4={\tau_2 \in \{\itt, \ione\}},
  p5={\pc\flowsto l_1},
  c={\stctxb{} \tc
	                  {\GEP{\var{x}}
                         {\type{\tau_1}{l_1}}
                         {\type{\tau_1*}{l_1}\var{y}}
                         {\type{\tau_2}{l_2} e}}}}{
  \inferrule*[Lab={[T-GEP]}]
  {#1 \\ #2 \\ #3 \\ #4 \\ #5}
  {#6}
}

\ekvcSplit\brsrctbrunconditional{
  p1={\pc \flowsto_{b} \Benv(b)},
  c={\stctxb{} \tc  \Br{ \var{b}}}
} {
  \inferrule*[Lab={[T-Br-Unconditional]}]
  {#1}
  {#2}
}

\ekvcSplit\brsrctbrconditional{
  p1={\stctxe{} \tc e : \type{\ione}{l}},
  p2={b_t = \mergepoint{\cur}},
  p3={\forall i.\ \pc \join \declat{l}{b_t} \flowsto_{b_i} \Benv(b_i)},
  c={\stctxb{} \tc  \BrCond{\type{\ione}{l}e}{ \var{b_0}}{ \var{b_1}}}
} {
  \inferrule*[Lab={[T-Br-Conditional]}]
  {#1 \\ #2 \\\\ #3}
  {#4}
}

\ekvcSplit\brsrctbody{
  p1={\stctxb{} \tc i},
  p2={\stctxb{} \tc body},
  c={\stctxb{} \tc  \Seq{i}{body}}
} {
  \inferrule*[Lab={[T-Body]}]
  {#1 \\\\ #2}
  {#3}
}

\ekvcSplit\brsrctbodyhigh{
  p1={\pc \flowsto \pc'},
  p2={\pc' \nflowsto \bot},
  p3={\stctxb{pc=\pc'} \tc body},
  c={\stctxb{} \tc  \fbbody}
} {
  \inferrule*[Lab={[T-Body-High]}]
  {#1 \\ #2 \\\\ #3}
  {#4 \\\\ \sam{maybe nuke this}}
}

\ekvcSplit\brsrctcall{
  p1={\Fenv(g) = \FHead{\type{\tau}{\lg} \fname{g}}{\type{\tau_1}{l_1}x_1, \type{\tau_2}{l_2}x_2, \dots, \type{\tau_n}{l_n}x_n}{}\{\cdots\}},
  p2={\stctxf{} \tc g},
  p3={\TypeEnv(x) = \type{\tau}{l_g}},
  p4={\pc \flowsto l_g},
  p5={\forall i.\ \TypeEnv \tc e_i : \type{\tau_i}{l_i}},
  c={\stctxb{} \tc  \FCall{\var{x}}
                         {\type{\tau}{l_g} \fname{g}}
                         {
                           \type{\tau_1}{l_1'}e_1,
                           \type{\tau_2}{l_2'}e_2,
                           \dots,
                           \type{\tau_n}{l_n'}e_n
                         }}}{
  \inferrule*[Lab={[T-Call]}]
  {#1 \\ #2 \\\\ #3 \\ #4 \\ #5}
  {#6}
}

\ekvcSplit\brsrctret{
  p1={\TypeEnv \tc e : \type{\tau}{\lf}},
  c={\stctxb{} \tc{\Ret{\type{\tau}{\lf} e}}}}{
  \inferrule*[Lab={[T-Ret]}]
  {#1}
  {#2}
}

\ekvcSplit\brsrctfunction{
  p1={\Fenv(f) = \FHead{\type{\tau}{\lf} \fname{f}}{\type{\tau_1}{l_1}x_1, \type{\tau_2}{l_2}x_2, \dots, \type{\tau_n}{l_n}x_n}{}\{b_1:\body_1; b_2:\body_2; \cdots\}},
  p2={\forall i.\ \pc_i = \Benv_f(b_i)},
  p3={\forall i.\ \lf \flowsto pc_i},
  p4={\forall i.\ \stctxb{pc=\pc_i} \tc \body_i},
  c={\stctxf{} \tc f}
} {
  \inferrule*[Lab={[T-Function]}]
  {#1 \\ #2 \\ #3 \\ #4}
  {#5}
}

\ekvcSplit\brsrctprogram{
  p1={\Fspace(f_0)=\FDef{\type{\itt}{\public}\main}{\arglist}{block^+}},
  p2={\Gamma_i = \Gamma\mid_i},
  p3={\forall i.\ \stctxf{g=\Gamma_{i}} \tc f_i},
  c={ \stctxp{} \tc  f_0, f_1, \dots,f_n}}{
  \inferrule*[Lab={[T-Program]}]
  {#1 \\ #2 \\ #3}
  {#4}
}

\ekvcSplit\brsrctmemloclow{
  p1={\Menv(m) = \type{\tau'}{\bot}},
  c={\tc m : \type{\tau*}{\bot}}
} {
  \inferrule*[Lab={[T-Mem-Loc-Low]}]
  {#1}
  {#2}
}

\ekvcSplit\brsrctmemlochigh{
  p1={\Menv(\fbloc) = \type{\tau'}{l}},
  p2={l \nflowsto \bot},
  c={\tc \fbloc : \type{\tau*}{l}}
} {
  \inferrule*[Lab={[T-Mem-Loc-High]}]
  {#1 \\ #2}
  {#3}
}

\ekvcSplit\brsrctmemoryconfig{
  p1={\forall m\in\Loc.\
        \Mspace(m) = \varnothing
        \lor \Menv \tc \Mspace(m) : \Menv(m)
      },
  c={\tc  \configtwo{\Sspace}{\Hspace}}
} {
  \inferrule*[Lab={[T-Memory-Config]}]
  {#1}
  {#2}
}

\ekvcSplit\brsrctvariables{
  p1={\forall x\in\Var.\ \Menv \tc \Vspace(x) : \TypeEnv(x)},
  c={\stctxe{} \tc \bVspace}
} {
  \inferrule*[Lab={[T-Variables]}]
  {#1}
  {#2}
}

\ekvcSplit\brsrctfunctions{
  p1={\forall f\in \Fcn.\ \tc f},
  c={ \tc  \Fspace}
} {
  \inferrule*[Lab={[T-Functions]}]
  {#1}
  {#2}
}

\ekvcSplit\brsrctcallstack{
  p1={\Fenv(\cur) = \FHead{\type{\tau}{\lh} \fname{h}}{\type{\tau_1}{l_1}x_1, \type{\tau_2}{l_2}x_2, \dots, \type{\tau_n}{l_n}x_n}{}},
  p2={\pc = \Benv_h(\cur)},
  p3={\TypeEnv_h \tc \bVspace},
  p4={\TypeEnv_h(x) = \type{\tau}{\lf}},
  p5={\isub{\bbody} \implies \TypeEnv_h, \pc, \lh \tc \body},
  p6={l \nflowsto \bot \\ \declat{l}{\trg} \flowsto \pc' \\ \pc \flowsto \pc' \\\\ \isbr{\bbody} \implies (\stctxb{g=\TypeEnv_h, pc=\pc', lf=\lh} \tc \body) \land (\pc' \flowsto \lf)},
  p7={\lh \flowsto \pc \flowsto \lf},
  p8={\trg, \lh \tc \Xspace},
  c={\trg, \lf \tc  \scfgx{body=\bbody}::\Xspace}
} {
  \inferrule*[Lab={[T-Callstack]}]
  {#1 \\\\ #2 \\ #3 \\ #4 \\\\ #5 \\\\ #6 \\\\ #7 \\ #8}
  {#9}
}

\ekvcSplit\brsrctcallstacklow{
  p1={\Fenv(\cur) = \FHead{\type{\tau}{\lh} \fname{h}}{\type{\tau_1}{l_1}x_1, \type{\tau_2}{l_2}x_2, \dots, \type{\tau_n}{l_n}x_n}{}},
  p2={\pc = \Benv_h(\cur)},
  p3={\TypeEnv_h \tc \Vspace},
  p4={\TypeEnv_h(x) = \type{\tau}{\lf}},
  p5={\TypeEnv_h, \pc, \lh \tc \body},
  p7={\lh \flowsto \pc \flowsto \lf},
  p8={\lh \tc \Xspace},
  c={\lf \tc  \bscfgx{cur=\cur,trg=\varnothing,body=\body}::\Xspace}
} {
  \inferrule*[Lab={[T-Callstack-Low]}]
  {#1 \\\\ #2 \\ #3 \\ #4 \\\\ #5 \\ #6 \\ #7}
  {#8}
}

\ekvcSplit\brsrctcallstackhigh{
  p1={\Fenv(\cur) = \FHead{\type{\tau}{\lh} \fname{h}}{\type{\tau_1}{l_1}x_1, \type{\tau_2}{l_2}x_2, \dots, \type{\tau_n}{l_n}x_n}{}},
  p2={\pc = \Benv_h(\cur)},
  p3={\TypeEnv_h \tc \bVspace},
  p4={\TypeEnv_h(x) = \type{\tau}{\lf}},
  p5={l \nflowsto \bot \\ \declat{l}{\trg} \flowsto \pc' \\ \stctxb{g=\TypeEnv_h, pc=\pc', lf=\lh} \tc \body},
  p6={\lh \flowsto \pc \flowsto \pc' \flowsto \lf},
  p7={\lh \tc \Xspace},
  c={\lf \tc  \bscfgx{body=\fbbody}::\Xspace}
} {
  \inferrule*[Lab={[T-Callstack-High]}]
  {#1 \\\\ #2 \\ #3 \\ #4 \\\\ #5 \\\\ #6 \\ #7}
  {#8}
}

\ekvcSplit\brsrctcallstackbase{
  p1={\Xspace = \varnothing},
  c={\Menv, \lf \tc  \Xspace}
} {
  \inferrule*[Lab={[T-Callstack-Base]}]
  {#1}
  {#2}
}

\ekvcSplit\brsrctexpressionconfig{
  p1={\stctxe{} \tc \bVspace},
  p2={\stctxe{} \tc e : \ty{}},
  c={\stctxe{} \tc  \configtwo{\bVspace}{e}}
} {
  \inferrule*[Lab={[T-Expression-Config]}]
  {#1 \\ #2}
  {#3}
}

\ekvcSplit\brsrctinstructionconfig{
  p1={\tc \configtwo{\Sspace}{\Hspace}},
  p2={\stctxe{} \tc \bVspace},
  p3={\stctxb{} \tc{i}},
  c={\stctxi{} \tc  \bscfgi{i=i}}
} {
  \inferrule*[Lab={[T-Instruction-Config]}]
  {#1 \\ #2 \\ #3}
  {#4}
}

\ekvcSplit\brsrctbodyconfiglow{
  p1={\tc \configtwo{\Sspace}{\Hspace}},
  p2={\stctxe{} \tc \Vspace},
  p3={\lf \tc \Xspace},
  p4={\stctxb{} \tc \body},
  p5={ \tc \Fspace},
  c={\typingcontextconfig \tc  \bscfgb{v=\Vspace,cur=\cur, trg=\varnothing, body=body}}
} {
  \inferrule*[Lab={[T-Body-Config-Low]}]
  {#1 \\ #2 \\ #3 \\\\ #4 \\ #5}
  {#6}
}

\ekvcSplit\brsrctbodyconfighigh{
  p1={\tc \configtwo{\Sspace}{\Hspace}},
  p2={\stctxe{} \tc \Vspace},
  p3={\lf \tc \Xspace},
  p4={l \nflowsto \bot},
  p5={\declat{l}{\trg} \flowsto \pc'},
  p6={\pc\flowsto\pc'},
  p7={\stctxb{pc=\pc',cur=\bcur} \tc{\body}},
  p8={\tc \Fspace},
  c={\stctxb{pc=\pc,cur=\bcur} \tc  \bscfgb{cur=\bcur,body=\fbbody}}
} {
  \inferrule*[Lab={[T-Body-Config-High]}]
  {#1 \\ #2 \\ #3 \\\\ #4 \\ #5 \\ #6 \\\\ #7 \\ #8}
  {#9}
}

\ekvcSplit\brsrctresultconfig{
  p1={\Menv \tc \configtwo{\varnothing}{\Hspace}},
  p2={ \tc c : \type{\itt}{\public}},
  c={\Menv \tc  \configtwo{\Hspace}{c}}
} {
  \inferrule*[Lab={[T-Result-Config]}]
  {#1 \\\\ #2}
  {#3}
}

\ekvcSplit\brsrctprogramconfig{
  p1={ \tc p},
  p2={\Menv \tc \Hspace},
  p3={\stctxe{} \tc \Vspace},
  c={\Menv, \TypeEnv \tc  \scfgp{}}}{
  \inferrule*[Lab={[T-Program-Config]}]
  {#1 \\ #2 \\ #3}
  {#4}
}

\ekvcSplit\brsrceqvconst{
  p1={\kappa_1 = \kappa_2 \lor (\isbr{\bkappa_1} \land \isbr{\bkappa_2})},
  c={\bkappa_1 \lequiv \bkappa_2}}{
  \inferrule*[Lab={[Eqv-Const]}]
  {#1}
  {#2}
}

\ekvcSplit\brsrceqvvar{
  c={x\lequiv x}
} {
  \inferrule*[Lab={[Eqv-Var]}]
  {~}
  {#1}
}

\ekvcSplit\brsrceqvblocklabels{
  p1={\blab_1 = \blab_2 \lor (\isbr{\blab_1} \land \isbr{\blab_2})},
  c={\blab_1 \lequiv \blab_2}}{
  \inferrule*[Lab={[Eqv-Block-Labels]}]
  {#1}
  {#2}
}

\ekvcSplit\brsrceqvmemories{
  p1={\bkappa_1 \lequiv \bkappa_2 \\\\
      \lor\ (\bkappa_1 = \varnothing \land \isbr{\bkappa_2})\\\\
      \lor\ (\isbr{\bkappa_1} \land\ \bkappa_2 = \varnothing)},
  c={\bkappa_1 \mequiv \bkappa_2}}{
  \inferrule*[Lab={[Eqv-Memories]}]
  {#1}
  {#2}
}

\ekvcSplit\brsrceqvbodylow{
  p1={\bbody_1 = \Seq{i_1}{\body_{1,t}} \\},
  p2={\bbody_2 = \Seq{i_2}{\body_{2,t}} \\\\},
  p3={i_1 \lequiv i_2 \\},
  p4={\body_{1,t} = \body_{2,t}},
  c={\bbody_1 \lequiv \bbody_2}
} {
  \inferrule*[Lab={[Eqv-Body-Low]}]
  {#1 #2 #3 #4}
  {#5}
}

\ekvcSplit\brsrceqvbodyhigh{
  p1={~},
  c={\fbbody_1 \lequiv \fbbody_2}
} {
  \inferrule*[Lab={[Eqv-Body-High]}]
  {#1}
  {#2}
}

\ekvcSplit\brsrceqvasgn{
  p1={e_1 \lequiv e_2},
  c={\biglequiv{\Asgn{\var{x}}{\type{\tau}{l} e_1}}
               {\Asgn{\var{x}}{\type{\tau}{l} e_2}}}
} {
  \inferrule*[Lab={[Eqv-Asgn]}]
  {#1}
  {#2}
}

\ekvcSplit\brsrceqvadd{
  p1={e_1 \lequiv e_2},
  p2={e_3 \lequiv e_4},
  c={\biglequiv{\Add{\var{x}}{\type{\tau}{l} e_1}{e_3}}
               {\Add{\var{x}}{\type{\tau}{l} e_2}{e_4}}}
} {
  \inferrule*[Lab={[Eqv-Add]}]
  {#1 \\ #2}
  {#3}
}

\ekvcSplit\brsrceqvload{
  p1={y_1 \lequiv y_2},
  c={\biglequiv{\Load{\var{x}}{\type{\tau}{l_1}}{\type{\ptr{\tau}}{l_2}y_1}}
               {\Load{\var{x}}{\type{\tau}{l_1}}{\type{\ptr{\tau}}{l_2}y_2}}}
} {
  \inferrule*[Lab={[Eqv-Load]}]
  {#1}
  {#2}
}

\ekvcSplit\brsrceqvstore{
  p1={e_1 \lequiv e_2 \\ x_1\lequiv x_2},
  c={\biglequiv{\Store{\type{\tau}{l_1} e_1}{\type{\ptr{\tau}}{l_2} x_1}}
               {\Store{\type{\tau}{l_1} e_2}{\type{\ptr{\tau}}{l_2} x_2}}}}{
  \inferrule*[Lab={[Eqv-Store]}]
  {#1}
  {#2}
}

\ekvcSplit\brsrceqvmalloc{
  p1={~},
  c={\biglequiv{\Malloc{\var{x}}{\type{\tau}{l}}}{\Malloc{\var{x}}{\type{\tau}{l}}}}
} {
  \inferrule*[Lab={[Eqv-Malloc]}]
  {#1}
  {#2}
}

\ekvcSplit\brsrceqvalloca{
  p1={~},
  c={\biglequiv{\Alloca{\var{x}}{\type{\tau}{l}}}{\Alloca{\var{x}}{\type{\tau}{l}}}}
} {
  \inferrule*[Lab={[Eqv-Alloca]}]
  {#1}
  {#2}
}

\ekvcSplit\brsrceqvfree{
  p1={y_1\lequiv y_2},
  c={\biglequiv{\Free{\type{\tau*}{l}y_1}}{\Free{\type{\tau*}{l}y_2}}}
} {
  \inferrule*[Lab={[Eqv-Free]}]
  {#1}
  {#2}
}

\ekvcSplit\brsrceqvphi{
  p1={\forall i.\ \blab_{1,i}\lequiv\blab_{2,i} \land e_{1,i} \lequiv e_{2,1}},
  c={\biglequiv{\PhiNode{\var{x}}{\type{\tau}{l}}{[b_{1,1}, e_{1,1}] \dots [b_{1,n},e_{1,n}]}}
               {\PhiNode{\var{x}}{\type{\tau}{l}}{[b_{2,1}, e_{2,1}] \dots [b_{2,n},e_{2,n}]}}}
} {
  \inferrule*[Lab={[Eqv-Phi]}]
  {#1}
  {#2}
}

\ekvcSplit\brsrceqvgep{
  p1={y_1 \lequiv y_2},
  p2={e_1 \lequiv e_2},
  c={\biglequiv{\GEP{\var{x}}{\ty{t=\tau_1,l=l_1}}{\ty{t=\tau*_1,l=l_1}y_1}{\ty{t=\tau_2,l=l_2}e_1}}
               {\GEP{\var{x}}{\ty{t=\tau_1,l=l_1}}{\ty{t=\tau*_1,l=l_1}y_2}{\ty{t=\tau_2,l=l_2}e_2}}}
} {
  \inferrule*[Lab={[Eqv-GEP]}]
  {#1 \\ #2}
  {#3}
}

\ekvcSplit\brsrceqvnop{
  p1={~},
  c={\Nop\lequiv\Nop}
} {
  \inferrule*[Lab={[Eqv-Nop]}]
  {#1}
  {#2}
}

\ekvcSplit\brsrceqvcall{
  p1={\forall i.\ e_{1,i} \lequiv e_{2,i}},
  c={\biglequiv{\FCall{\var{x}}{{\ty{l=\lg}} g}{\ty{t=\tau_1,l=1_1}e_{1,1}, \dots, \ty{t=\tau_n,l=l_n}e_{1,n}}}
               {\FCall{\var{x}}{{\ty{l=\lg}} g}{\ty{t=\tau_1,l=1_1}e_{2,1}, \dots, \ty{t=\tau_n,l=l_n}e_{1,n}}}}
} {
  \inferrule*[Lab={[Eqv-Call]}]
  {#1}
  {#2}
}

\ekvcSplit\brsrceqvret{
  p1={e_1 \lequiv e_2},
  c={\biglequiv{\Ret{\type{\tau}{l_f} e_1}}{\Ret{\type{\tau}{l_f} e_2}}}
} {
  \inferrule*[Lab={[Eqv-Ret]}]
  {#1}
  {#2}
}

\ekvcSplit\brsrceqvbrconditional{
  p1={e_1 \lequiv e_2},
  c={\biglequiv{\BrCond{\type{\ione}{l} e}{\var{b_0}}{\var{b_1}}}
               {\BrCond{\type{\ione}{l} e}{\var{b_0}}{\var{b_1}}}}
} {
  \inferrule*[Lab={[Eqv-Br-Conditional]}]
  {#1}
  {#2}
}

\ekvcSplit\brsrceqvbrunconditional{
  p1={~},
  c={\Br{\var{b}}\lequiv\Br{\var{b}}}
} {
  \inferrule*[Lab={[Eqv-Br-Unconditional]}]
  {#1}
  {#2}
}

\ekvcSplit\brsrceqvheap{
  p1={\forall m\in \Loc.\ \Hspace_1(m) \mequiv \Hspace_2(m)},
  c={\Hspace_1 \lequiv \Hspace_2}}{
  \inferrule*[Lab={[Eqv-Heap]}]
  {#1}
  {#2}
}

\ekvcSplit\brsrceqvstacklow{
  p1={\forall m\in \Loc.\ \Stspace_1(m) \mequiv \Stspace_2(m)},
  p2={\Sigma_1 \lequiv \Sigma_2},
  c={\Stspace_1::\Sspace_1 \lequiv \Stspace_2::\Sspace_2}}{
  \inferrule*[Lab={[Eqv-Stack-Low]}]
  {#1 \\ #2}
  {#3}
}

\ekvcSplit\brsrceqvstackhigh{
  p1={\Sspace_1 \lequiv \bStspace_2::\Sspace_2},
  c={\fbStspace_1::\Sspace_1 \lequiv \bStspace_2::\Sspace_2}}{
  \inferrule*[Lab={[Eqv-Stack-High]}]
  {#1}
  {#2}
}

\ekvcSplit\brsrceqvvariables{
  p1={\forall x\in\Var.\ \Vspace_1(x) \mequiv \Vspace_2(x)},
  c={\Vspace_1 \lequiv \Vspace_2}}{
  \inferrule*[Lab={[Eqv-Variables]}]
  {#1}
  {#2}
}

\ekvcSplit\brsrceqvcallstacklow{
  p1={\Xspace_1=\bscfgx{v=\Vspace'_1,prev=\bprev_1,cur=\cur,trg=\varnothing,r=x,body=\body}::\Xspace_{1,t}},
  p2={\Xspace_2=\bscfgx{v=\Vspace'_2,prev=\bprev_2,cur=\cur,trg=\varnothing,r=x,body=\body}::\Xspace_{2,t}},
  p3={\Vspace_1\lequiv\Vspace_2},
  p4={\bprev_1\lequiv\bprev_2},
  p5={\scfgvx{v=\Vspace'_1,x=\Xspace_{1,t}}\lequiv\scfgvx{v=\Vspace'_2,x=\Xspace_{2,t}}},
  c={\scfgvx{x=\Xspace_1,v=\Vspace_1}\lequiv\scfgvx{x=\Xspace_2,v=\Vspace_2}}
} {
  \inferrule*[Lab={[Eqv-Callstack-Low]}]
  {#1 \\ #2 \\\\ #3 \\ #4 \\ #5}
  {#6}
}

\ekvcSplit\brsrceqvcallstackhigh{
  p1={\Xspace_1=\bscfgx{v=\bVspace'_1,prev=\bprev_1,cur=\fbcur,r=x,body=\fbbody}::\Xspace_{1,t}},
  p2={\scfgvx{v=\bVspace'_1,x=\Xspace_{1,t}}\lequiv\scfgvx{v=\bVspace_2,x=\Xspace_2}},
  c={\scfgvx{x=\Xspace_1,v=\fbVspace_1}\lequiv\scfgvx{x=\Xspace_2,v=\bVspace_2}}
} {
  \inferrule*[Lab={[Eqv-Callstack-High]}]
  {#1 \\ #2}
  {#3}
}

\ekvcSplit\brsrceqvcallstackbase{
  p1={\Vspace_1\lequiv\Vspace_2},
  c={\scfgvx{x=\varnothing,v=\Vspace_1}\lequiv\scfgvx{x=\varnothing,v=\Vspace_2}}
} {
  \inferrule*[Lab={[Eqv-Callstack-Base]}]
  {#1}
  {#2}
}

\ekvcSplit\brsrceqvexpressionconfig{
  p1={\Vspace_1 \lequiv \Vspace_2},
  p2={e_1 \lequiv e_2},
  c={\bscfge{v=\Vspace_1, e=e_1} \lequiv \bscfge{v=\Vspace_2, e=e_2}}
} {
  \inferrule*[Lab={[Eqv-Expression-Config]}]
  {#1 \\ #2}
  {#3}
}

\ekvcSplit\brsrceqvinstructionconfig{
  p1={\Hspace_1 \lequiv \Hspace_2},
  p2={\Sspace_1 \lequiv \Sspace_2},
  p3={\Vspace_1 \lequiv \Vspace_2},
  p4={\bprev_1 \lequiv \bprev_2},
  p5={i_1 \lequiv i_2},
  c={\bscfgi{h=\Hspace_1, s=\Sspace_1, v=\Vspace_1, prev=\bprev_1, i=i_1} \lequiv
     \bscfgi{h=\Hspace_2, s=\Sspace_2, v=\Vspace_2, prev=\bprev_2, i=i_2}}
}
{
  \inferrule*[Lab={[Eqv-Instruction-Config]}]
  {#1 \\ #2 \\ #3 \\\\ #4 \\ #5}
  {#6}
}

\ekvcSplit\brsrceqvbodyconfiglow{
  p1={\Hspace_1 \lequiv \Hspace_2},
  p2={\Sspace_1 \lequiv \Sspace_2},
  p3={\bscfgvx{v=\Vspace_1,x=\Xspace_1} \lequiv \bscfgvx{v=\Vspace_2,x=\Xspace_2}},
  p4={\bprev_1 \lequiv \bprev_2},
  p5={body_1 \lequiv body_2},
  p6={\isub{body_1}},
  c={\biglequiv{\bscfgb{h=\Hspace_1, s=\Sspace_1, v=\Vspace_1, x=\Xspace_1, prev=\bprev_1, cur=\cur, trg=\varnothing, body=body_1}}
               {\bscfgb{h=\Hspace_2, s=\Sspace_2, v=\Vspace_2, x=\Xspace_2, prev=\bprev_2, cur=\cur, trg=\varnothing, body=body_2}}}
} {
  \inferrule*[Lab={[Eqv-Body-Config-Low]}]
  {#1 \\ #2 \\ #3 \\ #4 \\\\ #5 \\ #6}
  {#7}
}

\ekvcSplit\brsrceqvbodyconfighigh{
  p1={\Hspace_1 \lequiv \Hspace_2},
  p2={\Sspace_1 \lequiv \Sspace_2},
  p3={\bscfgvx{v=\bVspace_1,x=\Xspace_1} \lequiv \bscfgvx{v=\bVspace_2,x=\Xspace_2}},
  c={\biglequiv{\bscfgb{h=\Hspace_1, s=\Sspace_1, v=\Vspace_1, x=\Xspace_1, prev=\bprev_1, cur=[\cur_1], body=[body_1]}}
               {\bscfgb{h=\Hspace_2, s=\Sspace_2, v=\Vspace_2, x=\Xspace_2, prev=\bprev_2, cur=[\cur_2], body=[body_2]}}}
} {
  \inferrule*[Lab={[Eqv-Body-Config-High]}]
  {#1 \\ #2 \\ #3}
  {#4}
}

\ekvcSplit\brsrctracenostep{
  c={\scfgb{body=\body} \bsstepbtr \scfgb{body=\body} \raises \configtwo{\Hspace}{\Sspace}}
} {
  \inferrule*[Lab=Trace-NoStep]
  {~}
  {#1}
}

\ekvcSplit\brsrctracenextstep{
  p1={\bscfgbn{body=\body, pow=0} \bsstepb \bscfgbn{body=\body, pow=1}},
  p2={\bscfgbn{body=\body, pow=1} \bsstepbtr \bscfgbn{body=\body, pow=n} \raises \eventpow{pow=1} \cdots \eventpow{pow=n}},
  c={\bscfgbn{body=\body, pow=0} \bsstepbtr \bscfgbn{body=\body, pow=n} \raises \eventpow{pow=0} :: \eventpow{pow=1} :: \cdots :: \eventpow{pow=n}}
} {
  \inferrule*[Lab=Trace-NextStep]
  {#1 \\ #2}
  {#3}
}

\ekvcSplit\brsrceqvevent{
  p1={\Hspace_1\lequiv\Hspace_2},
  p2={\Sspace_1\lequiv\Sspace_2},
  c={\eventsub{sub=1}\lequiv\eventsub{sub=2}}
} {
  \inferrule*[Lab=Eqv-Event]
  {#1 \\ #2}
  {#3}
}

\ekvcSplit\brsrceqvtracelockstep{
  p1={T_1 = \eventsub{sub=1} :: T_1'},
  p2={T_2 = \eventsub{sub=2} :: T_2'},
  p3={\eventsub{sub=1} \lequiv \eventsub{sub=2}},
  p4={T'_1 \simeq T'_2}, 
  c={T_1\simeq T_2}
} {
  \inferrule*[Lab=Eqv-Trace-LockStep]
  {#1 \\ #2 \\\\ #3 \\ #4}
  {#5}
}

\ekvcSplit\brsrceqvtracehalfstep{
  p1={T_1 = \eventsub{sub=1} :: T'_1},
  p2={T_2 = \eventsub{sub=2} :: T'_2},
  p3={\eventsub{sub=1} \lequiv \eventsub{sub=2}},
  p4={T'_1 \simeq T_2}, 
  c={T_1\simeq T_2}
} {
  \inferrule*[Lab=Eqv-Trace-HalfStep]
  {#1 \\ #2 \\\\ #3 \\ #4}
  {#5}
}

\ekvcSplit\brsrceqvtracesymm{
  p1={T_2 \simeq T_1}, 
  c={T_1\simeq T_2}
} {
  \inferrule*[Lab=Eqv-Trace-Symm]
  {#1}
  {#2}
}

\ekvcSplit\brsrceqvtraceempty{
  p1={T_1 = T_2 = \varnothing},
  c={T_1\simeq T_2}
} {
  \inferrule*[Lab=Eqv-Trace-Empty]
  {#1}
  {#2}
}

\ekvcSplit\trgtracenostep{
  c={\tcfgb{body=\body} \sstepbtr \tcfgb{body=\body} \raises \configtwo{\THspace}{\TSspace}}
} {
  \inferrule*[Lab=Trace-NoStep]
  {~}
  {#1}
}

\ekvcSplit\trgtracenextstep{
  p1={\tcfgbn{body=\body, pow=0} \sstepb \tcfgbn{body=\body, pow=1}},
  p2={\tcfgbn{body=\body, pow=1} \sstepbtr \tcfgbn{body=\body, pow=n} \raises \tcfgeventpow{pow=1} \cdots \tcfgeventpow{pow=n}},
  c={\tcfgbn{body=\body, pow=0} \sstepbtr \tcfgbn{body=\body, pow=n} \raises \tcfgeventpow{pow=0} :: \tcfgeventpow{pow=1} :: \cdots :: \tcfgeventpow{pow=n}}
} {
  \inferrule*[Lab=Trace-NextStep]
  {#1 \\ #2}
  {#3}
}

\ekvcSplit\trgeqvevent{
  p1={\THspace_1 \simeq \THspace_2},
  p2={\TSspace_1 \simeq \TSspace_2},
  c={\tcfgeventsub{sub=1}\simeq\tcfgeventsub{sub=2}}
} {
  \inferrule*[Lab=L-Eqv-Event]
  {#1 \\ #2}
  {#3}
}

\ekvcSplit\trgeqvtracelockstep{
  p1={\TT_1 = \tcfgeventsub{sub=1} :: \TT_1'},
  p2={\TT_2 = \tcfgeventsub{sub=2} :: \TT_2'},
  p3={\tcfgeventsub{sub=1} \simeq \tcfgeventsub{sub=2}},
  p4={\TT'_1 \simeq \TT'_2}, 
  c={\TT_1\simeq \TT_2}
} {
  \inferrule*[Lab=L-Eqv-Trace-LockStep]
  {#1 \\ #2 \\\\ #3 \\ #4}
  {#5}
}

\ekvcSplit\trgeqvtracehalfstep{
  p1={\TT_1 = \tcfgeventsub{sub=1} :: \TT'_1},
  p2={\TT_2 = \tcfgeventsub{sub=2} :: \TT'_2},
  p3={\tcfgeventsub{sub=1} \simeq \tcfgeventsub{sub=2}},
  p4={\TT'_1 \simeq \TT_2}, 
  c={\TT_1\simeq \TT_2}
} {
  \inferrule*[Lab=L-Eqv-Trace-HalfStep]
  {#1 \\ #2 \\\\ #3 \\ #4}
  {#5}
}

\ekvcSplit\trgeqvtracesymm{
  p1={\TT_2 \simeq \TT_1}, 
  c={\TT_1\simeq \TT_2}
} {
  \inferrule*[Lab=L-Eqv-Trace-Symm]
  {#1}
  {#2}
}

\ekvcSplit\trgeqvtraceempty{
  p1={\TT_1 = \TT_2 = \varnothing},
  c={\TT_1\simeq \TT_2}
} {
  \inferrule*[Lab=L-Eqv-Trace-Empty]
  {#1}
  {#2}
}

\ekvcSplit\trgeqvhevent{
  p1={\THspace_1 \hequiv \THspace_2},
  p2={\TSspace_1 \hequiv \TSspace_2},
  c={\tcfgeventsub{sub=1}\hequiv\tcfgeventsub{sub=2}}
} {
  \inferrule*[Lab=H-Eqv-Event]
  {#1 \\ #2}
  {#3}
}

\ekvcSplit\trgeqvhtracelockstep{
  p1={\TT_1 = \tcfgeventsub{sub=1} :: \TT_1'},
  p2={\TT_2 = \tcfgeventsub{sub=2} :: \TT_2'},
  p3={\tcfgeventsub{sub=1} \hequiv \tcfgeventsub{sub=2}},
  p4={\TT'_1 \hequiv \TT'_2}, 
  c={\TT_1\hequiv \TT_2}
} {
  \inferrule*[Lab=H-Eqv-Trace-LockStep]
  {#1 \\ #2 \\\\ #3 \\ #4}
  {#5}
}

\ekvcSplit\trgeqvhtracehalfstep{
  p1={\TT_1 = \tcfgeventsub{sub=1} :: \TT'_1},
  p2={\TT_2 = \tcfgeventsub{sub=2} :: \TT'_2},
  p3={\tcfgeventsub{sub=1} \hequiv \tcfgeventsub{sub=2}},
  p4={\TT'_1 \hequiv \TT_2}, 
  c={\TT_1\hequiv \TT_2}
} {
  \inferrule*[Lab=H-Eqv-Trace-HalfStep]
  {#1 \\ #2 \\\\ #3 \\ #4}
  {#5}
}

\ekvcSplit\trgeqvhtracesymm{
  p1={\TT_2 \hequiv \TT_1}, 
  c={\TT_1\hequiv \TT_2}
} {
  \inferrule*[Lab=H-Eqv-Trace-Symm]
  {#1}
  {#2}
}

\ekvcSplit\trgeqvhtraceempty{
  p1={\TT_1 = \TT_2 = \varnothing},
  c={\TT_1\hequiv \TT_2}
} {
  \inferrule*[Lab=H-Eqv-Trace-Empty]
  {#1}
  {#2}
}

\ekvcSplit\trgeconst{c={\tcfge{e=c} \tstepe c}}{
  \inferrule*[Lab=Et-Const]
  {~}
  {#1}
}

\ekvcSplit\trgevar{p={\TVspace(x) = \kappa},c={ \tcfge{e=\var{x}} \tstepe \kappa }}{
  \inferrule*[Lab=Et-Var]
  {#1}
  {#2}
}

\ekvcSplit\trgeasgn{
  p1={\configtwo{\TVspace}{e} \tstepe \kappa},
  p2={\TVspace' = \TVspace[x\mapsto \kappa]},
  c={\tcfgi{i=\Asgn{\var{x}}{\type{\tau}{l} e}} \tstepi  \tcfgi{v=\TVspace'}}
} {
  \inferrule*[Lab=Et-Asgn]
  {#1 \\ #2}
  {#3}
}

\ekvcSplit\trgeadd{
  p1={\configtwo{\TVspace}{e_i} \tstepe c_i},
  p2={\TVspace' = \TVspace[x\mapsto c_1+c_2]},
  c={\tcfgi{i=\Add{\var{x}}{\type{\tau}{l}  e_1}{e_2}} \tstepi  \tcfgi{v=\TVspace'}}
} {
  \inferrule*[Lab=Et-Add]
  {#1 \\ #2}
  {#3}
}

\ekvcSplit\trgealloca{
  p1={m = \mathrm{next}(\TMenv,\TStspace,\ty{}) \\ \TMenv(m) = \ty{}},
  p2={\TStspace(m) = \varnothing},
  p3={\TVspace'=\TVspace[x\mapsto m]},
  p4={\TStspace' = \TStspace[m\mapsto0]},
  c={\tcfgi{s=\TStspace::\TSspace, i=\Alloca{\var{x}}{\type{\tau}{l}}} \tstepi  \tcfgi{s=\TStspace'::\TSspace, v=\TVspace'}}
} {
  \inferrule*[Lab=Et-Alloca]
  {#1 \\ #2 \\\\ #3 \\ #4}
  {#5}
}

\ekvcSplit\trgemalloc{
  p1={m = \mathrm{next}(\TMenv,\THspace_{\mu'},\ty{}) \\ \TMenv(m) = \ty{}},
  p2={\THspace_{\mu'}(m) = \varnothing},
  p3={\TVspace'=\TVspace[x\mapsto m]},
  p4={\THspace' = \THspace[m\mapsto0]},
  p5={(\mu = \mu' \lor \mu'=0)},
  c={\tcfgi{i=\TgtMalloc{\var{x}}{\type{\tau}{l}}{\mu'}} \tstepi  \tcfgi{h=\THspace', v=\TVspace'}}
} {
  \inferrule*[Lab=Et-Malloc]
  {#1 \\ #2 \\\\ #3 \\ #4 \\ #5}
  {#6}
}

\ekvcSplit\trgefree{
  p1={\TVspace(x) = m},
  p2={m\in \Loc_{\THspace_{\mu_x}}},
  p3={\mu = \mu_x \lor \mu_x = 0},
  p4={\THspace[m] \ne \varnothing},
  p5={\THspace' = \THspace[m\mapsto\varnothing]},
  c={\tcfgi{i=\Free{\type{\tgtptr{\tau}{\mu}}{l}\var{x}}} \tstepi  \tcfgi{h=\THspace'}}
} {
  \inferrule*[Lab=Et-Free]
  {#1 \\ #2 \\ #3 \\ #4 \\ #5}
  {#6}
}

\ekvcSplit\trgestore{
  p1={\configtwo{\TVspace}{e} \tstepe \kappa},
  p2={\TVspace(y) = m},
  p3={\TMenv(m) =  (\mu_y, \type{\tgtptr{\tau'}{\mu'}}{l_2}) \implies \kappa \in \region{m}},
  p4={m\ne0},
  p5={\mu_y \in\{\mu, 0\}},
  p6={\sirenmemory{\mu}[m]\neq \varnothing},
  p7={\TMspace' = \TMspace [m\mapsto \kappa]},
  c={\tcfgi{i=\Store{\type{\tau}{l_1} e}{\type{\tgtptr{\tau}{\mu_y}}{l_2} \var{y}}} \tstepi  \tcfgi{h=\THspace',s=\TSspace'}}
} {
  \inferrule*[Lab=Et-Store]
  {#1 \\ #2 \\ #3 \\ #4 \\ #5 \\ #6 \\ #7}
  {#8}
}

\ekvcSplit\trgeload{
  p1={\TVspace(y) = m},
  p2={m \neq 0},
  p3={\mu_y \in\{\mu, 0\}},
  p4={\sirenmemory{\mu}(m) = \kappa},
  p5={\mathrm{isptr}(\tau) \implies \kappa\in\region{m}},
  p6={\TVspace'=\TVspace[x\mapsto\kappa]},
  c={\tcfgi{i=\Load{\var{x}}{\type{\tau}{l_1}}{\type{\tgtptr{\tau}{\mu_y}}{l_2}\var{y}}} \tstepi  \tcfgi{v=\TVspace'}}
} {
  \inferrule*[Lab=Et-Load]
  {#1 \\ #2 \\ #3 \\\\ #4 \\\\ #5 \\ #6}
  {#7}
}

\ekvcSplit\trgebrconditional{
  p1={\configtwo{\TVspace}{e} \tstepe d},
  p2={d\in\{0,1\}},
  c={\tcfgb{body=\BrCond{\type{\mathrm{Int_1}}{l} e}{\var{b_0}}{ \var{b_1}}} \tstepb  \tcfgb{body=\Br{ \var{b_c}}}}
} {
  \inferrule*[Lab=Et-Br-Conditional]
  {#1 \\ #2}
  {#3}
}

\ekvcSplit\trgebrunconditional{
  p1={\langle b, \mu \rangle: body \in \TFspace(b_c)},
  p2={b_p'=b_c},
  c={\tcfgb{body=\Br{ \var{b}}} \tstepb  \tcfgb{mode=\mu, prev=b_p', cur=b, body=body}}
} {
  \inferrule*[Lab=Et-Br-Unconditional]
  {#1 \\ #2}
  {#3}
}

\ekvcSplit\trgeinstruction{
p1={\tcfgi{i=i} \tstepi {\tcfgi{h=\THspace', s=\TSspace', v=\TVspace', e=\Elive', i=i'}}},
c={\tcfgb{body=\Seq{i}{\body}} \tstepb  \tcfgb{h=\THspace', s=\TSspace', v=\TVspace', e=\Elive', body=\Seq{i'}{\body}}}
} {
  \inferrule*[Lab=Et-Instruction]
  {#1}
  {#2}
}

\ekvcSplit\trgenop{
  p1={~},
  c={\tcfgb{body=\Seq{\Nop}{body}} \tstepb  \tcfgb{body=body}}
} {
  \inferrule*[Lab=Et-Nop]
  {#1}
  {#2}
}

\ekvcSplit\trgephi{
  p1={i\in [1, n]},
  p2={\prev = b_i},
  p3={\tcfge{e=e_i} \tstepe \kappa},
  p4={\TVspace' = \TVspace[x\mapsto\kappa]},
  c={\tcfgi{i=\PhiNode{\var{x}}{\type{\tau}{l}}{[b_1, e_1] \dots [b_n,e_n]}} \tstepi    \tcfgi{v=\TVspace',i=\Nop}}
} {
  \inferrule*[Lab=Et-Phi]
  {#1 \\ #2 \\ #3 \\ #4}
  {#5}
}

\ekvcSplit\trgeecreate{
  p1={\langle b, \mu\rangle \in \TFspace(b_c)},
  p2={\Elive' = \Elive \cup \{\mu\}},
  p3={\THspace' = \bigcup_{\mu \in h\cup\Elive'} \THspace_\mu},
  c={\tcfgi{mode=0, i={\EnclaveCreate{\mu}{ \var{b}}}} \tstepi  \tcfgi{mode=0, h=\THspace', e=\Elive'}}
} {
  \inferrule*[Lab=Et-Ecreate]
  {#1 \\ #2 \\ #3}
  {#4}
}

\ekvcSplit\trgeekill{
  p1={\mu\in \Elive},
  p2={\Elive' = \Elive \setminus \{\mu\}},
  p3={\THspace' = \THspace\setminus\THspace_{\mu}},
  c={\tcfgi{mode=0, i=\EnclaveKill{\mu}} \tstepi  \tcfgi{mode=0, h=\THspace', e=\Elive'}}
} {
  \inferrule*[Lab=Et-Ekill]
  {#1 \\ #2 \\ #3}
  {#4}
}

\ekvcSplit\trgeeenter{
  p1={\langle b, \mu\rangle: body' \in \TFspace(b_c)},
  p2={\mu\in{\Elive}},
  p3={\fresh{\TStspace, \mu}},
  p4={\TXspace' = \tcfgxenclave{mode=0} :: \TXspace},
  c={\tcfgb{mode=0, body=\Seq{\EnclaveEnter{\var{b}}}{body}} \tstepb  \tcfgb{x=\TXspace', s=\TStspace::\TSspace, prev=b_c, cur=b, mode=\mu,
        body=body'}}
} {
  \inferrule*[Lab=Et-Eenter]
  {#1 \\ #2 \\ #3 \\ #4}
  {#5}
}

\ekvcSplit\trgeeexit{
  p1={\langle b_c', 0\rangle \in \TFspace(b_c)},
  p2={\mu \ne 0},
  p3={\TXspace = \tcfgxenclave{prev=b_p', cur=b_c', mode=0, body=\body'} :: \TXspace'},
  c={\tcfgb{s=\TStspace::\TSspace, body=\EnclaveExit} \tstepb  \tcfgb{x=\TXspace', v=\TVspace, prev=b_p', cur=b_c', mode=0, body=\body'}}
} {
  \inferrule*[Lab=Et-Eexit]
  {#1 \\ #2 \\ #3 }
  {#4}
}

\ekvcSplit\trgeret{
  p1={\configtwo{\TVspace}{e} \tstepe \kappa},
  p2={\TXspace = \tcfgxcall{v=\TVspace', prev=\prev', cur=\cur'} :: \TXspace'},
  p3={\TVspace'' = \TVspace'[x\mapsto\kappa]},
  c={\tcfgb{s=\TStspace::\TSspace, body=\Ret{\type{\tau}{l} e}} \tstepb
       \tcfgb{v=\TVspace'', x=\TXspace', cur=\cur', prev=\prev', body=\body} }
} {
  \inferrule*[Lab=Et-Ret]
  {#1 \\ #2 \\ #3}
  {#4}
}

\ekvcSplit\trgeoret{
  p1={\configtwo{\TVspace}{e} \tstepe \kappa},
  p2={\TXspace = \tcfgxcall{kind=ocall, v=\TVspace', prev=b'_p, cur=b'_c} :: \TXspace'},
  p3={\TVspace'' = \TVspace'[x\mapsto\kappa]},
  c={\tcfgb{x=\TXspace, mode=0, body=\ORet{\type{\tau}{l} \var{x}}} \tstepb  \tcfgb{v=\TVspace'', x=\TXspace', prev=\prev', cur=\cur',
        body=body}}
} {
  \inferrule*[Lab=Et-ORet]
  {#1 \\ #2 \\ #3}
  {#4}
}

\ekvcSplit\trgecall{
  p1={\TFspace(g) = \FDef{\type{\tau}{l_g}~g^{\mu,\ocstat}}{\paramlist}{\overline{\langle b_i, \mu_i\rangle:body_i}}},
  p2={block_1 = b_1^{\mu} : \mathit{body_1}},
  p3={\configtwo{\TVspace}{\argi} \tstepe \kappa_i},
  p4={\TVspace' = \{\parami : \kappa_i\}},
  p5={\fresh{\TStspace, \mu}},
  p6={\TXspace' = \tcfgxcall{} :: \TXspace},
  c={\tcfgb{body=\Seq{\FCall{\var{x}}
          {\type{\tau}{l_f}\fname{f}}
          {\arglist}}
        {body}} \tstepb \\ \tcfgb{s=\TStspace::\TSspace, v=\TVspace', x=\TXspace', prev=\varnothing, cur=b_1,
        body=body_1}}
} {
  \inferrule*[Lab=Et-Call]
  {#1 \\ #3 \\ #4 \\ #5 \\ #6}
  {#7}
}

\ekvcSplit\trgeocall{
  p1={\TFspace(f) = \FDef{\type{\tau}{l_f}~\fname{f}^{0, \ocyes}}{\paramlist}{\overline{block}}},
  p2={block_0 = b_0^{0} : \mathit{body_0}},
  p3={\TVspace' = \{\parami : c_i\}},
  p4={\configtwo{\TVspace}{e_i} \tstepe c_i},
  p5={\fresh{\TStspace, 0}},
  p6={\TXspace' = \tcfgxcall{v=\{\}, prev=\varnothing, cur=\varnothing, mode=0, body=\code{oret\ }\type{\tau}{l}\var{x}} :: \tcfgxcall{kind=ocall} :: \TXspace},
  p7={\mode \ne 0},
  c={\tcfgb{body=\Seq{\OCall{\var{x}}
        {\type{\tau}{l_{ret}}\fname{f}}
        {\arglist}}
      {body}} \tstepb  \tcfgb{s=\TStspace::\TSspace, v=\TVspace', x=\TXspace', prev=\varnothing, cur=b_0, mode=0,
      body=\mathit{body}_0}}
} {
  \inferrule*[Lab=Et-OCall]
  {#1 \\ #2 \\ #3 \\ #4 \\ #5 \\ #6 \\ #7}
  {#8}
}

\ekvcSplit\trgegep{
  p1={\TVspace(y) = m},
  p2={\configtwo{\TVspace}{e} \tstepe c},
  p3={m' = m + c},
  p4={m'\in\region{m} \\ m'\in\modereg{m}},
  p5={\TVspace' = \TVspace[x \mapsto m']},
  c={\tcfgi{i=\GEP{\var{x}}
        {\type{\tau_1}{l_1}}
        {\type{\tgtptr{\tau_1}{\mu}}{l_1}\var{y}}
        {\type{\tau_2}{l_2} e}} \tstepi  \tcfgi{v=\TVspace'}}
} {
  \inferrule*[Lab=Et-Gep]
  {#1 \\ #2 \\ #3 \\\\ #4 \\ #5}
  {#6}
}

\ekvcSplit\trgepreserve{
  p1={\TVspace' = \TVspace[x \mapsto \TVspace(y)]},
  c={\tcfgi{i=\Preserve{\var{x}}{\type{\tau}{l} \var{y}}} \tstepi  \tcfgi{v=\TVspace', i=\Nop}}
} {
  \inferrule*[Lab=Et-Preserve]
  {#1}
  {#2}
}

\ekvcSplit\trgeprogram{
  p1={\forall i.\ \func_i = \FDef{\type{\tau}{l} f_i^{\mu_f,oc_f}}{\paramlist}{\overline{block}}},
  p2={\TFspace_f=\{f_i :\func_i \mid \forall i\}},
  p3={\TFspace_b= \{b : \func_i \mid \forall i,\ \forall b \in \func_i\}},
  p4={\TFspace = \TFspace_f \cup \TFspace_b},
  p5={\func_0 = \FDef{\type{\itt}{\bot} \main^{0, \ocno}}{}{\mathit{block}^+_{\main}}},
  p6={\forall j.\  x_{j} \in \TVspace},
  p7={\tcfgb{h=\THspace, e=\Elive, s=\varnothing, mode=0, x=\varnothing, prev=\varnothing, cur=\varnothing, body=\Seq{\FCall{\var{x}}{\main}{} }{\Ret{\var{x}}}}
    \ttrstepb
    \tcfgb{h=\THspace', e=\Elive', s=\varnothing, mode=0, v=\TVspace', x=\varnothing, prev=\varnothing, cur=\varnothing, body=\Ret{x}}},
  c={\langle \THspace, \TVspace, \func_0, \dots, \func_n \rangle \tstepp \configtwo{\THspace'}{\TVspace'(x)}}}{
  \inferrule*[Lab=Et-Program]
  {#1 \\ #5  \\ #7}
  {#8}
}

\ekvcSplit\trgprogramtrace{
  p1={\forall i.\ \func_i = \FDef{\type{\tau}{l} f_i^{\mu_f,oc_f}}{\paramlist}{\overline{block}}},
  p2={\TFspace_f=\{f_i :\func_i \mid \forall i\}},
  p3={\TFspace_b= \{b : \func_i \mid \forall i,\ \forall b \in func_i\}},
  p4={\TFspace = \TFspace_f \cup \TFspace_b},
  p5={func_0 = \FDef{\type{\itt}{\bot} \main^{0, \ocno}}{}{block^+_{\main}}},
  p6={\forall j.\  x_{j} \in V},
  p7={\Elive = \{1\} \\ \THspace = \THspace_0 \cup \THspace_1},
  p8={\tcfgb{s=\varnothing, x=\varnothing, prev=\varnothing, cur=\varnothing, body=\Seq{\FCall{\var{x}}{\type{\itt}{\public} \main}{}}{\Ret{\var{x}}}}
    \tsteptr
    \tcfgb{h=\Hspace', s=\varnothing, v=\TVspace', x=\varnothing, prev=\varnothing, cur=\varnothing, body=\Ret{\var{x}}}
    \triangleright
    T},
  c={\langle \THspace, \TVspace, func_0, \dots, func_N \rangle \tsteptr \configtwo{\THspace'}{\TVspace'(x)}}
  \triangleright
  T}{
  \inferrule*[Lab=Trace-Program]
  {#1 \\ #5 \\ #8}
  {#9}
}

\ekvcSplit\trgtdownarrowa{
  p1={l_1 \sqsubseteq l_2},
  c={l_1\downarrow_b \sqsubseteq l_2}
} {
  \inferrule*[Lab=Tt-DownArrow1]
  {#1}
  {#2}
}

\ekvcSplit\trgtdownarrowb{
  p1={l_1 \sqsubseteq l_2},
  c={l_1\downarrow_{b_1} \sqsubseteq_{b_2} l_2}
} {
  \inferrule*[Lab=Tt-DownArrow2]
  {#1}
  {#2}
}

\ekvcSplit\trgtdownarrowc{
  p1={~},
  c={l_1\downarrow_b \sqsubseteq_b \bot}
} {
  \inferrule*[Lab=Tt-DownArrow3]
  {#1}
  {#2}
}

\ekvcSplit\trgtconst{
  p1={\tau\in\{\mathrm{Int_1}, \mathrm{Int_{32}}\}},
  c={\ttctxe{} \tc  c : \type{\tau}{\bot}}
} {
  \inferrule*[Lab=Tt-Const]
  {#1}
  {#2}
}

\ekvcSplit\trgtvar{
  p1={\TTypeEnv(x) = \type{\tau}{l}},
  c={\ttctxe{} \tc  \var{x} : \type{\tau}{l}}
} {
  \inferrule*[Lab=Tt-Var]
  {#1}
  {#2}
}

\ekvcSplit\trgtsub{
  p1={\TTypeEnv \tc e : \type{\tau}{l_1}},
  p2={\type{\tau}{l_1} \sqsubseteq \type{\tau}{l_2}},
  c={\ttctxe{}\tc  e : \type{\tau}{l_2}}
} {
  \inferrule*[Lab=Tt-Sub]
  {#1 \\ #2}
  {#3}
}

\ekvcSplit\trgtnop{
  p1={~},
  c={\ttctxb{} \tc  \Nop}
} {
  \inferrule*[Lab=Tt-Nop]
  {#1}
  {#2}
}

\ekvcSplit\trgtasgn{
  p1={\mu\neq0 \lor l\sqsubseteq\public},
  p2={\TTypeEnv(x) = \type{\tau}{l}},
  p3={\TTypeEnv \tc e : \type{\tau}{l}},
  p4={\pc \flowsto l},
  c={\ttctxi{} \tc  \Asgn{\var{x}}{\type{\tau}{l}e}}
} {
  \inferrule*[Lab=Tt-Asgn]
  {#1 \\ #2 \\\\ #3 \\ #4}
  {#5}
}

\ekvcSplit\trgtadd{
  p1={\mu\neq0 \lor l\sqsubseteq\public},
  p2={\TTypeEnv(x) = \type{\tau}{l}},
  p3={\forall i: \TTypeEnv \vdash e_i : \type{\tau}{l}},
  p4={\pc\sqsubseteq l},
  c={\ttctxi{} \tc  \Add{\var{x}}{\type{\tau}{l} e_1}{e_2}}
} {
  \inferrule*[Lab=Tt-Add]
  {#1 \\ #2 \\\\ #3 \\ #4}
  {#5}
}

\ekvcSplit\trgtalloca{
  p1={\mu\neq0 \lor l\sqsubseteq\public},
  p2={\TTypeEnv(x) = \type{\tgtptr{\tau}{\mu}}{l}},
  p3={\pc \flowsto l},
  c={\ttctxi{} \tc  \Alloca{\var{x}}{\type{\tau}{l}}}
} {
  \inferrule*[Lab=Tt-Alloca]
  {#1 \\ #2 \\ #3}
  {#4}
}

\ekvcSplit\trgtmalloc{
  p1={\mu\neq0 \lor l\sqsubseteq\public},
  p2={\TTypeEnv(x) = \type{\tgtptr{\tau}{\mu_x}}{l}},
  p3={\pc \flowsto l},
  c={\ttctxi{} \tc  \TgtMalloc{\var{x}}{\type{\tau}{l}}{\mu_x}}
} {
  \inferrule*[Lab=Tt-Malloc]
  {#1 \\ #2 \\ #3}
  {#4}
}

\ekvcSplit\trgtfree{
  p1={\mu\neq0 \lor l\sqsubseteq\public},
  p2={\TTypeEnv(x) = \type{\tgtptr{\tau}{\mu_x}}{l}},
  p3={\pc \flowsto l},
  c={\ttctxi{} \tc  \Free{\type{\tgtptr{\tau}{\mu_x}}{l}}{\var{x}}}
} {
  \inferrule*[Lab=Tt-Free]
  {#1 \\ #2 \\ #3}
  {#4}
}

\ekvcSplit\trgtload{
  p1={\mu\neq0 \lor l_1\sqsubseteq\public},
  p2={\TTypeEnv(x)=\type{\tau}{l_1}},
  p3={\TTypeEnv \tc y : \type{\type{\ptr{\tau}}{\mu'}}{l_2}},
  p4={(\mu=\mu') \lor (\mu'=0)},
  p5={\type{\tau}{l_2} \sqsubseteq \type{\tau}{l_1}},
  p6={\pc \sqsubseteq l_1},
  c={\ttctxi{} \tc  \Load{\var{x}}{\type{\tau}{l_1}}{\type{\type{\ptr{\tau}}{\mu'}}{l_2}\var{y}}}
} {
  \inferrule*[Lab=Tt-Load]
  {#1 \\ #2 \\ #3 \\\\ #4 \\ #5 \\ #6}
  {#7}
}

\ekvcSplit\trgtstore{
  p1={\mu\neq0 \lor l_2\sqsubseteq\public},
  p2={\TTypeEnv \tc e : \type{\tau}{l_1}},
  p3={\TTypeEnv \tc \var{y} : \type{\type{\ptr{\tau}}{\mu'}}{l_2}},
  p4={(\mu=\mu_y) \lor (\mu_y=0)},
  p5={\type{\tau}{l_1} \sqsubseteq \type{\tau}{l_2}},
  p6={\pc \sqsubseteq l_2},
  c={\ttctxi{} \tc \Store{\type{\tau}{l_1}e}{\type{\type{\ptr{\tau}}{\mu_y}}{l_2}\var{y}}}
} {
  \inferrule*[Lab=Tt-Store]
  {#1 \\ #2 \\ #3 \\\\ #4 \\ #5 \\ #6}
  {#7}
}

\ekvcSplit\trgtphi{
  p1={\mu\neq0 \lor l\sqsubseteq\public},
  p2={\TTypeEnv(x) = \type{\tau}{l}},
  p3={\forall i.\ \TTypeEnv \tc e_i : \type{\tau}{l} },
  p4={\forall i.\ \TBenv(b_i) \flowsto l},
  p5={\pc\flowsto l},
  c={\ttctxi{} \tc  \PhiNode{\var{x}}{\type{\tau}{l}}{[b_0, e_0] \dots [b_n,e_n]}}
} {
  \inferrule*[Lab=Tt-Phi]
  {#1 \\ #2 \\ #3 \\\\ #4 \\ #5}
  {#6}
}

\ekvcSplit\trgtgep{
  p1={\mu\neq0 \lor l_1\sqsubseteq\public},
  p2={\TTypeEnv(x) = \TTypeEnv(y) = \type{\tgtptr{\tau_1}{\mu_y}}{l_1}},
  p3={l_2 \sqsubseteq l_1},
  p4={\tau_2 \in \{\mathrm{Int_{32}}, \mathrm{Int_1}\}},
  p5={\pc\sqsubseteq l_1},
  p6={(\mu_y = \mu \lor \mu_y = 0)},
  c={\ttctxi{} \tc  \GEP{\var{x}}{\type{\tau_1}{l_1}}{\type{\tgtptr{\tau_1}{\mu_y}}{l_1}\var{y}}{\type{\tau_2}{l_2} e}}
} {
  \inferrule*[Lab=Tt-Gep]
  {#1 \\ #2 \\ #3 \\\\ #4 \\ #5}
  {#7}
}

\ekvcSplit\trgtbrunconditional{
  p1={\TBenv(b) = (\mu, l_b)},
  p2={\pc \sqsubseteq_{b} l_b},
  c={\ttctxb{} \tc  \Br{\var{b}}}
} {
  \inferrule*[Lab=Tt-Br-Unconditional]
  {#1 \\ #2}
  {#3}
}

\ekvcSplit\trgtbrconditional{
  p1={\TTypeEnv \tc e : \type{\mathrm{Int_1}}{l}},
  p2={\trg = \mergepoint{\cur}},
  p3={\forall i: \TBenv(b_i) = (\mu, l_i)},
  p4={\forall i: \pc \sqcup l\downarrow_{\trg} \sqsubseteq_{b_i} l_i},
  c={\ttctxb{} \tc  \BrCond{\type{\mathrm{Int_1}}{l}e}{\var{b_1}}{\var{b_2}}}
} {
  \inferrule*[Lab=Tt-Br-Conditional]
  {#1 \\ #2 \\\\ #3 \\ #4}
  {#5}
}

\ekvcSplit\trgtcall{
  p1={\TFenv(g) = \FHead{\type{\tau}{l_g} \fname{g}^{\mu_g,oc_g}}{\paramlist}{}},
  p2={\mu\neq0 \lor l_g\sqsubseteq\public},
  p3={\TTypeEnv(x) = \type{\tau}{l_g}},
  p4={l_x = l_g},
  p5={\pc \sqsubseteq l_g},
  p6={\forall i:\TTypeEnv \tc e_i : \type{\tau_i}{l_i}},
  p7={\mu_g = \mu},
  p8={oc_g = {\ocyes} \lor oc = {\ocno}},
  c={\ttctxb{} \tc  \FCall{\var{x}}{\type{\tau}{l_g} \fname{g}}{\type{\tau_1}{l_1}e_1, \type{\tau_2}{l_2}e_2\dots, \type{\tau_n}{l_n}e_n}}
} {
  \inferrule*[Lab=Tt-Call]
  {#1 \\ #2 \\ #3 \\ #5 \\ #6 \\ #7 \\ #8}
  {#9}
}

\ekvcSplit\trgtret{
  p1={\TTypeEnv \tc e : \type{\tau}{l_f}},
  p2={\type{\tau}{l} = \type{\tau}{\lf}},
  p3={\mu \neq 0 \lor l \flowsto \bot},
  c={\ttctxb{} \tc  \Ret{\type{\tau}{l_f} e}}
} {
  \inferrule*[Lab=Tt-Ret]
  {#1 \\ #3}
  {#4}
}

\ekvcSplit\trgtbody{
  p1={\ttctxb{} \tc  i},
  p2={\ttctxb{} \tc  body},
  c={\ttctxb{} \tc  \Seq{i}{body}}
} {
  \inferrule*[Lab=Tt-Body]
  {#1 \\\\ #2}
  {#3}
}

\ekvcSplit\trgteenter{
  p1={\TBenv(b) = (\mu_b, l_b)},
  p2={\mu_b \ne 0},
  p3={\mu = 0},
  p4={\pc \sqsubseteq l_b},
  c={\ttctxb{} \tc  \EnclaveEnter{\var{b}}}
} {
  \inferrule*[Lab=Tt-Eenter]
  {#1 \\ #2 \\ #3 \\ #4}
  {#5}
}

\ekvcSplit\trgteexit{
  p1={\mu \ne 0},
  c={\ttctxb{} \tc  \EnclaveExit}
} {
  \inferrule*[Lab=Tt-Eexit]
  {#1}
  {#2}
}

\ekvcSplit\trgtecreate{
  p1={\TBenv(b) = (\mu_e, l_b)},
  p2={\mu = 0},
  p4={\mu_e\ne 0},
  p5={pc\sqsubseteq\public},
  c={\ttctxb{} \tc  \EnclaveCreate{ \mu_e}{\var{b}}}
} {
  \inferrule*[Lab=Tt-Ecreate]
  {#1 \\ #2 \\ #3\\ #4}
  {#5}
}

\ekvcSplit\trgtekill{
  p1={\mu = 0},
  p2={\mu_e\ne 0},
  p3={pc\sqsubseteq\public},
  c={\ttctxb{} \tc  \EnclaveKill{\mu_e}}
} {
  \inferrule*[Lab=Tt-Ekill]
  {#1 \\ #2 \\ #3}
  {#4}
}

\ekvcSplit\trgtocall{
  p1={\TFenv(f) = \FHead{\type{\tau}{l_g} \fname{g}^{\mu_g,oc_g}}{\paramlist}{}},
  p2={\mu\neq0 \lor l_x\sqsubseteq\public},
  p3={\TTypeEnv(x) = \type{\tau}{l_x}},
  p4={l_x = l_g},
  p5={\pc \sqsubseteq l_g},
  p6={\forall i:\TTypeEnv \tc e_i : \type{\tau_i}{l_i}},
  p7={\mu_g = 0 \\ \mu \neq 0},
  p8={oc_g=oc=\ocyes},
  c={\ttctxb{} \tc  \OCall{\var{x}}{\type{\tau}{l_g} \fname{g}}{\type{\tau_1}{l_1}e_1, \type{\tau_2}{l_2}e_2\dots, \type{\tau_n}{l_n}e_n}}
} {
  \inferrule*[Lab=Tt-OCall]
  {#1 \\ #2 \\ #3 \\ #4 \\ #5 \\ #6 \\ #7 \\ #8}
  {#10}
}

\ekvcSplit\trgtoret{
  p1={\TTypeEnv \tc e : \type{\tau}{l}},
  p2={\type{\tau}{l} \flowsto \type{\tau}{\lf}},
  p3={\mu=0},
  p4={l\flowsto\bot},
  c={\ttctxb{} \tc  \ORet{\type{\tau}{l} e}}
} {
  \inferrule*[Lab=Tt-ORet]
  {#1 \\ #2 \\ #3 \\ #4}
  {#5}
}

\ekvcSplit\trgtpreserve{
  p1={\mu = 0},
  p2={\TTypeEnv(x) = \type{\tau}{l}},
  p3={\TTypeEnv \tc \var{y} : \type{\tau}{l}},
  p4={\pc\sqsubseteq l},
  c={\ttctxb{} \tc  \Preserve{\var{x}}{\type{\tau}{l} \var{y}}}
} {
  \inferrule*[Lab=Tt-Preserve]
  {#1 \\ #2 \\ #3 \\ #4}
  {#5}
}

\ekvcSplit\trgtfunction{
  p1={\TFenv(f) = \FHead{\type{\tau}{\lf} \fname{f}^{\mu_f, oc}}{\new\paramlist}{}\{\configtwo{b_1}{\mu_1} : \body_1; \configtwo{b_2}{\mu_2}; \cdots\}},
  p2={\forall i: \pc_i = \TBenv(b_i)},
  p3={\forall i: l_f \sqsubseteq \pc_i},
  p4={\forall i: \TTypeEnv_f, \mu_i, oc, \pc_i, b_i, l_f \tc body_i},
  p5={\mu_f=0 \lor \forall i: \mu_f = \mu_i},
  p6={\mu_f \neq 0 \lor \lf \flowsto \bot},
  p7={(oc = \code{\new{\ocno}}) \lor (oc=\new{\ocyes} \land \forall i:\mu_i=0)},
  c={\ttctxf{} \tc \new{\fname{f}}}
} {
  \inferrule*[Lab=Tt-Function]
  {#1 \\ #2 \\ #3 \\ #4 \\ #5 \\ #6 \\ #7}
  {#8}
}

\ekvcSplit\trgtprogram{
  p1={\TFspace(f_0)=\FDef{\type{\mathrm{\itt}}{\public} \fname{main^{0, \ocno}}}{\arglist}{blocks}},
  p2={\forall~i.~\TTypeEnv_{i} = \restrict{\TTypeEnv}{i}},
  p3={\forall~i.~\TBenv_{f_i} = \lfloor \TBenv \rfloor_{f_i}},
  p4={\forall~i.~\TTypeEnv_{i} \tc f_i},
  c={\ttctxp{} \tc  f_0, f_1, \dots,f_n}
} {
  \inferrule*[Lab=Tt-Program]
  {#1 \\ #2  \\ #4}
  {#5}
}

\ekvcSplit\trgtmemoryconfig{
  p1={\forall m\in{\Loc}.\ (\TSspace\cup\THspace)(m) = \varnothing \lor \TMenv \tc (\TSspace\cup\THspace)(m) : \TMenv(m)},
  c={ \tc  \configtwo{\TSspace}{\THspace}}}{
  \inferrule*[Lab=Tt-Memory-Config]
  {#1}
  {#2}
}

\ekvcSplit\trgtvariables{
  p1={\forall x\in\Var.\ \TMenv,\TTypeEnv \tc \TVspace(x) : \TTypeEnv(x)},
  c={ \TTypeEnv \tc  \TVspace}}{
  \inferrule*[Lab=Tt-Variables]
  {#1}
  {#2}
}

\ekvcSplit\trgtfunctions{
  p1={\forall f\in \Fcn.\ ~ \tc \TFspace(f)},
  c={ \tc  \TFspace}}{
  \inferrule*[Lab=Tt-Functions]
  {#1}
  {#2}
}

\ekvcSplit\trgtcallstackcall{
  p1={\TFenv(\cur) = \FHead{\type{\tau}{l_g} \fname{g}^{\mu_g, oc_g}}{\paramlist}{}},
  p2={\mu_f=\mu_{b_c}},
  p3={l_g \flowsto pc \flowsto l_f},
  p4={\pc = \TBenv(\cur)},
  p5={\TTypeEnv_g \tc \TVspace},
  p6={\TTypeEnv_g(x) = \type{\tau}{\lf}},
  p7={\TTypeEnv_g, \pc, l_g \tc body},
  p8={\mu_g, l_g \tc \TXspace},
  c={\mu_f, \lf \tc  \tcfgxcall{mode=\mu_{b_c}}::\TXspace}}{
  \inferrule*[Lab=Tt-Callstack-Call]
  {#1 \\ #2 \\ #3 \\ #4 \\ #5 \\ #6 \\ #7 \\ #8}
  {#9}
}

\ekvcSplit\trgtcallstackocall{
  p1={\TFenv(\cur) = \FHead{\type{\tau}{l_g} \fname{g}^{\mu_g, oc_g}}{\paramlist}{}},
  p2={\mu_{b_c}\neq0},
  p3={\mu_f=0 \\ {l_g \flowsto pc \flowsto l_f}},
  p4={\pc = \TBenv_g(\cur)},
  p5={\TTypeEnv_g \tc \TVspace},
  p6={\TTypeEnv_g(x) = \type{\tau}{\lf}},
  p7={\TTypeEnv_g, \pc, l_g \tc body},
  p8={\mu_g, l_g \tc \TXspace},
  c={\mu_f, \lf \tc  \tcfgxcall{kind=ocall,mode=\mu_{b_c}}::\TXspace}}{
  \inferrule*[Lab=Tt-Callstack-Ocall]
  {#1 \\ #2 \\ #3 \\ #4 \\ #5 \\ #6 \\ #7 \\ #8}
  {#9}
}

\ekvcSplit\trgtcallstackeenter{
  p1={\TFenv(\cur) = \FHead{\type{\tau}{l_h} \fname{h}^{0, \ocno}}{\paramlist}{}},
  p2=\mu_{b_c}= 0,
  p3={\pc = \TBenv(\cur)},
  p4={l_h \flowsto \pc},
  p5=~,
  p6=~,
  p7={\TTypeEnv, \pc, l_h \tc body},
  p8={0, l_h \tc \TXspace},
  c={\mu_f, \lf \tc  \tcfgxenclave{mode=\mu_{b_c}}::\TXspace}}{
  \inferrule*[Lab=Tt-Callstack-Eenter]
  {#1 \\ #2 \\ #3 \\ #4 \\ #7 \\ #8}
  {#9}
}

\ekvcSplit\trgtcallstackbase{
  p1={\TXspace = \varnothing},
  c={\TMenv, \mu_f, \lf \tc  \TXspace}}{
  \inferrule*[Lab=Tt-Callstack-Base]
  {#1}
  {#2}
}

\ekvcSplit\trgtinstructionconfig{
  p1={\TMenv \tc \configtwo{\TSspace}{\THspace}},
  p2={\TMenv, \TTypeEnv \tc \TVspace},
  p3={\trgtypingcontextcom \tc{i}},
  p4={\mu = \mu_{b_c}},
  p5={\tc \configtwo{\Elive}{\THspace}},i
  c={\trgtypingcontextconfig \tc  \tcfgi{mode=\mu_{b_c}, i=i}}}{
  \inferrule*[Lab=Tt-Instruction-Config]
  {#1 \\ #2 \\ #3 \\ #4}
  {#6}
}

\ekvcSplit\trgtelive{
  p1={\forall\mu:(\mu\in \Elive \Leftrightarrow \THspace_\mu\subseteq \THspace)},
  c={\tc \configtwo{\Elive}{\THspace}}}{
  \inferrule*[Lab=Tt-E-live]
  {#1}
  {#2}
}

\ekvcSplit\trgtbodyconfig{
  p1={\tc \configtwo{\TSspace}{\THspace}},
  p2={\TTypeEnv \tc \TVspace},
  p3={\mu, \lf \tc \TXspace},
  p4={\tc \TFspace},
  p5={\tc \configtwo{\Elive}{\THspace}},i
  p6={\trgtypingcontextcom \tc{body}},
  p7={\mu = \mu_{b_c}},
  c={\trgtypingcontextconfig \tc  \tcfgb{mode=\mu_{b_c}, body=body}}}{
  \inferrule*[Lab=Tt-Body-Config]
  {#1 \\ #2 \\ #3 \\ #4 \\ #6 \\ #7}
  {#8}
}

\ekvcSplit\trgtexpressionconfig{
  p1={\TMenv, \TTypeEnv \tc \TVspace},
  p2={\TTypeEnv \tc e : \ty{}},
  c={\TMenv, \TTypeEnv \tc  \tcfge{}}
} {
  \inferrule*[Lab=Tt-Expression-Config]
  {#1 \\ #2}
  {#3}
}

\ekvcSplit\trgtresultconfig{
  p1={\TMenv \tc \configtwo{\varnothing}{\THspace}},
  p2={\tc c : \type{\itt}{\public}},
  c={\TMenv \tc \configtwo{\THspace}{c}}
} {
  \inferrule*[Lab=Tt-Result-Config] 
  {#1 \\\\ #2}
  {#3}
}

\ekvcSplit\trgtprogramconfig{
  p1={ \tc p},
  p2={\TMenv \tc \THspace},
  p3={\TMenv, \TTypeEnv \tc \TVspace},
  c={\TMenv, \TTypeEnv \tc  \tcfgp{}}}{
  \inferrule*[Lab=Tt-Program-Config]
  {#1 \\ #2 \\ #3}
  {#4}
}

\ekvcSplit\brtrgeconst{
  p1={~},
  c={\tcfge{e=\bconst} \btstepe \bconst}}{
  \inferrule*[Lab={[Et-Const]}]
  {#1}
  {#2}
}

\ekvcSplit\brtrgevar{
  p1={\TVspace(x) = \bkappa},
  c={\tcfge{e=\var{x}} \btstepe \bkappa}}{
  \inferrule*[Lab={[Et-Var]}]
  {#1}
  {#2}
}

\ekvcSplit\brtrgeasgn{
  p1={\tcfge{e=e} \btstepe \bkappa},
  p2={\TVspace' = \TVspace[x\mapsto \higher{\bkappa}{l}]},
  c={\btcfgi{i=\Asgn{\var{x}}{\type{\tau}{l} e}} \btstepi 
    \btcfgi{v=\bTVspace'}}}{
  \inferrule*[Lab={[Et-Asgn]}]
  {#1 \\ #2}
  {#3}
}

\ekvcSplit\brtrgeadd{
  p1={\configtwo{\TVspace}{e_i} \tstepe \bconst_i},
  p2={\bTVspace' = \bTVspace\! [x\mapsto \higher{(\bconst_1+\bconst_2)}{l}]},
  c={\tcfgi{i=\Add{\var{x}}{\type{\tau}{l}  e_1}{e_2}} \tstepi  \tcfgi{v=\bTVspace'}}
} {
  \inferrule*[Lab={[Et-Add]}]
  {#1 \\ #2}
  {#3}
}

\ekvcSplit\brtrgealloca{
  p1={m = \mathrm{next}(\TMenv,\TStspace,\ty{}) \\ \TMenv(m) = \ty{}},
  p2={\bTStspace(m) = \varnothing},
  p3={\TVspace'=\TVspace[x\mapsto \higher{m}{l}]},
  p4={\bTStspace' = \bTStspace{} [m\mapsto\higher{0}{l}]},
  c={\btcfgi{s=\bTStspace::\TSspace, i=\Alloca{\var{x}}{\type{\tau}{l}}} \btstepi \btcfgi{s=\bTStspace'::\TSspace, v=\TVspace'}}
} {
  \inferrule*[Lab={[Et-Alloca]}]
  {#1 \\ #2 \\\\ #3 \\ #4}
  {#5}
}

\ekvcSplit\brtrgemalloc{
  p1={m = \mathrm{next}(\TMenv,\THspace_{\mu'},\ty{}) \\ \TMenv(m) = \ty{}},
  p2={\THspace_{\mu'}(m) = \varnothing},
  p3={\TVspace'=\TVspace[x\mapsto\higher{m}{l}]},
  p4={\THspace' = \THspace[m\mapsto\higher{0}{l}]},
  p5={(\mu = \mu' \lor \mu'=0)},
  c={\btcfgi{i=\TgtMalloc{\var{x}}{\type{\tau}{l}}{\mu'}} \tstepi \btcfgi{h=\THspace', v=\TVspace'}}
} {
  \inferrule*[Lab={[Et-Malloc]}]
  {#1 \\ #2 \\\\ #3 \\ #4 \\ #5}
  {#6}
}

\ekvcSplit\brtrgefree{
  p1={\TVspace(x) = \bloc},
  p2={\bloc\in \Loc_{\THspace_{\mu_x}}},
  p3={\mu=\mu_x \lor \mu_x = 0},
  p4={\THspace(\bloc) \ne \varnothing},
  p5={\THspace' = \THspace[\bloc\mapsto\varnothing]},
  c={\btcfgi{i=\Free{\type{\tgtptr{\tau}{\mu_x}}{l}\var{x}}} \tstepi  \btcfgi{h=\THspace'}}
} {
  \inferrule*[Lab={[Et-Free]}]
  {#1 \\ #2 \\ #3 \\ #4 \\ #5}
  {#6}
}

\ekvcSplit\brtrgeload{
  p1={\TVspace(y) = \bloc},
  p2={m\ne0},
  p3={\mu_y \in\{\mu, 0\}},
  p4={\sirenmemory{\mu}(\bloc) = \bkappa},
  p5={\mathrm{isptr}(\tau) \implies \kappa\in\region{m}},
  p6={\TVspace'=\TVspace[x\mapsto \higher{\bkappa}{l_1}]},
  c={\btcfgi{i=\Load{\var{x}}{\type{\tau}{l_1}}{\type{\ptr{\tau}, \mu_y}{l_2}\var{y}}} \btstepi
       \btcfgi{v=\TVspace'}}
} {
  \inferrule*[Lab={[Et-Load]}]
  {#1 \\ #2 \\ #3 \\ #4 \\ #5 \\ #6}
  {#7}
}

\ekvcSplit\brtrgestore{
  p1={\configtwo{\TVspace}{e} \btstepe \bkappa},
  p2={\TVspace(y) = m},
  p3={\Menv(m) = \ty{t=\tau*,l=l_2} \implies \kappa \in \region{m}},
  p4={m\ne0},
  p5={\mu_y \in\{\mu, 0\}},
  p6={\sirenmemory{\mu}[m]\neq \varnothing},
  p7={\TMspace' = \TMspace [m\mapsto \higher{\bkappa}{l_2}]},
  c={\btcfgi{i=\Store{\type{\tau}{l_1} e}{\type{\tgtptr{\tau}{\mu_y}}{l_2} \var{y}}} \btstepi  \btcfgi{h=\THspace',s=\TSspace'}}
} {
  \inferrule*[Lab={[Et-Store]}]
  {#1 \\ #2 \\ #3 \\ #4 \\ #5 \\ #6 \\ #7}
  {#8}
}

\ekvcSplit\brtrgephi{
  p1={i\in [1, n]},
  p2={b_p = b_i},
  p3={\configtwo{\TVspace}{e_i} \btstepe {\bkappa}},
  p4={\bTVspace'={\bTVspace}[x\mapsto \higher{\bkappa}{l}]},
  c={\btcfgi{i=\PhiNode{\var{x}}{\type{\tau}{l}}{[b_1, e_1] \dots [b_n,e_n]}} \tstepi  \btcfgi{i=\Asgn{\var{x}}{e_i}}}
} {
  \inferrule*[Lab={[Et-Phi]}]
  {#1 \\ #2 \\ #3 \\ #4}
  {#5}
}

\ekvcSplit\brtrgegep{
  p1={\TVspace(y) = \bloc},
  p2={\configtwo{\TVspace}{e} \btstepe \bconst},
  p3={\bloc' = \bloc + \bconst},
  p4={\bloc'\in\region{m} \\ m'\in\modereg{m}},
  p5={\TVspace' = \TVspace[x \mapsto \bloc']},
  c={\btcfgi{i=\GEP{\var{x}}
        {\type{\tau_1}{l_1}}
        {\type{\tgtptr{\tau_1}{\mu_y}}{l_1}\var{y}}
        {\type{\tau_2}{l_2} e}} \btstepi  \btcfgi{v=\TVspace'}}
} {
  \inferrule*[Lab={[Et-Gep]}]
  {#1 \\ #2 \\ #3 \\\\ #4 \\ #5}
  {#6}
}

\ekvcSplit\brtrgenoplow{
  p1={~},
  c={\btcfgb{trg=\varnothing,body=\Seq{\Nop}{body}} \btstepb  \btcfgb{trg=\varnothing, body=body}}
} {
  \inferrule*[Lab={[Et-Nop-Low]}]
  {#1}
  {#2}
}

\ekvcSplit\brtrgenophigh{
  p1={~},
  c={\btcfgb{cur=\fbcur,body=[\Seq{\Nop}{body}]} \btstepb  \btcfgb{cur=\fbcur,body=body}}
} {
  \inferrule*[Lab={[Et-Nop-High]}]
  {#1}
  {#2}
}

\ekvcSplit\brtrgeinstructionlow{
p1={\btcfgi{i=i} \btstepi {\btcfgi{h=\THspace', s=\TSspace', v=\TVspace', e=\Elive', i=i'}}},
c={\btcfgb{trg=\varnothing, body=\Seq{i}{\body}} \btstepb  \btcfgb{h=\THspace', s=\TSspace', v=\TVspace', e=\Elive', trg=\varnothing, body=\Seq{i'}{\body}}}
} {
  \inferrule*[Lab={[Et-Instruction-Low]}]
  {#1}
  {#2}
}

\ekvcSplit\brtrgeinstructionhigh{
p1={\btcfgi{i=i} \btstepi {\btcfgi{h=\THspace', s=\TSspace', v=\TVspace', e=\Elive', i=i'}}},
c={\btcfgb{cur=\fbcur, body=[\Seq{i}{\body}]} \btstepb  \btcfgb{h=\THspace', s=\TSspace', v=\TVspace', e=\Elive', cur=\fbcur, body=[\Seq{i'}{\body}]}}
} {
  \inferrule*[Lab={[Et-Instruction-High]}]
  {#1}
  {#2}
}

\ekvcSplit\brtrgebrconditionallow{
  p1={\configtwo{\TVspace}{e} \tstepe \bdonst},
  p2={\higher{\bdonst}{l}\in\{0,1\}},
  c={\btcfgb{body=\BrCond{\type{i1}{l} e}{\var{b_0}}{ \var{b_1}}} \btstepb \btcfgb{trg=\varnothing, body=\Br{\var{b_d}}}}
} {
  \inferrule*[Lab={[Et-Br-Conditional-Low]}]
  {#1 \\ #2}
  {#3}
}

\ekvcSplit\brtrgebrconditionalraise{
  p1={\configtwo{\TVspace}{e} \tstepe \bdonst},
  p2={\higher{\bdonst}{l}\in\{\high{0},\high{1}\}},
  p3={\trg = \mathrm{ipdom}(cur)},
  c={\btcfgb{body=\BrCond{\type{i1}{l} e}{\var{b_0}}{ \var{b_1}}} \btstepb \btcfgb{body=\high{\Br{\var{b_d}}}}}
} {
  \inferrule*[Lab={[Et-Br-Conditional-Raise]}]
  {#1 \\ #2 \\ #3}
  {#4}
}

\ekvcSplit\brtrgebrconditionalhigh{
  p1={\configtwo{\TVspace}{e} \tstepe \bdonst},
  c={\btcfgb{body={[\BrCond{\type{i1}{l} e}{\var{b_0}}{ \var{b_1}}]}} \btstepb \btcfgb{cur=\fbcur,body=[\Br{\var{b_d}}]}}
} {
  \inferrule*[Lab={[Et-Br-Conditional-High]}]
  {#1}
  {#2}
}

\ekvcSplit\brtrgebrunconditionalhigh{
  p1={\langle b, \mu \rangle: body \in \TFspace(b_c)},
  p2={\fbprev'=\fbcur},
  p3={b \neq \trg},
  c={\btcfgb{body=[\Br{\var{b}}]} \btstepb \btcfgb{prev=\fbprev', cur=[b], body=[body]}}
} {
  \inferrule*[Lab={[Et-Br-Unconditional-High]}]
  {#1 \\ #2 \\ #3}
  {#4}
}

\ekvcSplit\brtrgebrunconditionallower{
  p1={\langle b, \mu \rangle: body \in \TFspace(b_c)},
  p2={\fbprev'=\fbcur},
  p3={b = \trg},
  c={\btcfgb{cur=\fbcur, body=\high{\Br{\var{b}}}} \btstepb \btcfgb{prev=\fbprev', cur=b, trg=\varnothing, body=body}}
} {
  \inferrule*[Lab={[Et-Br-Unconditional-Lower]}]
  {#1 \\ #2 \\ #3}
  {#4}
}

\ekvcSplit\brtrgebrunconditionallow{
  p1={\langle b, \mu \rangle: body \in \TFspace(b_c)},
  p2={\prev'=\cur},
  c={\btcfgb{body=\Br{\var{b}}} \btstepb \btcfgb{prev=b_p', cur=b, trg=\varnothing, body=body}}
} {
  \inferrule*[Lab={[Et-Br-Unconditional-Low]}]
  {#1 \\ #2}
  {#3}
}

\ekvcSplit\brtrgeretlow{
  p1={\configtwo{\TVspace}{e} \btstepe \bkappa},
  p2={\TXspace = \btcfgxcall{v=\TVspace', prev=\bprev', cur=\cur', trg=\varnothing} :: \TXspace'},
  p3={\TVspace''=\TVspace'[x\mapsto{\higher{\bkappa_r}{l}}},
  c={\btcfgb{s=\TStspace::\TSspace,
             trg=\varnothing,
             body=\Ret{\type{\tau}{l} e}} \btstepb 
     \btcfgb{cur=\cur, trg=\varnothing, v=\TVspace'', x=\TXspace', prev=\bprev', cur=\cur', 
        body={body}}}
} {
  \inferrule*[Lab={[Et-Ret-Low]}]
  {#1 \\ #2 \\ #3}
  {#4}
}

\ekvcSplit\brtrgerethigh{
  p1={\configtwo{\fbTVspace}{e} \btstepe \bkappa},
  p2={\TXspace = \btcfgxcall{v=\bTVspace', prev=\bprev', cur=\fbcur', body=\fbbody} :: \TXspace'},
  p3={\bTVspace''=\bTVspace'[x\mapsto{\higher{\bkappa_r}{l}}]\\},
  c={\btcfgb{v=\fbTVspace, s=\fbTStspace::\TSspace, cur=\fbcur, trg=\bigstar, body=[\Ret{\type{\tau}{l} e}]} \btstepb
      \btcfgb{s=\TSspace, v=\bTVspace'', prev=\bprev', cur=\fbcur', x=\TXspace',
        body=[body]}}
} {
  \inferrule*[Lab={[Et-Ret-High]}]
  {#1 \\ #2 \\ #3}
  {#4}
}

\ekvcSplit\brtrgecalllow{
  p1={\TFspace(g) = \FDef{\type{\tau}{l_g} g^{\mu_g,oc_g}}{\ty{t=\tau_1,l=l_1}x_1, \dots, \ty{t=\tau_n,l=l_n}x_n}{\configtwo{b_1}{\mu_1}:\body_1; \cdots; \configtwo{b_m}{\mu_m}:\body_m}},
  p2={\configtwo{\TVspace}{e_i} \btstepe \bkappa_i},
  p3={\TVspace' = \{x_i : \higher{\bkappa_i}{l_i}\}},
  p4={\fresh{\TStspace, \mu_g}},
  p5={\TXspace' = \btcfgxcall{trg=\varnothing} :: \TXspace},
  c={\btcfgb{trg=\varnothing, body=\Seq{\FCall{\var{x}}
          {\type{\tau}{l_f}\fname{f}}
          {e_1, e_2, \dots, e_n}}
        {body}} \btstepb \\ \btcfgb{s=\TStspace::\TSspace, v=\TVspace', x=\TXspace', prev=\varnothing, cur=b_1,
        body=body_1}}
} {
  \inferrule*[Lab={[Et-Call-Low]}]
  {#1 \\ #2 \\ #3 \\ #4 \\\\ #5}
  {#6}
}

\ekvcSplit\brtrgecallhigh{
  p1={\TFspace(g) = \FDef{\type{\tau}{l_g} g^{\mu_g,oc_g}}{\ty{t=\tau_1,l=l_1}x_1, \dots, \ty{t=\tau_n,l=l_n}x_n}{\configtwo{b_1}{\mu_1}:\body_1; \cdots; \configtwo{b_m}{\mu_m}:\body_m}},
  p2={\configtwo{\TVspace}{e_i} \btstepe \bkappa_i},
  p3={\TVspace' = \{x_i : \higher{\bkappa_i}{l_i}\}},
  p4={\fresh{\fbTStspace, \mu_g}},
  p5={\TXspace' = \btcfgxcall{v=\bTVspace, cur=\fbcur, prev=\bprev} :: \TXspace},
  c={\btcfgb{cur=\fbcur, v=\bTVspace, body=[\Seq{\FCall{\var{x}}
          {\type{\tau}{l_f}\fname{f}}
          {e_1, e_2, \dots, e_n}}
        {body}]} \btstepb \\ \btcfgb{s=\TStspace::\TSspace, v=\fbTVspace', x=\TXspace', prev=\varnothing, cur=[b_1],
        trg=\bigstar,
        body=[body_1]}}
} {
  \inferrule*[Lab={[Et-Call-High]}]
  {#1 \\ #2 \\ #3 \\ #4 \\\\ #5}
  {#6}
}

\ekvcSplit\brtrgeprogram{
  p1={\forall i.\ func_i = \FDef{\type{\tau}{l} f_i^{\mu_f,oc_f}}{\paramlist}{\overline{block}}},
  p2={\TFspace_f=\{f_i :func_i \mid \forall i\}},
  p3={\TFspace_b= \{b : func_i \mid \forall i,\ \forall b \in func_i\}},
  p4={\TFspace = \TFspace_f \cup \TFspace_b},
  p5={func_0 = \FDef{\type{\itt}{\bot} \main^{0, \ocno}}{}{block^*_{\main}}},
  p6={\forall j.\  x_{j} \in \TVspace},
  p7={\mathit{body_{init}}=(\Seq{\FCall{\var{x}}{\type{\itt}{\public} \main}{x_1,\ldots, x_m}}{\Ret{\var{x}}})},
  p8={\btcfgb{s=\varnothing, x=\varnothing, prev=\varnothing, cur=\varnothing, trg=\varnothing, body=\mathit{body_{init}}}
    \bttrstepb
    \btcfgb{h=\Hspace', s=\varnothing, v=\TVspace', x=\varnothing, prev=\varnothing, cur=\varnothing, body=\Ret{\var{x}}}},
  c={\langle \THspace, \TVspace, func_0, \dots, func_N \rangle \tstepp \configtwo{\THspace'}{\TVspace'(x)}}}{
  \inferrule*[Lab={[Et-Program]}]
  {#1 \\ #5 \\ #8 \\ #7}
  {#9}
}

\ekvcSplit\brtrgepreserve{
  p1={\TVspace' = \TVspace[x\mapsto \TVspace(y)]},
  c={\tcfgi{i=\Preserve{\var{x}}{\type{\tau}{l} \var{y}}} \tstepi  \tcfgi{v=\TVspace'}}
} {
  \inferrule*[Lab={[Et-Preserve]}]
  {#1}
  {#2}
}

\ekvcSplit\brtrgeecreate{
  p1={\TFspace(b_c) = \FDef{\ty{} f}{\paramlist}{\ldots, \langle b, \mu\rangle: body, \ldots}},
  p2={\Elive' = \Elive \cup \{\mu\}},
  p3={\THspace' = \bigcup_{\mu \in h\cup\Elive'} \THspace_\mu},
  p4={(\THspace_\mu\subseteq\THspace) \lor (\forall m\in\dom{\THspace_\mu}: \THspace_\mu(m)=\varnothing)},
  c={\btcfgi{mode=0, i={\EnclaveCreate{\mu}{ \var{b}}}} \tstepi  \btcfgi{mode=0, h=\THspace', e=\Elive'}}
} {
  \inferrule*[Lab={[Et-Ecreate]}]
  {#1 \\ #2 \\ #3}
  {#5}
}

\ekvcSplit\brtrgeekill{
  p1={\mu\in \Elive},
  p2={\Elive' = \Elive \setminus \{\mu\}},
  p3={\THspace' = \THspace\setminus\THspace_{\mu}},
  c={\btcfgi{mode=0, i=\EnclaveKill{\mu}} \tstepi  \btcfgi{mode=0, h=\THspace', e=\Elive'}}
} {
  \inferrule*[Lab={[Et-Ekill]}]
  {#1 \\ #2 \\ #3}
  {#4}
}

\ekvcSplit\brtrgeeenter{
  p1={\TFspace(b_c) = \FDef{\ty{} f^{0, \ocno}}{\paramlist}{\ldots, \langle b, \mu\rangle: body', \ldots}},
  p2={\mu\in{\Elive}},
  p3={\fresh{\TStspace, \mu}},
  p4={\TXspace' = \tcfgxenclave{mode=0} :: \TXspace},
  c={\btcfgb{mode=0, trg=\varnothing, body=\Seq{\EnclaveEnter{\var{b}}}{body}} \tstepb
       \btcfgb{trg=\varnothing, x=\TXspace', s=\TStspace::\TSspace, prev=b_c, cur=b, mode=\mu,
        body=body'}}
} {
  \inferrule*[Lab={[Et-Eenter]}]
  {#1 \\ #2 \\ #3 \\ #4}
  {#5}
}

\ekvcSplit\brtrgeeexit{
  p1={\TFspace(b_c) = \FDef{\ty{} f^{\mu_f}}{\paramlist}{\ldots, \langle b_c', 0\rangle: body', \ldots}},
  p2={\mu \ne 0},
  p3={\TXspace = \tcfgxenclave{prev=b_p', cur=b_c', mode=0} :: \TXspace'},
  c={\btcfgb{trg=\varnothing, s=\TStspace::\TSspace, body=\EnclaveExit} \tstepb  \btcfgb{trg=\varnothing, x=\TXspace', prev=b_p', cur=b_c', mode=0, body=body}}
} {
  \inferrule*[Lab={[Et-Eexit]}]
  {#1 \\ #2 \\ #3}
  {#4}
}

\ekvcSplit\brtrgeoret{
  p1={\configtwo{\TVspace}{e} \tstepe \kappa},
  p1={\TXspace = \tcfgxcall{kind=ocall, v=\TVspace', prev=\prev', cur=\cur'} :: \TXspace'},
  p3={\TVspace'' = \TVspace'[x\mapsto\kappa]},
  c={\btcfgb{mode=0, trg=\varnothing, body=\ORet{\type{\tau}{l} \var{x}}} \tstepb
       \btcfgb{v=\TVspace'', trg=\varnothing, x=\TXspace', prev=\prev', cur=\cur',
        body=body}}
} {
  \inferrule*[Lab={[Et-ORet]}]
  {#1 \\ #2 \\ #3}
  {#4}
}

\ekvcSplit\brtrgeocall{
  p1={\TFspace(f) = \FDef{\type{\tau}l_f l_{pc} \fname{f}^{0,\ocyes}}{x_0, x_1, \dots, x_n}{block_0; block_1; \dots block_m}},
  p2={block_i = \langle b, id\rangle : body_i},
  p3={\TVspace' = \{x_i : c_i\}},
  p4={\configtwo{\TVspace}{e_i} \tstepe c_i},
  p5={\fresh{\TStspace, 0}},
  p6={\TXspace' = \btcfgxcall{v=\{\}, trg=\varnothing, prev=\varnothing, cur=\varnothing, mode=0, body=\code{oret\ }\type{\tau}{l}\var{x}} :: \tcfgxcall{kind=ocall} :: \TXspace},
  p7={\mode \ne 0},
  c={\btcfgb{trg=\varnothing, body=\Seq{\OCall{\var{x}}
        {\type{\tau}{l_{ret}}\fname{f}}
        {e_1, e_2, \dots, e_n}}
      {body}} \tstepb \\
      \btcfgb{trg=\varnothing, s=\TStspace::\TSspace, v=\TVspace', x=\TXspace', prev=\varnothing, cur=b0, mode=0,
      body=body_0}}
} {
  \inferrule*[Lab={[Et-OCall]}]
  {#1 \\ #2 \\ #3 \\ #4 \\ #5 \\ #6 \\ #7}
  {#8}
}

\ekvcSplit\brtrgtconstlow{
  p1={\tau\in\{\mathrm{Int_1}, \mathrm{Int_{32}}\}},
  c={\ttctxe{} \tc  c : \type{\tau}{\bot}}
} {
  \inferrule*[Lab={[Tt-Const-Low]}]
  {#1}
  {#2}
}

\ekvcSplit\brtrgtconsthigh{
  p1={\tau\in\{\mathrm{Int_1}, \mathrm{Int_{32}}\} \\ l \nflowsto \bot},
  c={\ttctxe{} \tc  \fbconst : \type{\tau}{\bot}}
} {
  \inferrule*[Lab={[Tt-Const-High]}]
  {#1}
  {#2}
}

\ekvcSplit\brtrgtvar{
  p1={\TTypeEnv(x) = \type{\tau}{l}},
  c={\ttctxe{} \tc  \var{x} : \type{\tau}{l}}
} {
  \inferrule*[Lab={[Tt-Var]}]
  {#1}
  {#2}
}

\ekvcSplit\brtrgtasgn{
  p1={\mu\neq0 \lor l\sqsubseteq\public},
  p2={\TTypeEnv(x) = \type{\tau}{l}},
  p3={\TTypeEnv \tc e : \type{\tau}{l}},
  p4={\pc \flowsto l},
  c={\ttctxi{} \tc  \Asgn{\var{x}}{\type{\tau}{l}e}}
} {
  \inferrule*[Lab={[Tt-Asgn]}]
  {#1 \\ #2 \\\\ #3 \\ #4}
  {#5}
}

\ekvcSplit\brtrgtadd{
  p1={\mu\neq0 \lor l\sqsubseteq\public},
  p2={\TTypeEnv(x) = \type{\tau}{l}},
  p3={\forall i: \TTypeEnv \vdash e_i : \type{\tau}{l}},
  p4={\pc\sqsubseteq l},
  c={\ttctxi{} \tc  \Add{\var{x}}{\type{\tau}{l} e_1}{e_2}}
} {
  \inferrule*[Lab={[Tt-Add]}]
  {#1 \\ #2 \\\\ #3 \\ #4}
  {#5}
}

\ekvcSplit\brtrgtalloca{
  p1={\mu\neq0 \lor l\sqsubseteq\public},
  p2={\TTypeEnv(x) = \type{\tgtptr{\tau}{\mu}}{l}},
  p3={\pc \flowsto l},
  c={\ttctxi{} \tc  \Alloca{\var{x}}{\type{\tau}{l}}}
} {
  \inferrule*[Lab={[Tt-Alloca]}]
  {#1 \\ #2 \\ #3}
  {#4}
}

\ekvcSplit\brtrgtmalloc{
  p1={\mu\neq0 \lor l\sqsubseteq\public},
  p2={\TTypeEnv(x) = \type{\tgtptr{\tau}{\mu_x}}{l}},
  p3={\pc \flowsto l},
  c={\ttctxi{} \tc  \TgtMalloc{\var{x}}{\type{\tau}{l}}{\mu_x}}
} {
  \inferrule*[Lab={[Tt-Malloc]}]
  {#1 \\ #2 \\ #3}
  {#4}
}

\ekvcSplit\brtrgtfree{
  p1={\mu\neq0 \lor l\sqsubseteq\public},
  p2={\TTypeEnv(x) = \type{\tgtptr{\tau}{\mu_x}}{l}},
  p3={\pc \flowsto l},
  c={\ttctxi{} \tc  \Free{\type{\tgtptr{\tau}{\mu_x}}{l}}{\var{x}}}
} {
  \inferrule*[Lab={[Tt-Free]}]
  {#1 \\ #2 \\ #3}
  {#4}
}

\ekvcSplit\brtrgtload{
  p1={\mu\neq0 \lor l_1\sqsubseteq\public},
  p2={\TTypeEnv(x)=\type{\tau}{l_1}},
  p3={\TTypeEnv \tc y : \type{\type{\ptr{\tau}}{\mu'}}{l_2}},
  p4={(\mu=\mu') \lor (\mu'=0)},
  p5={\type{\tau}{l_2} \sqsubseteq \type{\tau}{l_1}},
  p6={\pc \sqsubseteq l_1},
  c={\ttctxi{} \tc  \Load{\var{x}}{\type{\tau}{l_1}}{\type{\type{\ptr{\tau}}{\mu'}}{l_2}\var{y}}}
} {
  \inferrule*[Lab={[Tt-Load]}]
  {#1 \\ #2 \\ #3 \\\\ #4 \\ #5 \\ #6}
  {#7}
}

\ekvcSplit\brtrgtstore{
  p1={\mu\neq0 \lor l_2\sqsubseteq\public},
  p2={\TTypeEnv \tc e : \type{\tau}{l_1}},
  p3={\TTypeEnv \tc \var{y} : \type{\type{\ptr{\tau}}{\mu_y}}{l_2}},
  p4={(\mu=\mu_y) \lor (\mu_y=0)},
  p5={\type{\tau}{l_1} \sqsubseteq \type{\tau}{l_2}},
  p6={\pc \sqsubseteq l_2},
  c={\ttctxi{} \tc \Store{\type{\tau}{l_1}e}{\type{\type{\ptr{\tau}}{\mu_y}}{l_2}\var{y}}}
} {
  \inferrule*[Lab={[Tt-Store]}]
  {#1 \\ #2 \\ #3 \\\\ #4 \\ #5 \\ #6}
  {#7}
}

\ekvcSplit\brtrgtphi{
  p1={\mu\neq0 \lor l\sqsubseteq\public},
  p2={\TTypeEnv(x) = \type{\tau}{l}},
  p3={\forall i.\ \TTypeEnv \tc e_i : \type{\tau}{l} },
  p4={\forall i.\ \TBenv(b_i) \flowsto l},
  p5={\pc\flowsto l},
  c={\ttctxi{} \tc  \PhiNode{\var{x}}{\type{\tau}{l}}{[b_0, e_0] \dots [b_n,e_n]}}
} {
  \inferrule*[Lab={[Tt-Phi]}]
  {#1 \\ #2 \\ #3 \\\\ #4 \\ #5}
  {#6}
}

\ekvcSplit\brtrgtgep{
  p1={\mu\neq0 \lor l_1\sqsubseteq\public},
  p2={\TTypeEnv(x) = \TTypeEnv(y) = \type{\tgtptr{\tau_1}{\mu_y}}{l_1}},
  p3={l_2 \sqsubseteq l_1},
  p4={\tau_2 \in \{\mathrm{Int_{32}}, \mathrm{Int_1}\}},
  p5={\pc\sqsubseteq l_1},
  p6={(\mu_y = \mu \lor \mu_y = 0)},
  c={\ttctxi{} \tc  \GEP{\var{x}}{\type{\tau_1}{l_1}}{\type{\tgtptr{\tau_1}{\mu_y}}{l_1}\var{y}}{\type{\tau_2}{l_2} e}}
} {
  \inferrule*[Lab={[Tt-Gep]}]
  {#1 \\ #2 \\ #3 \\\\ #4 \\ #5}
  {#7}
}

\ekvcSplit\brtrgtnop{
  p1={~},
  c={\ttctxb{} \tc  \Nop}
} {
  \inferrule*[Lab={[Tt-Nop]}]
  {#1}
  {#2}
}

\ekvcSplit\brtrgtcall{
  p1={\TFenv(g) = \FHead{\type{\tau}{l_g} \fname{g}^{\mu_g,oc_g}}{\type{\tau_1}{l_1}x_1, \type{\tau_2}{l_2}x_2, \dots, \type{\tau_n}{l_n}x_n}{}},
  p2={\mu\neq0 \lor l_g\sqsubseteq\public},
  p3={\TTypeEnv(x) = \type{\tau}{l_g}},
  p4={l_x = l_g},
  p5={\pc \sqsubseteq l_g},
  p6={\forall i:\TTypeEnv \tc e_i : \type{\tau_i}{l_i}},
  p7={\mu_g = \mu},
  p8={oc_g = {\ocyes} \lor oc = {\ocno}},
  c={\ttctxb{} \tc  \FCall{\var{x}}{\type{\tau}{l_g} \fname{g}}{\type{\tau_1}{l_1}e_1, \type{\tau_2}{l_2}e_2\dots, \type{\tau_n}{l_n}e_n}}
} {
  \inferrule*[Lab={[Tt-Call]}]
  {#1 \\ #2 \\ #3 \\ #5 \\ #6 \\ #7 \\ #8}
  {#9}
}

\ekvcSplit\brtrgtret{
  p1={\TTypeEnv \tc e : \type{\tau}{l_f}},
  p2={l=l_f},
  p3={\mu \neq 0 \lor l \flowsto \bot},
  c={\ttctxb{} \tc  \Ret{\type{\tau}{l_f} e}}
} {
  \inferrule*[Lab={[Tt-Ret]}]
  {#1 \\ #3}
  {#4}
}

\ekvcSplit\brtrgtbody{
  p1={\ttctxb{} \tc  i},
  p2={\ttctxb{} \tc  body},
  c={\ttctxb{} \tc  \Seq{i}{body}}
} {
  \inferrule*[Lab={[Tt-Body]}]
  {#1 \\\\ #2}
  {#3}
}

\ekvcSplit\brtrgtbrunconditional{
  p1={\TBenv(b) = (\mu, l_b)},
  p2={\pc \sqsubseteq_{b} l_b},
  c={\ttctxb{} \tc  \Br{\var{b}}}
} {
  \inferrule*[Lab={[Tt-Br-Unconditional]}]
  {#1 \\ #2}
  {#3}
}

\ekvcSplit\brtrgtbrconditional{
  p1={\TTypeEnv \tc e : \type{\mathrm{Int_1}}{l}},
  p2={\forall i: \TBenv(b_i) = (\mu, l_i)},
  p3={\forall i: \pc \sqcup l\downarrow_{b_t} \sqsubseteq_{b_i} l_i},
  p4={b_t = \text{ ipdom }(b_c)},
  c={\ttctxb{} \tc  \BrCond{\type{\mathrm{Int_1}}{l}e}{\var{b_1}}{\var{b_2}}}
} {
  \inferrule*[Lab={[Tt-Br-Conditional]}]
  {#1 \\ #2 \\\\ #3 \\ #4}
  {#5}
}

\ekvcSplit\brtrgtfunction{
  p1={\TFenv(f) = \FHead{\type{\tau}{\lf} \fname{f}^{\mu_f, oc}}{\paramlist}{}\{\configtwo{b_1}{\mu_1} : \body_1; \configtwo{b_2}{\mu_2}; \cdots\}},
  p2={\forall i: \pc_i = \TBenv(b_i)},
  p3={\forall i: l_f \sqsubseteq \pc_i},
  p4={\forall i: \TTypeEnv_f, \mu_i, oc, \pc_i, b_i, l_f \tc body_i},
  p5={\mu_f=0 \lor \forall i: \mu_f = \mu_i},
  p6={\mu_f \neq 0 \lor \lf \flowsto \bot},
  p7={(oc = \code{\ocno}) \lor (oc=\ocyes \land \forall i:\mu_i=0)},
  c={\ttctxf{} \tc  \fname{f}}
} {
  \inferrule*[Lab={[Tt-Function]}]
  {#1 \\ #2 \\ #3 \\ #4 \\ #5 \\ #6 \\ #7}
  {#8}
}

\ekvcSplit\brtrgtprogram{
  p1={\TFspace(f_0)=\FDef{\type{\mathrm{\itt}}{\public} \fname{main^{0, \ocno}}}{\arglist}{blocks}},
  p2={\forall~i.~\TTypeEnv_{f_i} = \lfloor \TTypeEnv \rfloor_{f_i}},
  p3={\forall~i.~\TBenv_{f_i} = \lfloor \TBenv \rfloor_{f_i}},
  p4={\forall~i.~\TTypeEnv_{f_i}, \tc f_i},
  c={\ttctxp{} \tc  f_0, f_1, \dots,f_n}
} {
  \inferrule*[Lab={[Tt-Program]}]
  {#1 \\ #2 \\ #4}
  {#5}
}

\ekvcSplit\brtrgteenter{
  p1={\TBenv(b) = (\mu_b, l_b)},
  p2={\mu_b \ne 0},
  p3={\mu = 0},
  p4={\pc \sqsubseteq l_b},
  c={\ttctxb{} \tc  \EnclaveEnter{\var{b}}}
} {
  \inferrule*[Lab={[Tt-Eenter]}]
  {#1 \\ #2 \\ #3 \\ #4}
  {#5}
}

\ekvcSplit\brtrgteexit{
  p1={\mu \ne 0},
  c={\ttctxb{} \tc  \EnclaveExit}
} {
  \inferrule*[Lab={[Tt-Eexit]}]
  {#1}
  {#2}
}

\ekvcSplit\brtrgtecreate{
  p1={\TBenv(b) = (\mu_e, l_b)},
  p2={\mu = 0},
  p4={\mu_e\ne 0},
  p5={pc\sqsubseteq\public},
  c={\ttctxb{} \tc  \EnclaveCreate{ \mu_e}{\var{b}}}
} {
  \inferrule*[Lab={[Tt-Ecreate]}]
  {#1 \\ #2 \\ #3 \\ #4}
  {#5}
}

\ekvcSplit\brtrgtekill{
  p1={\mu = 0},
  p2={\mu_e\ne 0},
  p3={pc\sqsubseteq\public},
  c={\ttctxb{} \tc  \EnclaveKill{\mu_e}}
} {
  \inferrule*[Lab={[Tt-Ekill]}]
  {#1 \\ #2 \\ #3}
  {#4}
}

\ekvcSplit\brtrgtocall{
  p1={\TFenv(f) = \FHead{\type{\tau}{l_g} \fname{g}^{\mu_g,oc_g}}{\type{\tau_1}{l_1}x_1, \type{\tau_2}{l_2}x_2, \dots, \type{\tau_n}{l_n}x_n}{}},
  p2={\TTypeEnv(x) = \type{\tau}{l_x}},
  p3={l_x = l_g},
  p4={\pc \sqsubseteq l_g},
  p5={\forall i:\TTypeEnv \tc e_i : \type{\tau_i}{l_i}},
  p6={\mu_g = 0 \\ \mu \neq 0},
  p7={oc_g=oc=\ocyes},
  c={\ttctxb{} \tc  \OCall{\var{x}}{\type{\tau}{l_g} \fname{g}}{\type{\tau_1}{l_1}e_1, \type{\tau_2}{l_2}e_2\dots, \type{\tau_n}{l_n}e_n}}
} {
  \inferrule*[Lab={[Tt-OCall]}]
  {#1 \\ #2 \\ #3 \\ #4 \\ #5 \\ #6 \\ #7}
  {#8}
}

\ekvcSplit\brtrgtoret{
  p1={\TTypeEnv \tc e : \type{\tau}{l}},
  p2={\type{\tau}{l} \flowsto \type{\tau}{\lf}},
  p3={\mu=0},
  p4={l\flowsto\bot},
  c={\ttctxb{} \tc \ORet{\type{\tau}{l} e}}
} {
  \inferrule*[Lab={[Tt-ORet]}]
  {#1 \\ #2 \\ #3 \\ #4}
  {#5}
}

\ekvcSplit\brtrgtpreserve{
  p1={\mu = 0},
  p2={\TTypeEnv(x) = \type{\tau}{l}},
  p3={\TTypeEnv \tc \var{y} : \type{\tau}{l}},
  p4={\pc\sqsubseteq l},
  c={\ttctxb{} \tc  \Preserve{\var{x}}{\type{\tau}{l} \var{y}}}
} {
  \inferrule*[Lab={[Tt-Preserve]}]
  {#1 \\ #2 \\ #3 \\ #4}
  {#5}
}

\ekvcSplit\brtrgtmemoryconfig{
  p1={\forall m\in{\Loc}.\ (\TSspace\cup\THspace)(m) = \varnothing \lor \TMenv \tc (\TSspace\cup\THspace)(m) : \TMenv(m)},
  c={\TMenv \tc  \configtwo{\TSspace}{\THspace}}}{
  \inferrule*[Lab={[Tt-Memory-Config]}]
  {#1}
  {#2}
}

\ekvcSplit\brtrgtvariables{
  p1={\forall x\in\Var.\ \TMenv,\TTypeEnv \tc \TVspace(x) : \TTypeEnv(x)},
  c={\TMenv, \TTypeEnv \tc  \TVspace}}{
  \inferrule*[Lab={[Tt-Variables]}]
  {#1}
  {#2}
}

\ekvcSplit\brtrgtcallstackbase{
  p1={\TXspace = \varnothing},
  c={\TMenv, \mu_f, \lf \tc  \TXspace}}{
  \inferrule*[Lab={[Tt-Callstack-Base]}]
  {#1}
  {#2}
}

\ekvcSplit\brtrgtelive{
  p1={\forall\mu:(\mu\in \Elive \Leftrightarrow \THspace_\mu\subseteq \THspace)},
  c={\tc \configtwo{\Elive}{\THspace}}}{
  \inferrule*[Lab={[Tt-E-live]}]
  {#1}
  {#2}
}

\ekvcSplit\brtrgtfunctions{
  p1={\forall f\in \Fcn.\ ~ \tc \TFspace(f)},
  c={ \tc  \TFspace}}{
  \inferrule*[Lab={[Tt-Functions]}]
  {#1}
  {#2}
}

\ekvcSplit\brtrgtresultconfig{
  p1={\TMenv \tc \configtwo{\varnothing}{\THspace}},
  p2={\tc c : \type{\itt}{\public}},
  c={\TMenv \tc \configtwo{\THspace}{c}}
} {
  \inferrule*[Lab={[Tt-Result-Config]}]
  {#1 \\\\ #2}
  {#3}
}

\ekvcSplit\brtrgtexpressionconfig{
  p1={\TMenv, \TTypeEnv \tc \TVspace},
  p2={\TMenv, \TTypeEnv \tc e : \ty{}},
  c={\TMenv,\TTypeEnv \tc  \tcfge{}}
} {
  \inferrule*[Lab={[Tt-Expression-Config]}]
  {#1 \\ #2}
  {#3}
}

\ekvcSplit\brtrgtinstructionconfig{
  p1={\TMenv \tc \configtwo{\TSspace}{\THspace}},
  p2={\TMenv, \TTypeEnv \tc \TVspace},
  p3={\trgtypingcontextcom \tc{i}},
  p4={\mu = \mu_{b_c}},
  p5={\tc \configtwo{\Elive}{\THspace}},i
  c={\trgtypingcontextconfig \tc  \btcfgi{mode=\mu_{b_c}, i=i}}}{
  \inferrule*[Lab={[Tt-Instruction-Config]}]
  {#1 \\ #2 \\ #3 \\ #4}
  {#6}
}

\ekvcSplit\brtrgtbodyconfiglow{
  p1={\TMenv \tc \configtwo{\TSspace}{\THspace}},
  p2={\TMenv, \TTypeEnv \tc \TVspace},
  p3={\TMenv, \mu, \lf \tc \TXspace},
  p4={ \tc \TFspace},
  p5={\tc \configtwo{\Elive}{\THspace}},i
  p6={\trgtypingcontextcom \tc{body}},
  p7={\mu = \mu_{b_c}},
  c={\trgtypingcontextconfig \tc  \btcfgb{mode=\mu_{b_c}, trg=\varnothing, body=body}}}{
  \inferrule*[Lab={[Tt-Body-Config-Low]}]
  {#1 \\ #2 \\ #3 \\ #4 \\ #6 \\ #7}
  {#8}
}

\ekvcSplit\brtrgtbodyconfighigh{
  p1={\TMenv \tc \configtwo{\TSspace}{\THspace}},
  p2={\TMenv, \TTypeEnv \tc \TVspace},
  p3={\TMenv, \mu, \lf \tc \TXspace},
  p4={ \tc \TFspace},
  p5={\tc \configtwo{\Elive}{\THspace}},i
  p6={\trgtypingcontextcom \tc{body}},
  p7={\mu = \mu_{b_c}},
  p8={l \nflowsto \bot \\ \declat{l}{\trg}\flowsto\pc' \\ \pc\flowsto\pc'},
  c={\trgtypingcontextconfig \tc  \btcfgb{cur=\fbcur, mode=\mu_{b_c}, body=[body]}}}{
  \inferrule*[Lab={[Tt-Body-Config-High]}]
  {#1 \\ #2 \\ #3 \\ #4 \\ #6 \\ #7 \\ #8}
  {#9}
}

\ekvcSplit\brtrgtprogramconfig{
  p1={ \tc p},
  p2={\TMenv \tc \THspace},
  p3={\TMenv, \TTypeEnv \tc \TVspace},
  c={\TMenv, \TTypeEnv \tc  \tcfgp{}}}{
  \inferrule*[Lab=Tt-Program-Config]
  {#1 \\ #2 \\ #3}
  {#4}
}

\ekvcSplit\brtrgeqvconst{
  p1={\kappa_1 = \kappa_2 \lor (\isbr{\bkappa_1} \land \isbr{\bkappa_2})},
  c={\bkappa_1 \lequiv \bkappa_2}}{
  \inferrule*[Lab={[TEqv-Const]}]
  {#1}
  {#2}
}

\ekvcSplit\brtrgeqvvar{
  c={x\lequiv x}
} {
  \inferrule*[Lab={[TEqv-Var]}]
  {~}
  {#1}
}

\ekvcSplit\brtrgeqvblocklabels{
  p1={\blab_1 = \blab_2 \lor (\isbr{\blab_1} \land \isbr{\blab_2})},
  c={\blab_1 \lequiv \blab_2}}{
  \inferrule*[Lab={[TEqv-Block-Labels]}]
  {#1}
  {#2}
}

\ekvcSplit\brtrgeqvmemories{
  p1={\bkappa_1 \lequiv \bkappa_2 \\\\
      \lor\ (\bkappa_1 = \varnothing \land \isbr{\bkappa_2})\\\\
      \lor\ (\isbr{\bkappa_1} \land\ \bkappa_2 = \varnothing)},
  c={\bkappa_1 \mequiv \bkappa_2}}{
  \inferrule*[Lab={[TEqv-Memories]}]
  {#1}
  {#2}
}

\ekvcSplit\brtrgeqvbodylow{
  p1={\bbody_1 = \Seq{i_1}{\body_{1,t}} \\},
  p2={\bbody_2 = \Seq{i_2}{\body_{2,t}} \\\\},
  p3={i_1 \lequiv i_2 \\},
  p4={\body_{1,t} = \body_{2,t}},
  c={\bbody_1 \lequiv \bbody_2}
} {
  \inferrule*[Lab={[TEqv-Body-Low]}]
  {#1 #2 #3 #4}
  {#5}
}

\ekvcSplit\brtrgeqvbodyhigh{
  p1={~},
  c={\fbbody_1 \lequiv \fbbody_2}
} {
  \inferrule*[Lab={[TEqv-Body-High]}]
  {#1}
  {#2}
}

\ekvcSplit\brtrgeqvasgn{
  p1={e_1 \lequiv e_2},
  c={\biglequiv{\Asgn{\var{x}}{\type{\tau}{l} e_1}}
               {\Asgn{\var{x}}{\type{\tau}{l} e_2}}}
} {
  \inferrule*[Lab={[TEqv-Asgn]}]
  {#1}
  {#2}
}

\ekvcSplit\brtrgeqvadd{
  p1={e_1 \lequiv e_2},
  p2={e_3 \lequiv e_4},
  c={\biglequiv{\Add{\var{x}}{\type{\tau}{l} e_1}{e_3}}
               {\Add{\var{x}}{\type{\tau}{l} e_2}{e_4}}}
} {
  \inferrule*[Lab={[TEqv-Add]}]
  {#1 \\ #2}
  {#3}
}

\ekvcSplit\brtrgeqvload{
  p1={y_1 \lequiv y_2},
  c={\biglequiv{\Load{\var{x}}{\type{\tau}{l_1}}{\type{\ptr{\tau}}{l_2}y_1}}
               {\Load{\var{x}}{\type{\tau}{l_1}}{\type{\ptr{\tau}}{l_2}y_2}}}
} {
  \inferrule*[Lab={[TEqv-Load]}]
  {#1}
  {#2}
}

\ekvcSplit\brtrgeqvstore{
  p1={e_1 \lequiv e_2 \\ x_1\lequiv x_2},
  c={\biglequiv{\Store{\type{\tau}{l_1} e_1}{\type{\ptr{\tau}}{l_2} x_1}}
               {\Store{\type{\tau}{l_1} e_2}{\type{\ptr{\tau}}{l_2} x_2}}}}{
  \inferrule*[Lab={[TEqv-Store]}]
  {#1}
  {#2}
}

\ekvcSplit\brtrgeqvmalloc{
  p1={~},
  c={\biglequiv{\Malloc{\var{x}}{\type{\tau}{l}}}{\Malloc{\var{x}}{\type{\tau}{l}}}}
} {
  \inferrule*[Lab={[TEqv-Malloc]}]
  {#1}
  {#2}
}

\ekvcSplit\brtrgeqvalloca{
  p1={~},
  c={\biglequiv{\Alloca{\var{x}}{\type{\tau}{l}}}{\Alloca{\var{x}}{\type{\tau}{l}}}}
} {
  \inferrule*[Lab={[TEqv-Alloca]}]
  {#1}
  {#2}
}

\ekvcSplit\brtrgeqvfree{
  p1={y_1\lequiv y_2},
  c={\biglequiv{\Free{\type{\tau*}{l}y_1}}{\Free{\type{\tau*}{l}y_2}}}
} {
  \inferrule*[Lab={[TEqv-Free]}]
  {#1}
  {#2}
}

\ekvcSplit\brtrgeqvphi{
  p1={\forall i.\ \blab_{1,i}\lequiv\blab_{2,i} \land e_{1,i} \lequiv e_{2,1}},
  c={\biglequiv{\PhiNode{\var{x}}{\type{\tau}{l}}{[b_{1,1}, e_{1,1}] \dots [b_{1,n},e_{1,n}]}}
               {\PhiNode{\var{x}}{\type{\tau}{l}}{[b_{2,1}, e_{2,1}] \dots [b_{2,n},e_{2,n}]}}}
} {
  \inferrule*[Lab={[TEqv-Phi]}]
  {#1}
  {#2}
}

\ekvcSplit\brtrgeqvgep{
  p1={y_1 \lequiv y_2},
  p2={e_1 \lequiv e_2},
  c={\biglequiv{\GEP{\var{x}}{\ty{t=\tau_1,l=l_1}}{\ty{t=\tau*_1,l=l_1}y_1}{\ty{t=\tau_2,l=l_2}e_1}}
               {\GEP{\var{x}}{\ty{t=\tau_1,l=l_1}}{\ty{t=\tau*_1,l=l_1}y_2}{\ty{t=\tau_2,l=l_2}e_2}}}
} {
  \inferrule*[Lab={[TEqv-GEP]}]
  {#1 \\ #2}
  {#3}
}

\ekvcSplit\brtrgeqvnop{
  p1={~},
  c={\Nop\lequiv\Nop}
} {
  \inferrule*[Lab={[TEqv-Nop]}]
  {#1}
  {#2}
}

\ekvcSplit\brtrgeqvcall{
  p1={\forall i.\ e_{1,i} \lequiv e_{2,i}},
  c={\biglequiv{\FCall{\var{x}}{{\ty{l=\lg}} g}{\ty{t=\tau_1,l=1_1}e_{1,1}, \dots, \ty{t=\tau_n,l=l_n}e_{1,n}}}
               {\FCall{\var{x}}{{\ty{l=\lg}} g}{\ty{t=\tau_1,l=1_1}e_{2,1}, \dots, \ty{t=\tau_n,l=l_n}e_{1,n}}}}
} {
  \inferrule*[Lab={[TEqv-Call]}]
  {#1}
  {#2}
}

\ekvcSplit\brtrgeqvret{
  p1={e_1 \lequiv e_2},
  c={\biglequiv{\Ret{\type{\tau}{l} e_1}}{\Ret{\type{\tau}{l} e_2}}}
} {
  \inferrule*[Lab={[TEqv-Ret]}]
  {#1}
  {#2}
}

\ekvcSplit\brtrgeqvbrconditional{
  p1={e_1 \lequiv e_2},
  c={\biglequiv{\BrCond{\type{\ione}{l} e}{\var{b_0}}{\var{b_1}}}
               {\BrCond{\type{\ione}{l} e}{\var{b_0}}{\var{b_1}}}}
} {
  \inferrule*[Lab={[TEqv-Br-Conditional]}]
  {#1}
  {#2}
}

\ekvcSplit\brtrgeqvbrunconditional{
  p1={~},
  c={\Br{\var{b}}\lequiv\Br{\var{b}}}
} {
  \inferrule*[Lab={[TEqv-Br-Unconditional]}]
  {#1}
  {#2}
}

\ekvcSplit\brtrgeqvheap{
  p1={\forall m\in \Loc.\ \THspace_1(m) \mequiv \THspace_2(m)},
  c={\THspace_1 \lequiv \THspace_2}}{
  \inferrule*[Lab={[TEqv-Heap]}]
  {#1}
  {#2}
}

\ekvcSplit\brtrgeqvstacklow{
  p1={\forall m\in \Loc.\ \Stspace_1(m) \mequiv \Stspace_2(m)},
  p2={\Sigma_1 \lequiv \Sigma_2},
  c={\Stspace_1::\TSspace_1 \lequiv \Stspace_2::\TSspace_2}}{
  \inferrule*[Lab={[TEqv-Stack-Low]}]
  {#1 \\ #2}
  {#3}
}

\ekvcSplit\brtrgeqvstackhigh{
  p1={\TSspace_1 \lequiv \bStspace_2::\TSspace_2},
  c={\fbStspace_1::\TSspace_1 \lequiv \bStspace_2::\TSspace_2}}{
  \inferrule*[Lab={[TEqv-Stack-High]}]
  {#1}
  {#2}
}

\ekvcSplit\brtrgeqvvariables{
  p1={\forall x\in\Var.\ \TVspace_1(x) \mequiv \TVspace_2(x)},
  c={\TVspace_1 \lequiv \TVspace_2}}{
  \inferrule*[Lab={[TEqv-Variables]}]
  {#1}
  {#2}
}

\ekvcSplit\brtrgeqvcallstacklow{
  p1={\TXspace_1=\btcfgxcall{v=\TVspace'_1,prev=\bprev_1,cur=\cur,body=\body}::\TXspace_{1,t}},
  p2={\TXspace_2=\btcfgxcall{v=\TVspace'_2,prev=\bprev_2,cur=\cur,body=\body}::\TXspace_{2,t}},
  p3={\TVspace_1\lequiv\TVspace_2},
  p4={\bprev_1\lequiv\bprev_2},
  p5={\tcfgvx{v=\TVspace'_1,x=\TXspace_{1,t}}\lequiv\tcfgvx{v=\TVspace'_2,x=\TXspace_{2,t}}},
  c={\tcfgvx{x=\TXspace_1,v=\TVspace_1}\lequiv\tcfgvx{x=\TXspace_2,v=\TVspace_2}}
} {
  \inferrule*[Lab={[TEqv-Callstack-Low]}]
  {#1 \\ #2 \\\\ #3 \\ #4 \\ #5}
  {#6}
}

\ekvcSplit\brtrgeqvcallstackhigh{
  p1={\TXspace_1=\btcfgxcall{v=\bVspace'_1,prev=\bprev_1,cur=\fbcur,body=\fbbody}::\TXspace_{1,t}},
  p2={\tcfgvx{v=\bVspace'_1,x=\TXspace_{1,t}}\lequiv\tcfgvx{v=\bVspace_2,x=\TXspace_2}},
  c={\tcfgvx{x=\TXspace_1,v=\fbVspace_1}\lequiv\tcfgvx{x=\TXspace_2,v=\bVspace_2}}
} {
  \inferrule*[Lab={[TEqv-Callstack-High]}]
  {#1 \\ #2}
  {#3}
}

\ekvcSplit\brtrgeqvcallstackbase{
  p1={\TVspace_1\lequiv\TVspace_2},
  c={\tcfgvx{x=\varnothing,v=\TVspace_1}\lequiv\tcfgvx{x=\varnothing,v=\TVspace_2}}
} {
  \inferrule*[Lab={[TEqv-Callstack-Base]}]
  {#1}
  {#2}
}

\ekvcSplit\brtrgeqvexpressionconfig{
  p1={\TVspace_1 \lequiv \TVspace_2},
  p2={e_1 \lequiv e_2},
  c={\btcfge{v=\TVspace_1, e=e_1} \lequiv \btcfge{v=\TVspace_2, e=e_2}}
} {
  \inferrule*[Lab={[TEqv-Expression-Config]}]
  {#1 \\ #2}
  {#3}
}

\ekvcSplit\brtrgeqvinstructionconfig{
  p1={\THspace_1 \lequiv \THspace_2},
  p2={\TSspace_1 \lequiv \TSspace_2},
  p3={\TVspace_1 \lequiv \TVspace_2},
  p4={\bprev_1 \lequiv \bprev_2},
  p5={i_1 \lequiv i_2},
  c={\btcfgi{h=\THspace_1, s=\TSspace_1, v=\TVspace_1, prev=\bprev_1, i=i_1} \lequiv
     \btcfgi{h=\THspace_2, s=\TSspace_2, v=\TVspace_2, prev=\bprev_2, i=i_2}}
}
{
  \inferrule*[Lab={[TEqv-Instruction-Config]}]
  {#1 \\ #2 \\ #3 \\\\ #4 \\ #5}
  {#6}
}

\ekvcSplit\brtrgeqvbodyconfiglow{
  p1={\THspace_1 \lequiv \THspace_2},
  p2={\TSspace_1 \lequiv \TSspace_2},
  p3={\tcfgvx{v=\TVspace_1,x=\TXspace_1} \lequiv \tcfgvx{v=\TVspace_2,x=\TXspace_2}},
  p4={\bprev_1 \lequiv \bprev_2},
  p5={body_1 \lequiv body_2},
  p6={\isub{body_1}},
  c={\biglequiv{\btcfgb{h=\THspace_1, s=\TSspace_1, v=\TVspace_1, x=\TXspace_1, prev=\bprev_1, cur=\cur, trg=\varnothing, body=body_1}}
               {\btcfgb{h=\THspace_2, s=\TSspace_2, v=\TVspace_2, x=\TXspace_2, prev=\bprev_2, cur=\cur, trg=\varnothing, body=body_2}}}
} {
  \inferrule*[Lab={[TEqv-Body-Config-Low]}]
  {#1 \\ #2 \\ #3 \\\\ #4 \\ #5 \\ #6}
  {#7}
}

\ekvcSplit\brtrgeqvbodyconfighigh{
  p1={\THspace_1 \lequiv \THspace_2},
  p2={\TSspace_1 \lequiv \TSspace_2},
  p3={\TXspace_1\lequiv\TXspace_2},
  p4={\TVspace_1\concat\flatten{\TXspace_1} \lequiv \TVspace_2\concat\flatten{\TXspace_2}},
  c={\biglequiv{\btcfgb{h=\THspace_1, s=\TSspace_1, v=\TVspace_1, x=\TXspace_1, prev=\bprev_1, cur=[\cur_1], body=[body_1]}}
               {\btcfgb{h=\THspace_2, s=\TSspace_2, v=\TVspace_2, x=\TXspace_2, prev=\bprev_2, cur=[\cur_2], body=[body_2]}}}
} {
  \inferrule*[Lab={[TEqv-Body-Config-High]}]
  {#1 \\ #2 \\ #3 \\\\ #4}
  {#5}
}

\ekvcSplit\trgeqvidentical{
  p1={~},
  c={instr \lequiv instr}
} {
  \inferrule*[Lab=Eqv-Identical]
  {#1}
  {#2}
}

\ekvcSplit\trgeqvbracketed{
  p1={~},
  c={[instr] \lequiv [instr']}
} {
  \inferrule*[Lab=Eqv-Bracketed]
  {#1}
  {#2}
}

\ekvcSplit\trgeqvasgn{
  p1={l\nsqsubseteq L \lor e = e'},
  c={\Asgn{\var{x}}{\type{\tau}{l} e} \lequiv \Asgn{\var{x}}{\type{\tau}{l}  e'}}
} {
  \inferrule*[Lab=Eqv-Asgn]
  {#1}
  {#2}
}

\ekvcSplit\trgeqvadd{
  p1={l\nsqsubseteq L \lor (e1 = e1' \land e2=e2')},
  c={\Add{\var{x}}{\type{\tau}{l}  e_1}{e_2} \lequiv \Add{\var{x}}{\type{\tau}{l}  e_1'}{e_2'}}
} {
  \inferrule*[Lab=Eqv-Add]
  {#1}
  {#2}
}

\ekvcSplit\trgeqvload{
  p1={l_1\nsqsubseteq L \lor (y_1 = y_2)},
  c={\Load{\var{x}}{\type{\tau}{l_1}}{\type{\ptr{\tau}}{l_2}\var{y_1}} \lequiv
      \Load{\var{x}}{\type{\tau}{l_1}}{\type{\ptr{\tau}}{l_2}\var{y_2}}}
} {
  \inferrule*[Lab=Eqv-Load]
  {#1}
  {#2}
}

\ekvcSplit\trgeqvstore{
  p1={l_2\nsqsubseteq L \lor (e_1 = e_2)},
  c={\Store{\type{\tau}{l_1} e_1}{\type{\ptr{\tau}}{l_2} \var{x}} \lequiv
      \Store{\type{\tau}{l_1} e_2}{\type{\ptr{\tau}}{l_2} \var{x}}}
} {
  \inferrule*[Lab=Eqv-Store]
  {#1}
  {#2}
}

\ekvcSplit\trgeqvseq{
  p1={instr_1 \lequiv instr_1'},
  p2={instr_2 \lequiv instr_2'},
  c={\Seq{instr_1}{instr_2} \lequiv \Seq{instr_1}{instr_2}}
} {
  \inferrule*[Lab=Eqv-Seq]
  {#1 \\ #2}
  {#3}
}

\ekvcSplit\trgeqvphi{
  p1={l\nsqsubseteq L \lor (\forall i\in[0,n]:  b_i=b_i' \land e_i = e_i')},
  c={\PhiNode{\var{x}}{\type{\tau}{l}}{[b_0, e_0] \dots [b_n,e_n]} \lequiv
      \PhiNode{\var{x}}{\type{\tau}{l}}{[b_0', e_0'] \dots [b_n',e_n']}}
} {
  \inferrule*[Lab=Eqv-Phi]
  {#1}
  {#2}
}

\ekvcSplit\trgeqvbrconditional{
  p1={l \nsqsubseteq L \lor (e=e' \land b_0=b_0' \land b_1=b_1')},
  c={\BrCond{\type{i1}{l} e}{\var{b_0}}{ \var{b_1}} \lequiv
      \BrCond{\type{i1}{l} e'}{\var{b_0'}}{ \var{b_1'}}}
} {
  \inferrule*[Lab=Eqv-Br-Conditional]
  {#1}
  {#2}
}

\ekvcSplit\trgeqvbrunconditional{
  p1={pc(b_1)\nsqsubseteq L},
  p2={pc(b_2)\nsqsubseteq L},
  c={\Br{ \var{b_1}} \lequiv\Br{ \var{b_2}}}
} {
  \inferrule*[Lab=Eqv-Br-Unconditional]
  {#1 \\ #2}
  {#3}
}

\ekvcSplit\trgeqvblock{
  p1={(pc(b_1) \nsqsubseteq L \land pc(b_2) \nsqsubseteq L) \lor b_1=b_2},
  c={b_1 \lequiv b_2}
} {
  \inferrule*[Lab=Eqv-Block]
  {#1}
  {#2}
}

\ekvcSplit\trgeqvinstctxt{
  p1={\THspace_1 \lequiv \THspace_2},
  p2={\TSspace_1 \lequiv \TSspace_2},
  p3={\TVspace_1 \lequiv \TVspace_2},
  p4={prev_i = prev_{i,t} :: prev_{b}},
  p5={\color{red}{\cancel{prev_1 \lequiv prev_2}}},
  p6={cur_1=cur_2},
  p7={instr_1 \approx_l instr_2},
  p8={unbracketed(instr_1)},
  c={\confignine{\TFspace}{\Fspace}{\THspace_1}{\TSspace_1}{\TVspace_1}{prev_1}{cur_1}{\varnothing}{instr_1} \lequiv
      \confignine{\TFspace}{\Fspace}{\THspace_2}{\TSspace_2}{\TVspace_2}{prev_2}{cur_2}{\varnothing}{instr_2}}
} {
  \inferrule*[Lab=Eqv-Inst-Ctxt]
  {#1 \\ #2 \\ #3 \\ #4 \\ #5 \\ #6 \\ #7 \\ #8}
  {#9}
}

\ekvcSplit\trgeqvinstctxtbracketed{
  p1={\THspace_1 \lequiv \THspace_2},
  p2={\TSspace_1 \lequiv \TSspace_2},
  p3={\TVspace_1 \lequiv \TVspace_2},
  p4={trg_1 = trg_2},
  c={\confignine{\TFspace}{\Fspace}{\THspace_1}{\TSspace_1}{\TVspace_1}{prev_1}{cur_1}{trg_1}{\high{instr_1}} \lequiv
      \confignine{\TFspace}{\Fspace}{\THspace_2}{\TSspace_2}{\TVspace_2}{prev_2}{cur_2}{trg_2}{\high{instr_2}}}
} {
  \inferrule*[Lab=Eqv-Inst-Ctxt-Bracketed]
  {#1 \\ #2 \\ #3 \\ #4}
  {#5}
}

\AtBeginDocument{%
  }

\lstdefinestyle{customc}{
  belowcaptionskip=1\baselineskip,
  breaklines=true,
  language=C,
  showstringspaces=false,
  basicstyle=\footnotesize\ttfamily,
  keywordstyle=\bfseries\color{green!40!black},
  commentstyle=\itshape\color{purple!40!black},
  identifierstyle=\color{blue},
  stringstyle=\color{orange},
  escapeinside={(*@}{@*)},
  postbreak=\mbox{\textcolor{red}{$\hookrightarrow$}\space}
}

\definecolor{codegreen}{rgb}{0,0.6,0}
\definecolor{codegray}{rgb}{0.5,0.5,0.5}
\definecolor{codepurple}{rgb}{0.58,0,0.82}
\definecolor{backcolour}{rgb}{0.95,0.95,0.92}
\definecolor{highlightcolor}{rgb}{0.90,0.90,0.86}
\definecolor{enclavecolor}{rgb}{0,0.6,0}

\lstdefinestyle{mystyle}{
    frame=single,
    commentstyle=\color{purple!30!blue},
    keywordstyle=\color{magenta},
    numberstyle=\tiny\color{codegray},
    stringstyle=\color{codepurple},
    basicstyle=\footnotesize\ttfamily,
    breakatwhitespace=false,         
    breaklines=true,                 
    captionpos=b,                    
    keepspaces=true,                 
    numbers=left,                    
    numbersep=5pt,                  
    showspaces=false,                
    showstringspaces=false,
    showtabs=false,                  
      tabsize=2,
    postbreak=\mbox{\textcolor{red}{$\hookrightarrow$}\space}
}

\makeatletter
\let\@authorsaddresses\@empty   
\makeatother

\begin{document}

\ifappendix
\title[Appendix]{{\color{teal}{Appendix:}} Type-Directed, Secure-by-Construction Enclave Partitioning for LLVM }

\else

\iftechreport
\title[Technical Report]{{\color{teal}{Technical Report:}} Type-Directed, Secure-by-Construction Enclave Partitioning for LLVM }

\else
\title{Type-Directed, Secure-by-Construction Enclave Partitioning for LLVM}
\fi

\fi

\author{Wesley B Nuzzo}
\email{Wesley\_Nuzzo@student.uml.edu}
\affiliation{%
  \institution{University of Massachusetts Lowell}
  \country{USA}
}

\author{Samuel Dodson}
\email{Samuel\_Dodson@student.uml.edu}
\affiliation{%
  \institution{University of Massachusetts Lowell}
  \country{USA}
}

\author{Benjamin Houle}
\email{benjamin\_houle@student.uml.edu}
\affiliation{%
  \institution{University of Massachusetts Lowell}
  \country{USA}
}

\author{Tarakaram Gollamudi}
\email{tarakaram\_gollamudi@uml.edu}
\affiliation{%
  \institution{University of Massachusetts Lowell}
  \country{USA}
}


\author{Anitha Gollamudi}
\email{anitha\_gollamudi@uml.edu}
\affiliation{%
  \institution{University of Massachusetts Lowell}
  \country{USA}
}

\setcopyright{none}
\settopmatter{printacmref=false}
\renewcommand\footnotetextcopyrightpermission[1]{}

\iftechreport
  \maketitle
  \setcounter{tocdepth}{3}
  \tableofcontents
  \input{llvm-enclaves-tr}
  \bibliography{mybib}
  \bibliographystyle{abbrvnat}
\else

\ifappendix

\maketitle

\else
\begin{abstract}

  Trusted Execution Environments (TEEs) provide hardware-supported
isolation through enclaves that protect code and data independently of
software abstractions. However, TEEs alone cannot enforce information
flow security. This problem is further aggravated in LLVM-like
low-level languages that allow unrestricted pointer manipulation and
unstructured control flow. Moreover, using TEEs effectively typically
requires manually partitioning applications into enclave and
non-enclave components, a process that is labor-intensive,
error-prone, and lacks fine-grained control.

We address these challenges with a three-step approach. First, we
formalize \srclang, an enclave-oblivious calculus based on LLVM IR,
equipped with a novel permissive type system that enforces security
against low-level attackers. To obtain meaningful guarantees,
\srclang combines information-flow control with security-aware
coarse-grained memory safety. Second, we extend \srclang\ to
\targetlang, an enclave-aware calculus that enforces noninterference
against stronger attackers capable of observing arbitrary
non-enclave memory. Third, we develop a type-driven,
type-preserving compilation from \srclang\ to \targetlang that
automatically produces secure enclave-aware programs, eliminating the
need for manual partitioning while providing fine-grained control over
host–enclave boundaries.

We implement and evaluate \splitacronym on thirteen micro-benchmarks and real-world workloads, including applications from SGXGauge, on Intel SGX hardware. \splitacronym scales to OpenSSL ($425{,}953$ LLVM IR instructions) and supports multiple objectives that expose trade-offs among enclave TCB size, host--enclave transitions, and boundary data movement; for OpenSSL, optimizing for transitions reduces them from $393$ to $187$. Runtime overhead is dominated by fixed enclave costs for short-running workloads, whereas long-running applications better amortize these costs and approach native performance.

\end{abstract}

\maketitle


\ifappendix
\section{Narrative}
Make sure to address the below concerns.

\begin{enumerate}
\item \todo{Why can't we put everything in enclaves? Ans: Resources and computation limitations.}

\item \todo{Existing works: do not handle fine-tuning partitions. Our work is a partitioning debugging tool. Extensible framework for adding more objective functions; changing security policies.}

\item \todo{Existing formalisms on partitioning: do not handle declassification. We support gradual release.}

\item \todo{LLVM specific challenges: low-level. Overtainting and memory safety vulnerabilities compromise real security guarantees.}

\item \todo{Acknowledge that using postdominators is well-studied but highlight how we are doing it locally in a locally composable mananer.}

\item \todo{Bring in block label assignment algorithm for reducing overtainting. Deemphasize the novelty.}
  
\item \todo{Overview example: focus on changing partitioning boundaries.}

\item \todo{Explain annotation effort aimed at reducing TCB. Involves careful placement of declassifications.}

\item \todo{Explain the difference between Glamdring source/sink vs full type annotations}
  
\item \todo{Talk about inferno for reducing the annotation effort. Explicitly mention that inferring declassifications is a future work.}

\item \todo{Explain why partitioning a fully-annotated LL file is non-trivial}

\item \todo{SMT bottlenecks for scaling to real applications}

\item \todo{SPLIT partitioning is probabilistic and is dependent on the SMT solver}
  
  \item \todo{Scope for optimizations. Switchless calls for enabling performance; turning on optimizations during the partitioning that will reduce the code size: future work}
\end{enumerate}
\fi

\section{Introduction}

Trusted Execution Environments (TEEs), such as Intel SGX~\cite{intelsgx} and ARM TrustZone~\cite{trustzone}, provide strong hardware-based isolation against privileged attackers, including a compromised operating system. By executing sensitive code inside secure \emph{enclaves}, TEEs protect code and data even when the kernel is malicious. For example, an authentication application may store secret credentials (e.g., passwords) within an enclave, protecting them from the operating system.

\begin{figure}[h]
  \vspace{-10pt}
\scalebox{0.9}{
  \tikzset{
  publabel/.style={
    anchor=south,
    font=\scriptsize
  }
}

\begin{tikzpicture}[
  block/.style={rectangle, draw, minimum width=1.6cm, minimum height=0.55cm, font=\scriptsize},
  edge/.style={->, >=Latex, thick},
  node distance=0.5cm
]


\node[block, align=left] (entry) {
\texttt{entry:}\\
\texttt{\%sec = icmp sgt i32 \%s1, \%s2}\\
\texttt{br i1 \%sec,  \%then,  \%else}
};
\node[publabel] at ([yshift=-0.7mm]entry.north west) {\textcolor{codegreen}{public}};

\node[block, below left=of entry, align=left, xshift=10mm] (then) {
\texttt{then:}\\
\texttt{\%a = add i32 \%s1, 10}\\
\texttt{br  \%merge}
};
\node[publabel] at ([yshift=-0.7mm]then.north west) {\textcolor{red}{private}};

\node[block, below right=of entry, align=left, xshift=-20mm] (else) {
\texttt{else:}\\
\texttt{\%b = sub i32 \%s2, 3}\\
\texttt{br  \%merge}
 };
\node[publabel] at ([yshift=-0.7mm]else.north west) {\textcolor{red}{private}};

\node[block, below=of $(then)!0.5!(else)$, yshift=-3.5mm,  align=left] (merge) {
  \texttt{merge:}\\
  \texttt{\%pub = icmp eq i32 ..., 0} \\
  \texttt{br i1 \%pub,  \%z,  \%nz}
  };
\node[publabel] at ([yshift=-1mm]merge.north west) {\textcolor{red}{private}};

\node[block, below left=of merge, align=left] (zero) {\texttt{z:}\\ \dots};
\node[publabel] at ([yshift=-1mm]zero.north west) {\textcolor{red}{private}};

\node[block, below right=of merge, align=left] (nonzero) {\texttt{nz:}\\ \dots};
\node[publabel] at ([yshift=-1mm, xshift=3mm]nonzero.north west) {\textcolor{red}{private}};

\draw[edge] (entry) -- (then);
\draw[edge] (entry) -- (else);

\draw[edge] (then) -- (merge);
\draw[edge] (else) -- (merge);

\draw[edge] (merge) --  (zero);
\draw[edge] (merge) --  (nonzero);

\end{tikzpicture}
}
\caption{Tracking implicit flows in LLVM using a standard IFC type system. The program gets rejected since there are public side-effect in the private block \texttt{merge}. \label{fig:implicitflow}}
\vspace{-10pt}
\end{figure}
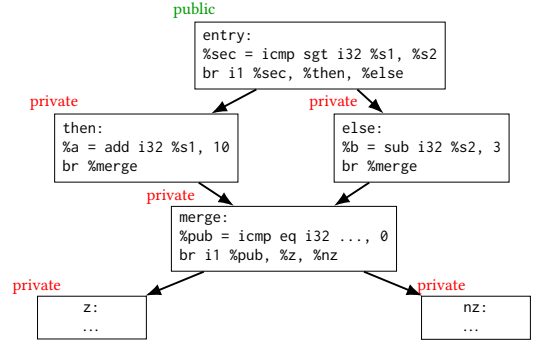

LLVM~\cite{llvm} is a widely used low-level intermediate language with
unstructured control flow (Figure~\ref{fig:implicitflow}) and arbitrary
pointer arithmetic. Our goal is to enforce security in LLVM using TEEs
while enabling flexible enclave partitioning. Given a security
specification identifying sensitive data, the system should
automatically place the corresponding code and data inside enclaves
with formal security guarantees. By targeting LLVM IR, our approach
aims to protect any application that compiles to LLVM, enabling broad
deployment.

Unfortunately, TEEs cannot prevent a vulnerable application from leaking its own
secrets; thus, placing an entire application inside an enclave does not
guarantee confidentiality. While \citet{gcoopsla2016} propose a formal
framework for enclave-based information-flow security, its guarantees
do not extend to LLVM-like low-level languages, which are memory-unsafe
(e.g., pointer arithmetic and \code{gep}-induced type confusion) and
feature unstructured control flow. Enforcing security in this setting
therefore introduces additional challenges.



\paragraph{Challenge 1. IFC for Low-level Languages.}
Enforcing IFC in LLVM requires tracking implicit flows from complex
control flow. A branch on secret data can taint all successors, leading
standard IFC systems to reject programs with public effects—an overly
restrictive outcome (Figure~\ref{fig:implicitflow}).
More permissive approaches allow taint downgrading
at join points, but prior work~\cite{siftal, webkitjs} does not integrate
cleanly with static IFC type system and often requires declassification to avoid
over-tainting.

Moreover, IFC does not address memory errors such as buffer overflows,
which can expose secrets via public pointers. Thus, meaningful security
requires preventing memory violations across security boundaries. 




\paragraph{Challenge 2. Using TEEs is challenging.}
TEE constraints require partitioning applications into enclave and host
components (e.g., SGX disallows system calls inside enclaves).
LibOS-based approaches~\cite{graphene} increase the TCB and attack
surface~\cite{graphenebugs}. Manual partitioning is costly and
error-prone~\cite{Kumar2022sgxgauge, Hasan2020PortOrShim}.
Several frameworks have explored automated application
partitioning~\cite{gcoopsla2016, civet, glamdring, cadote, rubinov,
hastee} for high-level languages, but none target LLVM IR–like
low-level languages.
Moreover, systems designed for real-world
applications~\cite{civet, glamdring, cadote, rubinov, hastee} do not
provide formal security guarantees, whereas \citet{gcoopsla2016}
offers strong formal guarantees but does not evaluate the approach
on real-world applications.



\paragraph{Challenge 3. Exploring Partitioning Trade-offs.}
Frequent host–enclave transitions incur significant overhead, making
naive partitioning inefficient. Most prior work, except
\citet{gcoopsla2016}, assumes a fixed objective and overlooks the
security–performance trade-offs inherent in practical partitioning.
While \citet{gcoopsla2016} proposes a proof-of-concept, their evaluation neither executes
the partitioned programs to analyze performance–security trade-offs nor
demonstrates the approach on real-world applications.


\begin{tcolorbox}[
colback=white!10,      
    colframe=black!70!white,
    title=Research Question, 
    rounded corners 
  ]
How can we enforce secure-by-construction enclave partitioning for LLVM
programs while enabling flexible, fine-grained partitioning trade-offs
between security and performance?

\end{tcolorbox}

Our first key insight is that coarse-grained memory safety suffices to
prevent buffer overflows from corrupting or exposing secret data. Our
second insight addresses over-tainting in low-level languages by
reconstructing control-flow dependencies from the control-flow graph
(CFG), which is key to scaling our framework to real-world benchmarks.
Third, we systematically study security–performance trade-offs in LLVM
by formulating fine-grained partitioning as a constraint optimization
problem with user-defined objective functions. Finally, to mitigate the
cost of constraint generation, we employ several optimizations for
efficient solving.

Our formalism  addresses the above research question by integrating TEE abstractions into the secure
compilation pipeline, enabling type- and policy-guided placement of
sensitive code and data inside enclaves to enforce end-to-end security
against  two distinct low-level attacker classes:
\lattacker, constrained by language-level security guarantees, and
\eattacker, limited only by hardware-enforced memory isolation.

\paragraph{Contributions.}
This paper makes the following contributions:

\begin{itemize}

\item \textbf{Secure low-level language with enclave support.}

\begin{itemize}

\item We formalize \underline{S}ecure \underline{I}mmediate
\underline{R}epresentation (\srclang), an enclave-agnostic calculus
based on LLVM IR~\cite{llvm} with information-flow policies (including
declassification), and prove that it enforces \emph{gradual release}—a
condition more permissive than noninterference—against \lattacker.

Our key insight is that IFC in low-level programs does not require
fine-grained memory safety. Instead, \srclang combines \emph{novel}
security-aware coarse-grained memory safety with permissive
control-flow–induced implicit-flow tracking, mitigating over-tainting
in LLVM-like code.

\item We extend \srclang\ with \underline{EN}clave capabilities to obtain
\targetlang, an enclave-aware calculus that enforces security against
\eattacker. \targetlang introduces novel enclave abstractions that model TEE memory isolation and prevent leakage of enclave-resident secrets across enclave--host transitions, enabling formal security reasoning for partitioned programs.


\end{itemize}







\item \textbf{Secure-by-construction enclave partitioning.}
We formulate secure application partitioning as a type-guided
compilation from enclave-agnostic \srclang programs to
enclave-aware \targetlang programs, and prove that the compilation
is \emph{security type-preserving}.
Crucially, although we focus on Intel SGX in this work, the enclave abstractions can be instantiated for other TEEs, such as ARM TrustZone, RISC-V enclaves, and AMD SEV, without changing the translation.

\item \textbf{Flexible partitioning via constraint optimization.}
We reduce enclave partitioning to a constraint optimization problem,
parameterized by developer-provided objective functions, enabling
systematic exploration of security–performance trade-offs.

\item \textbf{Implementation.}
We instantiate \srclang with LLVM IR and \targetlang with Intel SGX
enclaves, and implement \splitacronym as an out-of-tree LLVM pass
targeting the Open Enclave SDK~\cite{openenclave}.


\item \textbf{Evaluation.}
We evaluate \splitacronym on thirteen applications spanning micro-benchmarks and real-world workloads, including BFS, Memcached, OpenSSL, and SVM workloads from SGXGauge~\cite{sgxgauge}, on Intel SGX hardware. Across four objectives ($\tmin$, $\emin$, $\tmin+\emin$, and $\objdata$), \splitacronym produces distinct partitions that expose concrete trade-offs among enclave size, host--enclave transitions, and boundary data movement. These trade-offs are most pronounced in larger applications such as Memcached and OpenSSL: OpenSSL contains 425,827 LLVM IR instructions, yet only $0.03\%$ are directly labeled private, leaving substantial flexibility in placing the surrounding public code. Our runtime evaluation further quantifies the performance impact of the generated partitions, demonstrating how partitioning decisions and enclave--host interactions affect end-to-end execution time.
\end{itemize}
\noindent \textit{Note. Full definitions and proofs are in the
accompanying \textbf{technical report}.}




\section{Overview} \label{sec:overview}

Figure~\ref{fig:login} shows the C-snippet of an authentication program, which implements HMAC-based one-time password (HOTP).
The function \texttt{verifyOtp} checks the user's submitted OTP \texttt{submitted} against the expected OTP \texttt{expected}, which is a secret.

Figure~\ref{fig:login-llvm} shows the corresponding  \srclang snippet for \texttt{verifyOtp}. We obtain this by compiling C program to LLVM IR and  annotating with a security policy via  metadata. \srclang supports functions, blocks and  low-level instructions such as  memory load and store, both conditional and unconditional branches, call/return, and getelementptr (\code{gep}) that permits arbitrary pointer arithmetic. We focus on novel aspects: security policy, declassification, security-aware memory safety, and the security condition.

Without loss of generality, we assume  a two point confidentiality lattice, public $\public$ and private $\private$ with  the usual information-flow ordering $\public \sqsubseteq \private$.
The annotations \texttt{"$\code{i32}^\bot$"} and \texttt{"$\code{i32}^\top$"} specify the confidentiality of variables such as \texttt{submitted} (lines 2 and 8) and \texttt{expected} (lines 12 and 13). The \texttt{ct\_compare} call (line 17) has a $\private$ policy on its return value, and $\private$, and $\public$ policies on its first and second arguments, respectively. Since \texttt{status} depends on sensitive data, it is labeled $\private$ and must be explicitly declassified before use in public contexts; this is achieved by \texttt{declassify.i32} (line 21), which converts \texttt{status} from $\private$ to $\public$. The two point lattice can be extended to a more complex lattice, with all non-$\bot$ levels being treated as ``secret'' for the purposes of enclave assignment.

\begin{figure*}[t]
\begin{subfigure}[c]{0.47\textwidth}
  \centering
  \scalebox{0.83}{
\begin{lstlisting}[style=mystyle, language=C, mathescape=true]

static uint32_t submittedCode;   /* public */
static uint64_t counter;         /* public */

/* set up seed and read (public) inputs */
void init(int argc, char **argv);

/* recompute expected OTP (SECRET result) */
uint32_t expectedOtp(uint64_t ctr);

/* constant-time compare of two 6-digit codes*/
int ct_compare(uint32_t a, uint32_t b);

/*
 * 1. recompute expected OTP from secret seed
 * 2. CT compare against the submitted code
 * 3. return DECLASSIFIED accept/reject status
*/
int verifyOtp(uint32_t submitted, uint64_t ctr) {
    uint32_t expected = expectedOtp(ctr);
    int status = ct_compare(expected, submitted);
    return declassify_i32(status);
}

int main(int argc, char **argv) {
    init(argc, argv);
    int status = verifyOtp(submittedCode, counter);
    if (status) printf("\nAccess Granted!\n");
    else        printf("\nAccess Denied!\n");
    return status ? 0 : 1;
}
\end{lstlisting}
  }
  \caption{HOTP Authentication in C (snippet) \label{fig:login}}
\end{subfigure}%
~
\begin{subfigure}[c]{0.62\textwidth}
  \scalebox{0.83}{
\begin{lstlisting}[style=mystyle, language=llvm, mathescape=true]

define i32$^\bot$ @verifyOtp(i32$^\bot$ %submitted, i64$^\bot$ %ctr) $(\bot)$ {
entry:
  %submitted_loc = alloca i32$^\bot$
  %ctr_loc = alloca i64$^\bot$
  %expected_loc = alloca i32$^\top$
  %status_loc = alloca i32$^\top$
  store i32$^\bot$ %submitted, ptr$^\bot$ %submitted_loc
  store i64$^\bot$ %ctr, ptr$^\bot$ %ctr_loc
  ;; expected = expectOtp(ctr)
  %ctr_1 = load i64$^\bot$, ptr$^\bot$ %ctr_loc
  %expected = call i32$^\top$ @expectedOtp(i64$^\bot$ %ctr_1)
  store i32$^\top$ %expected, ptr$^\top$ %expected_loc
  ;; status = ct_compare(expected, submitted)
  %expected_1 = load i32$^\top$, ptr$^\top$ %expected_loc
  %submitted_1 = load i32$^\bot$, ptr$^\bot$ %submitted_loc
  %status = call i32$^\top$ @ct_compare(i32$^\top$ %expected_1, i32$^\bot$ %submitted_1)
  store i32$^\top$ %status, ptr$^\top$ %status_loc
  ;; return declassify_i32(status)
  %status_1 = load i32$^\top$, ptr$^\top$ %status_loc
  %out = declassify i32$^\bot$, i32$^\top$ %status_1
  ret i32$^\bot$ %out
}
\end{lstlisting}
  }
  \caption{\code{verifyOtp()} in \srclang \label{fig:login-llvm}}
\end{subfigure}
\vspace{-10pt}
\end{figure*}

\srclang's memory consists of both stack and heap regions.
Novel to \srclang, memory is partitioned into security-indexed
regions, and allocation instructions allocate memory from the
corresponding regions; the memory security specification maps
each memory location to its security region. In the example above,
lines 3 and 4 allocate memory from public regions.

The \code{gep} instruction performs pointer arithmetic and is a
well-known attack vector for exposing secrets through public
pointers. A key contribution of \srclang is to close this attack
vector by enforcing coarse-grained memory safety. The semantics of
\code{gep} therefore differ from standard LLVM IR: the instruction
aborts if the computed address crosses security boundaries.

Unlike fine-grained memory safety mechanisms (e.g., fat pointers
that store base addresses and offsets), this check only ensures
that pointer arithmetic does not cross security regions; arithmetic
may still overflow within the same region. The \code{gep}
instruction also models integer-to-pointer conversions, which can
lead to type-confusion attacks (discussed in
Section~\ref{fig:ptrtoint}).

The \srclang IFC type system assigns security levels to basic blocks
and tracks information flows across them. In particular, it propagates
taints induced by conditional branches to all successor blocks. While
this propagation is standard, \srclang also statically identifies
points in the control-flow graph where such taints can be safely
removed.

Similar to prior work, \srclang uses postdominator nodes to determine
these points. However, \srclang introduces a novel form of postdominator-indexed
security labels that allow the type system to safely remove taints at
the corresponding nodes.  We defer the discussion to Section~\ref{sec:source}.

We consider an attacker \lattacker that is restricted to observing
memory regions labeled at or below its security level $l$.
For example, a \pubattacker (i.e., an \lattacker with $l=\public$)
cannot access $\private$ memory regions.

\srclang's type system enforces security against \lattacker.
For programs that do not perform declassification, it guarantees
\emph{noninterference} (Theorem~\ref{thm:sirni}): if two input memories
differ only in locations whose security level is strictly higher than
$l$, and are otherwise indistinguishable to the attacker, then their
outputs are also indistinguishable to an attacker at security level $l$.

For programs that include declassification, \srclang enforces
\emph{gradual release} (Theorem~\ref{thm:sirgr}): if two input memories
differ in locations with security level strictly higher than $l$, but
release the same information through declassification, and are
otherwise indistinguishable to the attacker, then their outputs remain
indistinguishable at security level $l$~\cite{gradualrelease}.
Gradual release is more permissive than noninterference, as it allows
controlled release of information according to semantic declassification
policies. Consequently, the program in Figure~\ref{fig:login-llvm} is secure.

\begin{figure*}[t]
  \centering
  \begin{tikzpicture}
    \regblock{id=mymain,text={my\_main},x=-7.6cm,y=0cm,args=double};
    \regblock{id=init,text={init},x=-7.6cm,y=-1.2cm};
    \regblock{id=printf,text={printf},x=-7.6cm,y=1.2cm,args={dashed}};
    \regblock{id=strtoul,text={strtoul},x=-5.8cm,y=-1.2cm,args={dashed}};

    \regblock{id=verifyotpecall,text={verifyOtp\_ecall},x=-3cm,y=0cm,args={fill=green!10}};
    \regblock{id=verifyotp,text={verifyOtp},x=-0.4cm,y=0cm,args={fill=green!10}};
    \regblock{id=ctcompare,text={ct\_compare},x=-0.4cm,y=1.2cm,args={fill=green!10}};
    \regblock{id=declassifyi,text={declassify.i32},x=-0.4cm,y=-1.2cm,args={fill=green!10}};
    \regblock{id=expectedotp,text={expectedOtp},x=1.8cm,y=0cm,args={fill=green!10}};
    \regblock{id=hmacsha,text={hmac\_sha1},x=4cm,y=0cm,args={fill=green!10}};
    \regblock{id=shainit,text={sha1\_init},x=1.8cm,y=1.2cm,args={fill=green!10}};
    \regblock{id=shafinal,text={sha1\_final},x=4cm,y=1.2cm,args={fill=green!10}};
    \regblock{id=shaupdate,text={sha1\_update},x=6.2cm,y=0cm,args={fill=green!10}};
    \regblock{id=shablock,text={sha1\_block},x=6.2cm,y=-1.2cm,args={fill=green!10}};
    \regblock{id=rol,text={rol},x=4cm,y=-1.2cm,args={fill=green!10}};
    \regblock{id=memset,text={memset},x=1.8cm,y=-1.2cm,args={fill=green!10,dashed}};
    \regblock{id=memcpy,text={memcpy},x=6.2cm,y=1.2cm,args={fill=green!10,dashed}};

    \draw[dashed,gray] (-4.75cm,1.9cm) -- (-4.75cm,-2.0cm);
    \node[anchor=south,font=\scriptsize,gray] at (-4.75cm,1.9cm) {enclave boundary};

    \draw
      (mymain) edge (printf)
      (mymain) edge (init)
      (init) edge (strtoul)

      (mymain) edge node[above,font=\scriptsize,pos=0.6] {ecall} (verifyotpecall)

      (verifyotpecall) edge (verifyotp)
      (verifyotp) edge (ctcompare)
      (verifyotp) edge (declassifyi)
      (verifyotp) edge (expectedotp)
      (expectedotp) edge (hmacsha)
      (hmacsha) edge (shainit)
      (hmacsha) edge (shafinal)
      (hmacsha) edge (shaupdate)
      (hmacsha) edge (memset)
      (hmacsha) edge (memcpy)
      (shafinal) edge (shaupdate)
      (shaupdate) edge (shablock)
      (shaupdate) edge (memcpy)
      (shablock) edge (rol)
    ;

    \begin{scope}[font=\scriptsize]
      \node[draw,fill=green!10,minimum width=4mm,minimum height=3mm] (lA) at (-7.6cm,-2.5cm) {};
      \node[anchor=west] at (-7.3cm,-2.5cm) {enclave (app)};
      \node[draw,dashed,fill=green!10,minimum width=4mm,minimum height=3mm] (lB) at (-4.6cm,-2.5cm) {};
      \node[anchor=west] at (-4.3cm,-2.5cm) {enclave (library)};
      \node[draw,minimum width=4mm,minimum height=3mm] (lC) at (-0.9cm,-2.5cm) {};
      \node[anchor=west] at (-0.6cm,-2.5cm) {host (app)};
      \node[draw,dashed,minimum width=4mm,minimum height=3mm] (lD) at (1.7cm,-2.5cm) {};
      \node[anchor=west] at (2cm,-2.5cm) {host (library)};
    \end{scope}
  \end{tikzpicture}
  \caption{Partitioned HOTP application, MinTCB partition.\label{fig:output}}
\end{figure*}
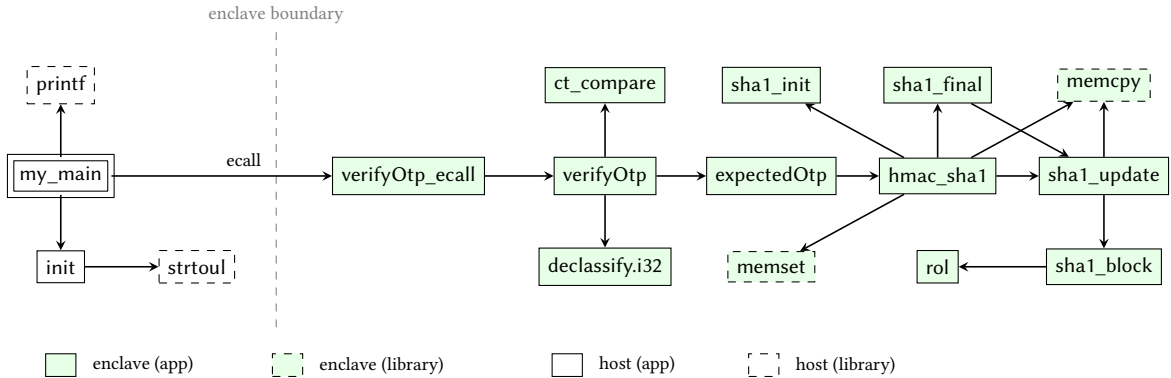

\paragraph{Secure Partitioning.}
\targetlang extends \srclang with novel enclave memory regions and
enclave-management instructions for creating and dismantling enclaves,
and for managing transitions into and out of enclave execution
(Section~\ref{sec:target}). These abstractions closely model TEE
execution and its associated defenses. The semantics enforce the key
invariant that enclave memory locations can only be accessed by code
executing within the same enclave, reflecting the memory isolation
guarantees provided by Intel SGX.

\targetlang's type system enforces noninterference against a more
powerful \eattacker (Theorem~\ref{thm:siren-ni-low}) that can observe
the entire host memory regardless of its security level (e.g., a
compromised kernel). In particular, if two input host memories are
indistinguishable to the attacker, then the resulting output host
memories remain indistinguishable to the \eattacker.

We reduce enclave partitioning to compiling an enclave-oblivious
\srclang program into an enclave-aware \targetlang program. Our
translation enforces only the essential security constraints—for
example, ensuring that secret pointers remain within their
corresponding enclaves—while allowing flexibility in how code is
assigned to enclaves. The resulting translation is type-driven and
type-preserving (Sections~\ref{sec:translation} and \ref{sec:transpreserve}),
ensuring that the generated \targetlang program is well-typed and
provably secure (Theorem~\ref{thm:transtypepreserve}).

\paragraph{Partitioning Trade-offs.}
Our implementation consists of two components: \infertool, which infers a complete security policy from sparsely annotated LLVM IR, and \splitacronym, which partitions the annotated IR into Intel SGX host and enclave components using the OpenEnclave SDK. By operating on LLVM IR, \splitacronym is agnostic to the source language (e.g., C, C++, or Rust) and facilitates future support for other TEE backends, such as ARM TrustZone and RISC-V enclaves.

We define the trusted computing base (TCB) as the code placed inside the enclave. Rather than naively placing the entire application in the enclave, \splitacronym supports four objectives: \tmin minimizes host--enclave transitions, \emin minimizes enclave code size (and hence the TCB), \objboth jointly minimizes transitions and TCB size, and \objdata minimizes data transferred across the enclave--host boundary.

Our evaluation demonstrates that these objectives expose concrete partitioning trade-offs. For the running HOTP example, minimizing the TCB produces 635 enclave instructions and 53 transitions, whereas minimizing transitions increases the enclave to 733 instructions but reduces transitions to 1. Figure~\ref{fig:output} shows the call graph of HOTP translated under the \emin objective. \code{my\_main}, \code{init}, and I/O functions such as \code{printf} and \code{strtoul} reside on the host, while sensitive routines such as \code{verifyOTP} and required library routines and wrappers reside in the enclave. Note that I/O functions are pinned to the host because Intel SGX enclaves cannot directly perform I/O. For TEEs that support I/O within the trusted domain, this restriction can be disabled.

This trade-off becomes more pronounced in larger applications. For BFS, minimizing transitions reduces their number from 137 to 17, while for OpenSSL (425,953 LLVM IR instructions), it reduces transitions from 393 to 187. In contrast, several benchmarks produce similar partitions across objectives because their security and placement constraints leave little optimization freedom. We evaluate these trade-offs in Section~\ref{sec:eval}.

  \section{\srclang: The Source Language} \label{sec:source}




\srcsyntax
\srclang  models LLVM with four key distinctions. First, each data type is explicitly associated with a security policy, and the type system enforces fine-grained information-flow tracking using these policies. Second, the memory model employs a partitioned allocator indexed by security levels. Third, memory access instructions are instrumented to enforce coarse-grained memory safety: pointer arithmetic and load/store operations are confined to the same security-level region. Fourth, it supports a declassify instruction that is operationally the same as assignment but allows information to flow from private to public variables. 

\srclang is simpler than VeLLVM~\cite{zellweger2012vellvm}---a framework for verifying LLVM programs---but introduces
security-oriented features. The goals of \srclang and VeLLVM are
orthogonal; in principle, the security mechanisms developed for
\srclang could be incorporated into VeLLVM to obtain similar
information-flow guarantees.
Below, we briefly summarize the syntax, well-formedness
and focus on the novel aspects of \srclang.
Readers familiar with LLVM semantics may directly skip to Section~\ref{sec:coarsegrained}.

Figure~\ref{fig:srcsyntax} shows the syntax of \srclang. Variables $x$, $y$, $z \in \Var$ map via environment $\Vspace$ to constants $\kappa$, which are memory locations $m$ or numeric literals $c$. Functions are $f$ and blocks $b$. Expressions $e$ are literals or variables. Ground types $\tau$ are boolean $\ione$, 32-bit integer $\itt$, or pointers. Each value has a security label $l \in {\public, \private}$ from a two-point lattice ($\public \sqsubseteq \private$). Full types $\mathit{Type}$ are tuples $(\tau, l)$.

A program $p$ is a non-empty list of functions. Each function $\func$ has a standard signature (return type, name, typed parameters), a sequence of blocks, and a security level $l$ capturing the minimum confidentiality of side effects and its return value. Blocks $\mathit{block}$ have identifiers $b$ and bodies $\body$, sequences of instructions ending with a terminator. All function names and block identifiers are unique, with $\Fspace$ mapping each name or block label to its function definition. \srclang enforces static single assignment (SSA), where each variable is defined at most once.

Memory $\memory$ maps locations to values and is partitioned by heap $\Loc_{\Hspace}$ and stack $\Loc_{\Sspace}$, with each divided into security regions indexed by level $l$. The memory specification $\Menv$ associates each location with a type, comprising a ground type and security label.

Instructions are classified as internal ($i$) or terminators ($t$), with terminators controlling program flow. \srclang models all major instruction classes, including arithmetic (\code{add}), pointer arithmetic (\code{gep}), memory access (\code{load}, \code{store}), branching (\code{br}), function calls and returns, and phi ($\phi$) for merging values along control-flow paths. Each instruction specifies operand and result types. Shaded instructions are not part of the surface syntax. While branching semantics are standard, \srclang uses a permissive typing approach to track control-flow-induced information flow. A program is \emph{well-formed} if it satisfies SSA, has a unique return node, and places all $\phi$ instructions at the start of each block.

\paragraph{\srclang Semantics}
\label{sec:srcsemantics}
The operational semantics follow a standard small-step formulation.
The small-step judgment
$\scfgi{i=i} \sstepi \scfgi{h=\Hspace', s=\Sspace', v=\Vspace', i=i'}$
describes the transition of a \srclang configuration
$\scfgi{i=i}$. A configuration is represented as the tuple
$\langle \Hspace, \Sspace, \Vspace, \prev, i \rangle$,
where $\Hspace$ denotes the heap, $\Sspace$ the stack,
$\Vspace$ the variable environment, $\prev$ the previously
evaluated block, and $i$ the current instruction.
$\Fspace$ denotes the function lookup environment.

We omit standard rules for operations such as assignments,
function calls,  branches and arithmetic expressions, and instead focus
on the novel semantics that enforce coarse-grained memory
safety. The full semantics appear in \S~1.2 of the technical report.


\subsection{Security-aware Coarse-grained Memory Safety} \label{sec:coarsegrained}

The syntax of \code{gep} is
$\GEP{x}{\type{\tau}{l}}{\type{\ptr{\tau}}{l}y}{\itt c}$,
where $y$ is the base pointer and $c$ is the offset. If $y$ points to
the base address $m$, then $x$ points to the address $m+c$, referred to
as the effective address. We motivate the need for enforcing memory safety in \srclang (and,
more generally, in LLVM-like low-level languages) by demonstrating the vulnerabilities arising from \code{gep} instruction.

\begin{figure}[h]
\begin{lstlisting}[style=mystyle, language=llvm, morekeywords={gep}, mathescape=true]
	%p = alloca <$\mathrm{Int_{32}}$, $\bot$> 
	%dp = alloca <$\mathrm{Int_{32}*}$, $\bot$> 
	store <$\mathrm{Int_{32}*}$, $\bot$> %p, <$\mathrm{Int_{32}**}$, $\bot$>  %dp 
	%dp_als = gep <$\mathrm{Int_{32}}$, $\bot$>,  <$\mathrm{Int_{32}*}$, $\bot$>  %p, <$\mathrm{Int_{32}}$, $\bot$> offset
	    ;ip is equivalent to (ptrtoint p)
	%ip = load <$\mathrm{Int_{32}}$, $\bot$>, <$\mathrm{Int_{32}*}$, $\bot$> %dp_als 
\end{lstlisting}
\caption{Type confusion vulnerability. \code{gep} models \code{ptrtoint}; \code{dp} and \code{dp\_als} are aliases with different ground types. \label{fig:ptrtoint}}
\end{figure}

First, the base and effective addresses may belong to different
security regions, allowing information leakage through seemingly
public pointers. Second, \code{gep} can simulate
\code{ptrtoint}/\code{inttoptr} conversions, enabling type confusion.
In such cases, a public pointer may be reinterpreted as a public
integer. This is dangerous because the integer can subsequently be
promoted to a private value and converted back into a private pointer,
potentially enabling writes of secret data through apparently private
pointers that alias public memory.

Figure~\ref{fig:ptrtoint} illustrates this vulnerability. Lines~1--2
allocate memory containing values of type $\mathrm{Int_{32}}$ and
$\mathrm{Int_{32}}^{*}$, respectively. Line~3 stores $\%p$ into
$\%dp$, and line~4 uses \code{gep} to create a pointer \texttt{\%dp\_als} by
adding an \texttt{offset} to $\%p$. For a suitable \texttt{offset},
\texttt{\%dp\_als} can alias $\%dp$. Notably, \code{gep} treats the contents of
\texttt{\%dp\_als} as having type $\mathrm{Int_{32}}$ rather than
$\mathrm{Int_{32}}^{*}$. Consequently, line~6 reads the contents of
\texttt{\%dp\_als} as an integer instead of a pointer.

A similar construction can simulate \code{inttoptr}. Together, these
operations form a gadget that can launder secrets even though the
program is well typed.

\begin{figure}[t]
\scalebox{0.85}{
  \begin{mathpar}
     \srcemalloc{}

    \srcegep{}

    \srceload{}
  \end{mathpar}
 }
  \caption{\srclang: Memory access semantics.\label{fig:srcmemsemantics}}
  \vspace{-15pt}
\end{figure}

Since the vulnerability relies on pointer aliasing, we ensure that a $\bot$ pointer can never alias a $\top$ pointer. We partition memory into disjoint $\bot$- and $\top$-regions and allocate objects according to their security levels. We enforce this invariant in two ways. First, we restrict \code{gep} so that pointer arithmetic cannot cross region boundaries: an operation that computes an address in a different security region is terminated, enforcing security-aware coarse-grained memory safety. Second, we restrict \code{load} so that a loaded pointer belongs to the same region as the memory location from which it is loaded. For example, when loading from a $(\tau*)*$ pointer $y$ into a $\tau*$ pointer $x$, both $x$ and $y$ must point to the same security region. Consequently, $\top$ pointers can only point into the $\top$ region and $\bot$ pointers into the $\bot$ region; since the regions are disjoint, cross-security pointer aliasing is impossible. We incorporate this invariant into type preservation by requiring well-typed memories to satisfy it, thereby preventing secrets from being laundered through pointer/integer confusion.

Figure~\ref{fig:srcmemsemantics} presents the semantics for the memory
access instructions that enforce this property. Rule~\rulename{E-Malloc}
allocates a fresh heap location $m$ from $\Loc_{H,l}$ for security
level $l$, zero-initializes $\Hspace(m)$, and binds $x$ to a pointer
to $m$ whereas \code{free} uninitializes $m$. This models
partitioned heap allocators such as Allocation Tokens in
Clang~\cite{clang-alloctoken}. The notation
$\Vspace[x \mapsto m]$ extends or updates $\Vspace$ by mapping $x$ to
$m$.

Rule~\rulename{E-GEP} computes the effective address $m'$ by adding the
evaluated offset $c$ (obtained from expression $e$) to the base address
$m$ stored in pointer $y$, binding the result to $x$. The rule requires
$m' \in \region{m}$, ensuring that pointer arithmetic does not cross
security-region boundaries. 
Note that this does not prevent buffer overflows within a security region.
Load and store operations are similarly
restricted to dereference pointers whose source and target addresses
belong to the same region.

Rule~\rulename{E-Load} evaluates a non-null pointer $y$ to obtain the
address $m$ and retrieves the corresponding value, binding it to $x$.
The premise $\mathrm{isptr}(\tau) \implies \kappa \in \region{m}$
ensures that pointers remain confined to their region, preventing
cross-boundary accesses.

For example, suppose addresses $0$--$4999$ belong to the $\bot$ region and $5000$--$9999$ to the $\top$ region. For $\GEP{x}{\type{\itt}{\bot}}{\type{\ptr{\itt}}{\bot} y}{\type{\itt}{\bot}10}$, if $y$ points to $4995$, the effective address is $5005\in\top$; since \code{gep} crosses the region boundary, evaluation cannot proceed (in practice, the address could instead be sandboxed within the $\bot$ region). Similarly, for $\Load{x}{\type{\ptr{\itt}}{\top}}{\type{\ptr{\ptr{\itt}}}{\top}\var{y}}$, if loading from $y\in\top$ yields $\kappa=4000\in\bot$, the premise $\mathrm{isptr}(\tau)\Longrightarrow\kappa\in\region{m}$ fails, preventing a $\top$ pointer from referencing the $\bot$ region.




\subsection{\srclang Security Definition}
\label{sec:srcthreatmodel}
We consider an \lattacker that can observe memory locations with
security level $l$ or lower. Two memories $\memory_1$ and $\memory_2$
are \emph{$l$-equivalent}, written $\memory_1 \simeq_{l}^{\Menv} \memory_2$,
if their $l$-projections are equal. 

We lift the small-step relation to produce a trace of memories, i.e.,
$\scfgp{p=p} \ssteptr \scfgres{} \raises T$.
Trace equivalence is defined pointwise over the corresponding elements
of the traces.
A program $p$, without any declassification instructions, is secure if the \lattacker cannot distinguish between the memory traces generated using $l$-equivalent initial memories. Definition~\ref{def:nilattacker} formally states the same.

\begin{definition}[\srclang Noninterference w.r.t  \lattacker]\label{def:nilattacker}
  Given a memory specification $\Menv$, a program $p$ is noninterfering against a \lattacker, if for all initial heaps $\Hspace^0_1$, $\Hspace^0_2$ and a variable map $\Vspace$:
  \begin{enumerate}
  \item the initial heaps are $l$-equivalent; that is, $\Hspace^0_1 \simeq_{l}^{\Menv} \Hspace^0_2$, and
  \item the program configuration $\scfgp{h=\Hspace^0_i}$ steps to the final configuration $\scfgres{h=\Hspace'_i, c=c_i}$, emitting the trace $T_i$; that is, for all $i \in \{1, 2\}$,  $\scfgp{h=\Hspace^0_i}\ssteptr \scfgres{h=\Hspace'_i, c=c_i} \raises T_i$,
  \end{enumerate}
  then $c_1 = c_2$ and the traces are $l$-equivalent, i.e., $T_1 \simeq_{l}^{\Menv} T_2$.
\end{definition}


Definition~\ref{def:nilattacker} enforces strict noninterference and thus forbids any intentional release of secret information. To accommodate controlled release, our second security condition captures \emph{gradual release}. Intuitively, a program satisfies gradual release if, at every point during execution, two runs that start from low-equivalent heaps and have released the same information so far produce indistinguishable public traces.

\begin{definition}[Release Function]
Given a memory specification $\Menv$, a release function
$R : \memory \mapsto \mathsf{Val}$,
maps an initial heap to the information that may be declassified during
execution. Two memories agree on the released information if $R(\memory_1) = R(\memory_2)$.
\end{definition}

\begin{definition}[Release Projection]
Let $T$ be a trace. The release projection $R(T)$ returns the subsequence of values explicitly declassified along $T$.
\end{definition}

\begin{definition}[\srclang Gradual Release w.r.t \lattacker]\label{def:gradualrelease}
Given a memory specification $\Menv$ and a release projection $R(\cdot)$ over traces,
a program $p$ satisfies gradual release against a \lattacker\ if for all initial heaps
$\Hspace^0_1$, $\Hspace^0_2$ and a variable map $\Vspace$:

\begin{enumerate}
\item the initial heaps are $l$-equivalent; that is,
$\Hspace^0_1 \simeq_{l}^{\Menv} \Hspace^0_2$,

\item the program configuration $\scfgp{h=\Hspace^0_i}$ steps to the final configuration
$\scfgres{h=\Hspace'_i, c=c_i}$, emitting the trace $T_i$; that is, for all $i \in \{1,2\}$,
$
\scfgp{h=\Hspace^0_i}
\ssteptr
\scfgres{h=\Hspace'_i, c=c_i}
\raises T_i,
$
\end{enumerate}

then $c_1 = c_2$ and for all prefixes $U_1 \preceq T_1$ and $U_2 \preceq T_2$, if
$R(U_1) = R(U_2)$,
then $U_1 \simeq_{l}^{\Menv} U_2$.
\end{definition}


\subsection{\srclang Type System}\label{sec:srctyping}


 
\begin{figure}[h]
\scalebox{0.85}{
 \begin{mathpar}
\srctgep{}

\srctload{}
\end{mathpar}
}
\caption{Typing \code{gep} and \code{load} \label{fig:srctyping}}
\vspace{-10pt}
\end{figure}
The typing judgment $\stctxb{} \tc{i}$ asserts that instruction $i$ is well-typed under context $(\pc, \cur, \lf)$, where $\pc$, the standard program counter label, tracks control-induced flows, $\cur$ is the current block, and $\lf$ is the function’s label.
A common premise $\pc \flowsto l$ (for result label $l$) prevents implicit flows within a block, ensuring that high-security contexts do not produce low-observable effects. As an exception, \code{declassify}  admits a flow from $\top$ to $\bot$. Most of the typing rules, except for branches, follow standard IFC type system; we therefore show representative typing rules for \code{gep} and \code{load} (Figure~\ref{fig:srctyping}). We exclusively discuss typing branches in Section~\ref{sec:permcontrolflow}. The full typing rules appear in \S~1.4 of the accompanying technical report.

Rule~\textsc{T-Load} states that the \code{load} instruction is
well typed if $\type{\tau}{l_2} \flowsto \type{\tau}{l_1}$ and
$\pc \flowsto l_2$. The former prevents explicit flows: a well-typed
instruction cannot load a value from a private location into a
public variable. The latter prevents implicit flows by ensuring
that the program counter label flows to the label of the accessed
location. 
The relation $\type{\tau}{l_2} \flowsto \type{\tau}{l_1}$ lifts the
$\flowsto$ relation on security labels to full types when $\tau$
is a non-pointer type. Pointer types are invariant under subtyping;
that is, $\type{\tau*}{l} \flowsto \type{\tau*}{l}$.
Rule~\textsc{T-GEP} (and \textsc{T-Store}, not shown here) has
analogous premises.

Notice that \textsc{T-Load} on its own (i.e., without the enforcement of $\top$ and $\bot$ regions) is not sufficient to prevent pointer type confusion attacks. If $y$ happens to be the result of a $\code{gep}$ instruction, the type of $y$ may not match the type of other pointers pointing to the same location (such as a pointer that was assigned that location via $\code{malloc}$). In particular, if we load from a location that was created as an integer location, it could have a value that is not valid as a pointer to an $l_1$ location (because the location it points to has a different security level). This is why we require the additional premise in \textsc{E-Load}.

\textsc{T-Load} alone, without enforcing disjoint $\top$ and $\bot$ regions, cannot prevent pointer type confusion. If $y$ is derived via \code{gep}, its type may differ from that of an alias to the same location. Consequently, loading a pointer from an integer-typed location may produce an address in a region inconsistent with its security label, motivating the additional premise in \textsc{E-Load}.

\paragraph{Typing HOTP.}

In our running example (Figure~\ref{fig:login-llvm}), values derived from the $\top$-labeled \code{\%expected} remain $\top$ unless declassified. Thus, \code{\%status}, \code{\%status\_loc}, and \code{\%status\_1} (lines~17--21) inherit the label $\top$ from \code{\%expected} and can be assigned to the lower-labeled \code{\%out} on line~21 only via declassification; an ordinary assignment would not typecheck. Hence, provided \code{\%expected} is labeled $\top$, the type system prevents its information from flowing to $\bot$ except through explicit declassification.

\subsection{Permissive Control-flow Typing}\label{sec:permcontrolflow}

The problem of over-tainting (or ``label creep'') is exacerbated in
languages with unstructured control flow, where unconditional branches
can induce overly restrictive typing (Figure~\ref{fig:implicitflow}).
Prior work commonly leverages postdominators, a node that executes
along all paths from the current block to the function's exit, to statically reset the
program counter ($\pc$) to a safe label. Our approach follows a similar
intuition but avoids complex continuations~\cite{siftal} and abstract
interpretation–based influence regions~\cite{aldous2015}. Instead, we
introduce a simpler and more expressive mechanism that tracks implicit
flows using postdominator-indexed pc labels, enabling taint to be
safely peeled off at postdominator nodes.

In \srclang, $\pc$ captures control dependencies from both conditional
and unconditional branches. The $\pc$ of a block represents a lower
bound on its control-flow dependencies. Each basic block is assigned a
$\pc$ of one of three forms: a standard confidentiality level $l$, the novel
block-indexed label $\declat{l}{b}$ (``level $l$ until block $b$''), or
a join $\pc_1 \sqcup \pc_2$. The block-indexed label designates $b$, an
immediate postdominator, as a safe reset point for $\pc$, marking the
end of the influence of the last conditional branch.
For example, the \texttt{merge} block in
Figure~\ref{fig:implicitflow} marks the end of the taint induced by
\texttt{\%sec} in the \texttt{entry} block. We first illustrate this idea
on a more complex control-flow example, and then present the type system.

\begin{figure}[h]
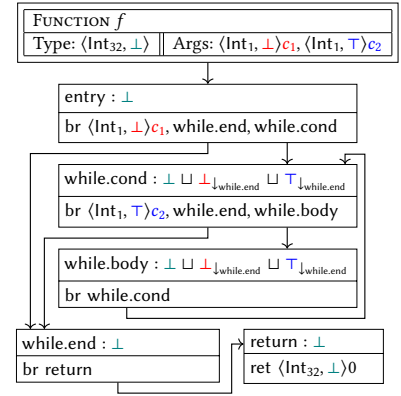

\scalebox{0.7}{
\blocklabeldiagram
}
\caption{Permissive control-flow typing using blocked-indexed $\pc$ labels.  \label{fig:permcflow}}
\end{figure}

Figure~\ref{fig:permcflow} illustrates a control-flow graph annotated
(possibly manually) with block-level security labels, modeling a
while-loop nested within an if-condition.
The block \code{while.end} immediately postdominates
both \code{while.cond} and \code{while.body}, serving as the convergence
point for the loop and the conditional. As we will describe later, the
\srclang type system admits this program.

We focus on the $\pc$ labels assigned to each block. The \code{entry}
block inherits the function’s label $\bot$. The $\pc$ for
\code{while.cond},
$\bot \sqcup \declat{\public}{\code{while.end}} \sqcup
\declat{\private}{\code{while.end}}$, captures taints induced by
conditions $c_1$ and $c_2$, both of which can be peeled off at the
postdominator \code{while.end}. The block \code{while.body} carries the
same $\pc$.

At \code{while.end}, the postdominator marks the end of these taints,
allowing the $\pc$ to reset to $\bot$ after peeling off the
corresponding policies. Thus, block-indexed $\pc$ labels enable
fine-grained tracking of taints from multiple conditionals without
requiring complex bookkeeping.

Rather than requiring manual assignment of block security labels, we
provide an algorithm that computes the most permissive $\pc$ labels,
that the type system admits. The formal algorithm is
presented in \S~1.4.1 of the accompanying technical report.
We incorporate this algorithm into \infertool (Section~\ref{sec:impl}), our policy inference
tool, to reduce annotation effort for real-world applications.

Intuitively, the algorithm propagates control-flow taints of the form
$\declat{l}{\code{pdom}}$ along conditional branches, where $l$ is the
label of the branch condition and \code{pdom} its immediate
postdominator. Successor blocks inherit these taints to capture
implicit flows. For example, the block
\code{while.cond} inherits $\declat{\public}{\code{while.end}}$ from the
$\public$-conditioned branch in \code{entry} and
$\declat{\private}{\code{while.end}}$ from the
$\private$-conditioned branch in \code{while.body}, yielding
$\bot \sqcup \declat{\public}{\code{while.end}} \sqcup
\declat{\private}{\code{while.end}}$.
At the postdominator \code{while.end}, these control dependencies are
discharged, and the $\pc$ resets to $\bot$.


\medskip
We now formalize typing rules for branches using block-indexed labels
$\declat{l}{b}$. A key challenge is capturing legal flows between $\pc$
labels. To this end, we introduce a novel block-indexed flows-to
relation, $\flowsto_b$, which refines the standard $\flowsto$ relation.
It propagates security levels along control paths and discharges the
control-flow taint $\declat{l}{b}$ at the corresponding immediate
postdominator, effectively lowering the $\pc$. Intuitively,
$\flowsto_b$ behaves like $\flowsto$ along control paths and peels off
the indexed label at the postdominator where control flow merges.



\begin{figure}[h]
\scalebox{0.85}{
\begin{mathpar}
    \bflowsuper{}
    \bflowhigh{}
    \bflowdecl{}
\end{mathpar}
}
\caption{Block-indexed flows-to \label{fig:bflowsto}}
\end{figure}

Figure~\ref{fig:bflowsto} shows the formal rules. Rules \textsc{BFlow-Super} and \textsc{BFlow-High} lift the standard $\flowsto$ relation to the $\flowsto_b$ relation; rule \textsc{BFlow-Decl} says that a security label can be peeled off once we get to a  safe merge point, namely the immediate postdominator node $b$; note that the block labels match in the conclusion of rule \textsc{BFlow-Decl}.

\begin{figure}[t]
\scalebox{0.85}{
\begin{mathpar}
    \srctbrconditional{}
    
    \srctbrunconditional{}
\end{mathpar}
}
\caption{Typing conditional and unconditional branches. \label{fig:permctrltype}}
\end{figure}

Figure~\ref{fig:permctrltype} shows the typing rules for branches. Rule \textsc{T-Br-Conditional} type-checks a conditional branch by requiring a boolean condition and matching security levels for target blocks. The premise $\trg = \mergepoint{\cur}$ identifies the post-dominator $b_t$. The label $\pc \join \declat{l}{\trg}$ allows permissive control-flow, peeling off $\declat{l}{\trg}$ at the post-dominator, while ensuring $\forall i.\ \pc \join \declat{l}{\trg} \flowsto_{b_i} \Benv(b_i)$. 
\textsc{T-Br-Unconditional} similarly enforces permissive flow to the target label.

Revisiting Figure~\ref{fig:permcflow}, the instruction ${\small{\BrCond{\bluec}{\code{while.end}}{\code{while.body}}}}$ is well-typed  because  $\declat{\top}{\code{while.end}}\, \flowsto_{\code{while.end}} \bot$. Similarly, revisiting  Figure~\ref{fig:implicitflow},  the type system  admits the program under a more permissive $\bot$ policy at the \texttt{merge} block.

\paragraph{Permissive Typing in HOTP.}

Consider extending Figure~\ref{fig:login} so that instead of returning the declassified \texttt{\%status}, the program modifies some secret state (e.g., a bank balance) only if the authentication is successful. Then this operation must take place in a private block (with a high $\pc$), and any subsequent public action (e.g., printing ``transaction complete'') can occur only after the pc is lowered, ensuring that effects observable to the \lattacker are independent of the branch taken.

We conclude the \srclang discussion by formalizing its type system’s security guarantee: it enforces noninterference and gradual release for a \pubattacker (\lattacker with $l=\public$).
While the theorem applies to any \lattacker, we use \pubattacker to represent the more typical attacker.

\begin{theorem}[\srclang Noninterference w.r.t \pubattacker]\label{thm:sirni}
  Let $p$ be a program that does not declassify. If $\stctxp{} \tc p$, then $p$ is noninterfering against a \pubattacker.
  That is, for all initial heaps $\Hspace^0_1$, $\Hspace^0_2$ and variable map $\Vspace$, if:
  \begin{enumerate}
    \item $p$ is well-formed and well-typed; that is, $\stctxp{} \tc p$, and
    \item initial heaps are $\public$-equivalent, that is, $\Hspace^0_1 \simeq_{\public}^{\Menv} \Hspace^0_2$, and
    \item program configurations $\scfgp{h=\Hspace^0_i}$ step to a result configuration $\scfgres{h=\Hspace'_i, c=c_i}$, each emitting a trace $T_i$;
      that is, for all $i \in \{1, 2\}$,  $\scfgp{h=\Hspace^0_i}\ssteptr \scfgres{h=\Hspace'_i,c=c_i} \raises T_i$,
  \end{enumerate}
  then $c_1 = c_2$ and traces are $\public$-equivalent, i.e., $T_1 \simeq_{\public}^{\Menv} T_2$.
\end{theorem}
\begin{proof}
        The proof follows the standard approach introduced by \citet{myers2011while}, building on the soundness proof of \citet{volpano}. When the target of an unconditional branch is a post-dominator, the analysis is more subtle because the $\pc$ may get lowered. For complete proof, refer to Section $\S 1.5$ of the accompanying technical report.
\end{proof}

\begin{theorem}[\srclang Gradual Release w.r.t \pubattacker]\label{thm:sirgr}
  If $\stctxp{} \tc p$, then $p$ satisfies gradual release against a \pubattacker.
  That is, for all initial heaps $\Hspace^0_1$, $\Hspace^0_2$ and variable map $\Vspace$, if:
  \begin{enumerate}
    \item initial heaps are $\public$-equivalent; that is,
    $\Hspace^0_1 \simeq_{\public}^{\Menv} \Hspace^0_2$, and
    \item program configurations $\scfgp{h=\Hspace^0_i}$ step to a result configuration
    $\scfgres{h=\Hspace'_i, c=c_i}$, each emitting a trace $T_i$; that is, for all
    $i \in \{1,2\}$, $\scfgp{h=\Hspace^0_i}\ssteptr \scfgres{h=\Hspace'_i,c=c_i} \raises T_i$,
  \end{enumerate}
  then $c_1 = c_2$ and, for all prefixes $U_1 \preceq T_1$ and $U_2 \preceq T_2$,
  if     $R(U_1) = R(U_2)$,   then $U_1 \simeq_{\public}^{\Menv} U_2$.

\end{theorem}
\begin{proof}
The proof builds on Theorem~\ref{thm:sirni}. The declassification case breaks the proof, however, since $R(U_1) = R(U_2)$, both executions declassify to the same value, thus restoring the trace equivalence.
\end{proof}

\FloatBarrier

  \section{\targetlang: \srclang + ENclaves} \label{sec:target}
  
\targetlang extends \srclang with seven enclave-specific instructions:
$\eenter$, $\eexit$, $\ekill$, $\ecreate$, $\ocall$, $\oret$, and
$\preserve$. Among these, $\preserve$ and $\ocall$ are novel to
\targetlang: they precisely capture the security requirements of TEEs
during enclave--host transitions and are not modeled in prior
work~\cite{gcoopsla2016}.
These instructions are essential for faithfully modeling real-world
TEE execution but introduce additional challenges for the translation.
In particular, $\ocall$ captures interactions with the host (e.g.,
enclaves do not support syscalls), which are common in practical applications. Supporting
these features enables a realistic evaluation; we explicitly
measure the cost of $\ocall$ during execution.

Figure~\ref{syntax-siren} shows only the extended syntax;
the complete syntax appears in \S~3.1 of the accompanying technical report.
The meta-variable $\mu$ denotes the execution mode: $\mu = 0$
corresponds to host execution, while $\mu > 0$ identifies execution
within a specific enclave. The set $\Elive$ tracks the currently active
enclaves (excluding $\mu = 0$, which is always active). Mode
transitions occur only via $\eenter$, $\eexit$, $\ocall$, and $\oret$.
Instruction $\EnclaveCreate{\mu}{\var{b}}$ initializes enclave $\mu$ (if inactive)
and associates block $b$ with it (analogous to how Intel SGX loads code into enclaves),
while $\EnclaveKill{\mu}$ terminates the enclave.
An enclave may be
entered multiple times via $\eenter$, possibly at different blocks. Each block
is assigned a mode $\mu$ in which it executes, and well-formed programs require
enclave blocks to be associated with the enclave before execution.

\begin{figure}[ht]
\scalebox{0.83}{
  \begin{minipage}[t]{0.32\textwidth}
  \centering
  \begin{align*}
    Mode \quad \mu \in&\ \mathbb{N} \\
    Live\ Enclaves \quad \Elive \subseteq&\ \mathbb{N}^{+} \\
    OCall\ status \quad \ocstat \in&\ \{\ocyes, \ocno \} \\
    Ground\ type \quad \tau \ddefeq& \mathrm{Int_1} \mid \mathrm{Int_{32}} \mid \tgtptr{\tau}{\mu} \\
    i \ddefeq& \dots  \mid \Preserve{\var{x}}{\type{\tau}{l} \var{y}} \\
	  & \mid \OCall{\var{x}}{\type{\tau}{l}\fname{f}}{\arglist} \\
	  & \mid \EnclaveCreate{\mu}{ \var{b}} \mid \EnclaveKill{\mu} \mid \EnclaveEnter{ \var{b}} \\
    t \ddefeq& \dots \mid \EnclaveExit \mid \colorbox{gray!50}{\ensuremath{\ORet{\type{\tau}{l} \var{x}}}} \\
    \mathit{block} \ddefeq& b^{\mu}: \mathit{body}\\
    \func \ddefeq& \FDef{\type{\tau}{\lf} \fname{f}^{\mu, \ocstat}}{\paramlist}{\overline{block}} \\
     \mathit{Memory~Spec} \quad \TMenv :&\ m \mapsto (\mu, Type) \\
     \mathit{Heap} \quad \THspace :&\ \bigcup_{\mu \in \Elive~\cup~h} \THspace_\mu \\
     \quad h = \{0\} \\
    \mathit{Typing~context} \quad \TTypeEnv :&\ x \mapsto Type \\
  \end{align*}
  \end{minipage}
}
\caption {\targetlang extensions \label{syntax-siren}}
\vspace{-10pt}
\end{figure}

Instruction $\EnclaveEnter{\var{b}}$ transfers control to enclave block $b$,
and $\EnclaveExit$ restores host mode. However, variables defined inside the
enclave may leak secrets when control returns to the host. In Intel SGX,
developers must explicitly clear such register state---or replace it with a
synthetic state---before exiting the enclave. We capture this requirement with
the instruction $\Preserve{\var{x}}{\type{\tau}{l}\var{y}}$, which records that
the value of variable $y$ is preserved through $x$, so that $x$ may be used
after enclave exit. A well-formedness condition (described later) requires $\preserve$ to appear
immediately before an $\eexit$.


The instruction
$\OCall{\var{x}}{\type{\tau}{l}\fname{f}}{\arglist}$
invokes a host function from within an enclave. It behaves like a
standard $\code{call}$ but performs a mode switch from enclave to host
and restores the enclave context upon return. Such transitions can
introduce subtle re-entrancy vulnerabilities; for example, a host
function may recursively invoke enclave code, leading to logic or
security errors.

To mitigate these issues, each function is annotated with a mode
$\mode$ and an attribute $\ocstat \in \{\ocyes, \ocno\}$. The mode
tracks the execution domain, while $\ocstat$ controls whether the
function may be invoked via an $\ocall$, thereby preventing unsafe
enclave re-entry, a key source of
subtle security vulnerabilities in TEE-based systems.
This design enforces a well-structured call discipline across enclave boundaries.

A function $f$ is \ocstat-callable only if
$\ocstat = \ocyes$. The $\oret$ instruction is automatically generated upon return from
an $\ocall$, restoring control from the host to the enclave.

Pointer operands (not shown) in \code{malloc}, \code{load},
\code{store}, and \code{gep} are annotated with mode types.
As in \srclang, the set of memory locations $\Loc$ includes both heap and stack
locations, indexed by security levels and modes. The meta-variables
$\THspace$, $\TSspace$, and $\TMenv$ denote the corresponding heap,
stack, and memory security specification, respectively.

\smallskip
\noindent \textbf{Well-Formedness.}
To simplify security enforcement, we restrict our attention to well-formed
programs that satisfy additional enclave-related structural constraints,
beyond the standard well-formedness conditions of LLVM programs. A key
property required for our analysis is \emph{variable dominance}, which ensures
that the definition of each variable reaches all of its uses. We 
adapt the standard LLVM block dominance  to host and enclave regions,
and use it to formalize variable dominance.





\begin{figure}[h]
\begin{lstlisting}[style=mystyle, language=llvm, mathescape=true]
b$_0$:
	eenter b$_e$;
	x = $\langle \mathrm{Int_{32}}$, $\bot \rangle$ y;
	ret 0;
b$_e$:
	z = $\langle \mathrm{Int_{32}}$, $\bot \rangle$ 3;
	y = $ \langle \mathrm{Int_{32}}$, $\bot \rangle$ preserve z;
	eexit;
\end{lstlisting}
\caption{Variable dominance: Line~7 dominates line~3.\label{fig:vardom}}
\end{figure}

\begin{definition}[Variable Dominance in \targetlang]\label{def:siren-block-vardom}
A definition of variable $x$ dominates instruction $i$ if one of the
following holds: (i) \textbf{Standard dominance:} the instruction
defining $x$ dominates $i$ in the control-flow graph;
(ii) \textbf{Enclave-exit propagation:} if $x$ is defined by a
\code{preserve} in a block reachable from $b_e$, then the instruction
$\EnclaveEnter{\bLabel \var{b_e}}$ is treated as a definition of $x$ for
dominance purposes; and (iii) \textbf{Enclave entry dominance:} if $i$
occurs in an enclave block $b_i$ and $x$ is defined in a host block,
then the definition of $x$ must dominate every
$\EnclaveEnter{\var{b_e}}$ instruction such that $b_e$ can reach $b_i$.
\end{definition}

\noindent Variable dominance ensures that enclave-local variables do not
escape without an explicit $\preserve$. Consider the program in Figure~\ref{fig:vardom},
which enters enclave block $b_e$ (line~2), assigns a value to $z$ (line~6),
preserves it in $y$ (line~7), and then uses $y$ (line~4) after exiting the
enclave (the control returns to instruction after $\eenter$).
Without line~7, the definition of $y$ would not dominate its use at
line~4. The enclave-exit propagation rule establishes this dominance by
propagating the preserved value to the corresponding enclave entry, allowing
$z$ to be used outside the enclave.

Note that variable dominance alone does not prevent secret values from escaping. For
example, $y$ would still dominate line~4 even if $\public$ were replaced with
$\private$. Such leaks are prevented by the type system.


A \targetlang program is \emph{well formed} if it satisfies standard SSA
discipline—each variable is defined exactly once and its definition
dominates all uses, with \code{phi} instructions placed at block
entries—together with additional enclave-specific constraints, namely
variable dominance (Definition~\ref{def:siren-block-vardom}) and
single-entry, single-exit enclave regions.

While the  first two conditions correspond to standard LLVM well-formedness
requirements; the remaining conditions introduce enclave-specific constraints.
In particular, variable dominance is unique to \targetlang. Unlike prior works~
\cite{gcoopsla2016}, which rely on flow-sensitive typing to track sensitive
variables within an enclave, variable dominance simplifies the type system by
enforcing constraints on $\preserve$ without requiring the tracking of other
variables.

\subsection{Semantics of \targetlang}
The \targetlang small-step semantics
${\tcfgi{i=i}} \longrightarrow_i {\tcfgi{h=\THspace', s=\TSspace',
v=\TVspace', e=\Elive', i=i'}}$ closely mirrors that of \srclang,
with two additional components in the configuration: the set of active
enclaves $\Elive$ and the current execution mode $\mu$. The set
$\Elive$ records all enclaves that have been initialized but not yet
terminated via $\ekill$.

The semantics of \targetlang differ from \srclang in two respects:
(i) memory-access instructions, whose behavior is modified, and
(ii) enclave-specific instructions, which are absent in \srclang.
We focus on these differences, using \code{load} as a representative
memory-access instruction.

Figure~\ref{fig:semantics-siren} presents the evaluation rules for
\code{load} and enclave-specific instructions. The complete semantics
are given in \S~3.2 of the accompanying technical report.

Data within an enclave can only be accessed by code executing in the
same enclave; \targetlang enforces this invariant. Rule~\rulename{E-Load}
retrieves the value $c$ stored at location $m$, pointed to by $y$, and
assigns it to variable $x$. The premise
$\mu_y \in \{\mu, 0\}$ ensures that $m$ resides either in the host
($\mu_y = 0$) or in the same enclave as the current instruction
($\mu_y = \mu$). The premises of the \rulename{E-Store} and
\rulename{E-Gep} rules are analogous and are omitted for brevity.
For example, $\Load{x}{\type{\itt}{\bot}}{\type{\tgtptr{\itt}{1} \var{y}}{\bot}}$ fails in host mode ($\mu=0$) because $\mu_y=1\notin\{0\}$, preventing the host from accessing enclave memory. Conversely, an enclave may access host memory since $0\in\{\mu,0\}$.


\begin{figure}[t]
\scalebox{0.9}{
  \begin{mathpar}
		\trgeload{}

		\trgeecreate{p1={\langle b,\mu \rangle \in \TFspace(\cur)}}

		\trgeeenter{p1={\langle b,\mu \rangle : \body' \in \TFspace(\cur)}}

		\trgeeexit{}
	\end{mathpar}
}
	\caption{\targetlang: Memory access and enclave semantics \label{fig:semantics-siren}}
        \vspace{-20pt}
\end{figure}

Rule~\rulename{E-Ecreate} initializes enclave $\mu$ by inserting block $b$ into its memory space, creating the enclave if it has not yet been activated.
The  premise $\langle b, \mu\rangle \in \TFspace(b_c)$ ensures that the mode of enclave block $b$ matches that of the enclave being created.
The premises $\Elive' = \Elive \cup \{\mu\}$ and  $\THspace' = \THspace \cup \THspace_\mu$ update the live enclaves and the global heap to include the created enclave.

Rule~\rulename{E-Eenter} handles transitions from host to enclave mode.
The current mode must be the host ($\mu = 0$) and the control transfers to enclave block $b$ under mode $\mu$,
Semantically, this resembles a function call but includes additional constraints.
The premise $\mu \in \Elive$ requires that enclave $\mu$ is live; $\fresh{\TStspace, \mu}$ allocates a new stack frame parameterized by the enclave mode; and $\TXspace' = \tcfgxenclave{} :: \TXspace$ records the  caller’s (the host continuation) configuration for later restoration. As with \textsc{E-Ecreate}, we require that $\langle b, \mu\rangle \in \TFspace(b_c)$, ensuring that the mode of $b$ matches $\mu$.

Rule~\rulename{E-Eexit} governs transitions from enclave to host mode.
Intuitively, this corresponds to a function return: the enclave stack is cleared before control resumes in the host.
The premise $\mu \neq 0$ enforces that exits occur only from within an enclave.
The final configuration ensures that upon restoration, control returns to the same host context saved by \rulename{E-Eenter}.
Finally, the premise $\langle b_c', 0\rangle \in \TFspace(b_c)$ confirms that the block being resumed indeed belongs to the host function.

For example, let host block $b_1$ contain $\EnclaveEnter{b_2}$, and let enclave blocks $b_2$ and $b_3$ have mode $1$, with $b_3$ containing $\EnclaveExit$. Before entering the enclave, $b_2$, $b_3$, and any intervening blocks must be associated with enclave $1$ via $\EnclaveCreate{1}{b}$. Executing $\eenter$ switches the mode from $0$ to $1$ and transfers control to $b_2$. Execution remains in enclave mode until $\eexit$, which returns control to the instruction following $\eenter$ in $b_1$ and restores host mode ($\mu=0$).



We briefly outline the semantics of the $\ekill$, $\ocall$, and $\oret$ instructions.
The instruction $\ekill~\mu$ tears down enclave $\mu$, updating the set $\Elive$.
The $\ocall$ instruction is used to make a call to a host-mode function.
Its semantics are similar to those of the $\code{call}$ instruction, except that the execution mode switches to $0$, and an additional frame containing the $\oret$ instruction is pushed onto the call stack $\TXspace$.
Upon returning, $\ocall$ first performs a normal $\code{ret}$, followed by $\oret$, which restores the original enclave execution mode.

\subsection{\targetlang Security Definition}
 The definitions of noninterference and gradual
 release against \lattacker are as in \srclang.
 In addition to \lattacker, we consider a more powerful \eattacker
that can observe the entire host memory. A program is secure against
\eattacker if it cannot distinguish the traces of two \targetlang
configurations with equivalent initial host memories, denoted
$\THspace_1 \simeq_h^{\TMenv} \THspace_2$. Trace equivalence,
$\mathcal{T}_1 \simeq^{\TMenv}_h \mathcal{T}_2$, is defined pointwise.
We require the memory specification $\TMenv$ to be well formed to ensure security against
\eattacker; in particular, no non-enclave memory location may be
labeled \private.

\begin{definition}[\targetlang Noninterference w.r.t. \eattacker] \label{def:trgnieattacker}
  Given a well-formed memory specification $\TMenv$, a well-formed program $p$ is noninterfering with respect to an \eattacker, if for all initial heaps $\THspace^0_1$, $\THspace^0_2$ and a variable map $\TVspace$:
  \begin{enumerate}
  \item the initial heaps are $h$-equivalent; that is, $\THspace^0_1 \simeq_h^{\TMenv} \THspace^0_2$, and
  \item the program configurations $\tcfgp{h=\THspace^0_i}$ step to result configurations $\tcfgres{h=\THspace'_i}$, each emitting the trace $\mathcal{T}_i$; that is, for all $i \in \{1, 2\}$,  $\tcfgp{h=\THspace^0_i}\tsteptr \tcfgres{h=\THspace'_i, c=c_i} \raises \mathcal{T}_i$,
  \end{enumerate}
  then $c_1 = c_2$ such that the traces are $h$-equivalent; that is, $\mathcal{T}_1 \simeq^{\TMenv}_h \mathcal{T}_2$.
\end{definition}

\begin{definition}[\targetlang Gradual Release w.r.t \eattacker]\label{def:gradualrelease}
Given a memory specification $\TMenv$ and a release projection $R(\cdot)$ over traces (defined as in SIR),
a program $p$ satisfies gradual release against a \lattacker\ if for all initial heaps
$\THspace^0_1$, $\THspace^0_2$ and a variable map $\TVspace$:

\begin{enumerate}
\item the initial heaps are $h$-equivalent; that is,
$\THspace^0_1 \simeq_{h}^{\TMenv} \Hspace^0_2$,

\item the program configuration $\tcfgp{h=\THspace^0_i}$ steps to the final configuration
$\tcfgres{h=\THspace'_i, c=c_i}$, emitting the trace $\TT_i$; that is, for all $i \in \{1,2\}$,
$
\tcfgp{h=\THspace^0_i}
\tsteptr
\tcfgres{h=\THspace'_i, c=c_i}
\raises \TT_i,
$
\end{enumerate}

then $c_1 = c_2$ and for all prefixes $U_1 \preceq T_1$ and $U_2 \preceq T_2$, if
$R(U_1) = R(U_2)$,
then $U_1 \simeq_{h}^{\Menv} U_2$.
\end{definition}

\subsection{\targetlang Type System} \label{sec:targettypes}

The typing judgment $\ttctxb{} \tc i$ states that instruction $i$,
executing in mode $\mode$, is well typed under the typing context
$\ttctxb{}$. This context extends that of \srclang with $\mode$ and
$\ocstat$, which track the current execution mode and the ocall status
of the enclosing function. Most typing rules mirror their \srclang
counterparts, with one additional premise:
$\mu \neq 0 \lor l \sqsubseteq \public$, which enforces that any
instruction operating on non-public data executes within an enclave.
Consequently, private data can only be manipulated during enclave
execution. Additionally, $\ocall$ checks $\ocstat$ compatibility and ensures that
the callee executes in the host, while regular function calls require
mode compatibility.

As with the semantics, we focus on memory-access instructions and
enclave-specific constructs. Figure~\ref{fig:trgtyping} presents the
typing rules corresponding to the semantics in
Figure~\ref{fig:semantics-siren}. The complete typing rules are given
in \S~3.4 of the accompanying technical report.

\begin{figure}[t]
\scalebox{0.9}{
    \begin{mathpar}
      \trgtload{}
      \trgteexit{}

      \trgteenter{}
      \trgtecreate{}
    \end{mathpar}
    }
    \caption{\targetlang: typing load and enclave instructions.\label{fig:trgtyping}}
    \vspace{-20pt}
\end{figure}


Rule \rulename{Tt-Load} requires that an enclave pointer can only be stored with a value when the execution mode matches.
The premise $(\mu = \mu_y) \lor (\mu_y = 0)$ ensures that if the pointer belongs to some enclave, then either the current execution mode matches the pointer mode, or the pointer belongs to the host. Note that this premise does not restrict storing to host pointers even if the current execution mode is not the host ($\mu \ne 0$).
This models the Intel SGX constraint that enclave memory is accessible only to code within the same enclave—while host code is denied access, enclave code retains access to host memory. The remaining premises are similar to those in \rulename{T-Load}, the \srclang counterpart for the \code{store} instruction.

Rule~\rulename{Tt-Eenter} specifies the conditions under which an
enclave can be entered from the host. The premise $\mu = 0$ ensures
that transitions originate only from host mode, reflecting the Intel
SGX restriction that disallows direct enclave-to-enclave transitions.
Since $\eenter$ behaves like an unconditional branch, the additional
premise $\pc \sqsubseteq l_b$ prevents implicit flows arising from
jumps from private to public blocks.

Rule~\rulename{Tt-Eexit} ensures that enclave exits occur in the
correct execution mode. Rule~\rulename{Tt-Ecreate} requires the current
execution mode to be host ($\mu = 0$). Additionally, enclave creation
is disallowed in secret contexts (premise $pc\sqsubseteq \public$), preventing enclave initialization
from being influenced by sensitive control flow.

\paragraph{\targetlang Typing HOTP.}

In Figure~\ref{fig:login-llvm}, secret variables such as \code{\%expected} and \code{\%status}, and instructions operating on them, must reside in the enclave; public code may execute in either the host or enclave. Enclave placement alone, however, does not ensure security. For example, \code{store i32 \%expected, \%ptr}, where \code{\%ptr} points to host memory, leaks \code{\%expected} to the \eattacker, since enclaves may write to host memory. SIREN's security guarantees therefore arise from combining enclave isolation with IFC enforcement.

\paragraph{Type Soundness} The \targetlang's type system is sound---a well-typed program is secure against both \lattacker and \eattacker threat models.  Note that we require $\TMenv$ to be well-formed for security against an \eattacker.  Theorems~\ref{thm:siren-ni-high} and \ref{thm:siren-ni-low} formally state the soundness of type system. For simplicity, we omit  gradual release version of the theorems.

\begin{theorem}[\targetlang Noninterference w.r.t a \pubattacker] \label{thm:siren-ni-high}
  Let $p$ be a program that does not declassify. If $\srctypecheckcom{\ttctxp{}}{p}$ then $p$ is noninterfering with respect to a \pubattacker for any well-formed specification $\TMenv$.\end{theorem}
\begin{proof}
 The proof is analogous to that of  \Cref{thm:sirni} (\nameref{thm:sirni}). For complete proof, refer to Section $\S 3.5$ of the accompanying technical report.
 \end{proof}

\begin{theorem}[\targetlang Gradual Release w.r.t a \pubattacker] \label{thm:siren-gr-high}
  If $\srctypecheckcom{\ttctxp{}}{p}$ then $p$ is satisfies gradual release with respect to a \pubattacker for any well-formed specification $\TMenv$.\end{theorem}
\begin{proof}
 The proof is analogous to that of  \Cref{thm:sirgr} (\nameref{thm:sirgr}).
 \end{proof}

\begin{theorem}[\targetlang Noninterference w.r.t an \eattacker] \label{thm:siren-ni-low}
  Let $p$ be a program that does not declassify. If $\srctypecheckcom{\ttctxp{}}{p}$ then $p$ is noninterfering with respect to an \eattacker for any well-formed specification $\TMenv$.
  Formally, for all $\TVspace$, $\THspace^0_1$, and $\THspace^0_2$, if:
\begin{enumerate}
\item The program is well-formed and well-typed; i.e., $\srctypecheckcom{\ttctxp{}}{p}$,
\item Initial heaps are $h$-equivalent and $l$-equivalent; i.e., $\THspace^0_1 \simeq_{h} \THspace^0_2$ and $\THspace^0_1 \simeq_{\public}^{\TMenv} \THspace^0_2$,
\item Program configuration $\configthree{\THspace^0_i}{\TVspace}{p}$ steps to the final configuration $\scfgres{h=\THspace'_i, c=c_i}$, emitting a trace $\Ttrace_i$; i.e.,
 for all $i = \{1, 2\}$,  $\srcsmallsteptrace{\configthree{\THspace^0_i}{\TVspace}{p}}{\configtwo{\THspace'_i}{c_i}}{\Ttrace_i}$
\end{enumerate}
then $c_1 = c_2$ such that the traces are $h$-equivalent; i.e., $\Ttrace_1 \simeq^{\TMenv}_{h} \Ttrace_2$.
\end{theorem}
\begin{proof}
 The proof builds on the top of \Cref{thm:siren-ni-high} (\nameref{thm:siren-ni-high}). It heavily relies on two crucial invariants enforced both by the small-step semantics and the security type system. First, a non-public memory location always belongs to some enclave. Second, enclave locations are accessed  by the same enclave; otherwise the program aborts. Thus, the host memory is indistinguishable at every step. Note that our security definition is termination-insensitive. For complete proof, refer to Section $\S 3.5$ of the accompanying technical report.
\end{proof}

\begin{theorem}[\targetlang Gradual Release w.r.t an \eattacker] \label{thm:siren-gr-low}
 If $\srctypecheckcom{\ttctxp{}}{p}$ then $p$ satisfies gradual release with respect to an \eattacker for any well-formed specification $\TMenv$.
\end{theorem}
\begin{proof}
 The proof builds off of \Cref{thm:siren-ni-low} (\nameref{thm:siren-ni-low}) in a manner similar to \Cref{thm:sirgr} (\nameref{thm:sirgr}) and \Cref{thm:siren-gr-high} (\nameref{thm:siren-gr-high}).
\end{proof}

  \section{Translation} \label{sec:translation}
  A key contribution of this work is to infer enclave placement in
\srclang programs in a secure-by-construction manner: translating a
well-typed \srclang program yields a well-typed \targetlang program.
The soundness of the type system (Theorem~\ref{thm:siren-ni-low})
then guarantees security.

A central challenge in type-preserving translation is reconciling the
differences between \srclang and \targetlang. A naive translation that
assigns enclave mode to all pointers and instructions not only leads to
inefficient use of enclave resources (e.g., limited enclave memory) but
may also violate security requirements, as enclave-resident secrets must
be explicitly scrubbed at exit points, and may be infeasible in practice
due to the need for $\ocall$s.
In contrast, our translation
captures only the minimal set of constraints required for security. In
particular, it ensures that secret pointers are placed in enclaves and
accessed only by code executing within the same enclave. Since multiple
valid translations may exist, we parameterize the translation with a
user-defined objective function (e.g., minimizing enclave code size) to
select a solution with minimal cost.

Translation proceeds in two phases. In the first phase, a well-typed \srclang program emits constraints and an intermediate \targetlangminus program, identical to the input \srclang program but with modes attached to pointers, blocks, instructions and functions; if the mode annotations are ignored, the generated \targetlangminus program is a well-typed \srclang program. Unifying constraints yields an enclave assignment for the \targetlangminus program.
In the second phase, the \targetlangminus program is transformed (by inserting host-enclave transition instructions) into a full-fledged \targetlang program, guided by the enclave assignment.

\subsection{Step One: Constraint Generation} \label{sec:stepone}

\begin{figure}[ht]
\centering
\scalebox{0.9}{
  \begin{mathpar}
    \centering
\inferrule*[Lab=TR-Program]
	       {
                 \translateSrc{\Menv}{\TMenv} \\
                 \translateSrc{\Benv}{\TBenv} \\
                 \translateSrc{\Fenv}{\TFenv} \\
                 \translateSrc{\TypeEnv}{\TTypeEnv} \\\\
                 \translateSrc{\TypeEnv_k}{\TTypeEnv_k} \\
                 p = \overline{f}\\
                 p' = \overline{f'}\\\\
                 \forall~k.~\translateSrc{\srctypecheckcom{\transrccontext_k , \trantrgcontext_k}{f_k}}{\withconstraints{f'_k}{\constr_k}}  \\
                             \constr = \bigcup_k \constr_k
                }
	        {\translateSrc{\srctypecheckcom{\transrccontext, \trantrgcontext}{p}}{\configtwo{p'}{\constr}}}
\end{mathpar}}
  \caption{Translation of program.}
  \label{fig:trmemspec}
  \vspace{-10pt}
\end{figure}

Rule \rulename{TR-Program} states that a well-typed \srclang program  emits \targetlangminus program $p'$ generating constraints $\constr$. Premise  $\translateSrc{\srctypecheckcom{\transrccontext_{k}, \trantrgcontext_{k}}{f_k}}{\withconstraints{f'_k}{\constr_k}}$ translates each  function $f_k$ to $f'_k$ and emits constraint set $\constr_k$. The final constraint, $\constr$ is the union of all intermediate constraint sets. Note that  $\constr$ may yield zero or more  solutions.

The premise $\translateSrc{\Menv}{\TMenv}$ translates \srclang memory specification $\Menv$ to $\TMenv$ and ensures that all private location to some (possibly distinct) enclave; while the remaining premises   $\translateSrc{\Benv}{\TBenv}$, $\translateSrc{\Fenv}{\TFenv}$, $\translateSrc{\TypeEnv}{\TTypeEnv}$ translate  \srclang block level map $\Benv$,  function definition map $\Fenv$  and variable typing context $\TypeEnv$ to corresponding \targetlang contexts.
Translation generates fresh mode variables for blocks, functions, and reference types, while the typing context is recursively translated to propagate these mode annotations consistently.

The translation of a function  proceeds by translating blocks that in turn proceeds by translating individual instructions. Notably, the block translation yields $\overline{\mode}$, a sequence of mode assignments for the instructions of the block. The mode of the block itself is the least of the instruction modes: if there are different instruction modes, then $\mode_{c}$ is the host; otherwise, it is equal to the mode of all instructions.

The instruction translation judgment
$\translateSrcInt{\srctypecheckcom{\transtypingcontextbody}{i}}{\withconstraints{\mode, i'}{\constr}}$
states that a well-typed \srclang instruction $i$ is translated to a
\targetlang instruction $i'$ that executes in mode $\mode$ and
satisfies the constraints $\constr$. All translation rules include the
common constraints $\pc \ne \public \Rightarrow \mu \ne 0$ and
$\mode_f \ne 0 \Rightarrow \mode = \mode_f$. The former ensures that
instructions in a secret context are placed within an enclave, while
the latter ensures that the execution mode of the instruction matches
that of the enclosing function.

Function calls introduce additional constraints: either the caller and
callee execute in the same mode, or, if the caller executes within an
enclave, the callee executes in the host. In the former case, the call
is translated as a standard function call; in the latter, it is
translated to $\ocall$.

Two classes of instructions deserve special attention: memory-access
instructions and conditional branches. For the former, the translation
must generate constraints that ensure pointer modes match the execution
mode of the instruction. For the latter, we adopt a conservative
strategy requiring that the source and successor blocks have the same
execution mode. This prevents undesirable cases such as multiple
enclave entry points. Figure~\ref{fig:trloadrw} shows the translation rules for load and
conditional branch instructions; the complete set of translation rules
is presented in \S~5.1 of the accompanying technical report.

To illustrate the issue with pointer modes, consider our running example (Figure~\ref{fig:login-llvm}). The public pointer \code{\%ctr\_loc} may point to either host or enclave memory. If it points to enclave memory, any instruction accessing it must execute in the enclave, since any attempt to access it from the host will fail; if it points to host memory, the instruction may execute in either mode, since enclaves can access host memory.

\begin{figure}[t]
\scalebox{0.85}{
\begin{mathpar}

         \inferrule*[Lab=TR-Load]
		{\TTypeEnv(y) = (\_, \mode', \_) \\
                 \TFenv(f) = (\mode_f, \ocstat_f) \\\\
                 i' = \Load{\var{x}}{\type{\tau}{l_1}}{\type{\ptr{\tau}, \mode'}{l_2} \var{y}} \\\\
		  \constr =
                          {\left\{\!\begin{aligned} 
                            & \pc \ne \public \Rightarrow \mode \ne 0, \\
                            & l_2 \ne \public \Rightarrow \mode' \ne 0, \\
		            & \mode'\ne0 \Rightarrow \mode = \mode' \\
                               & \mu_f \ne 0 \Rightarrow \mu = \mu_f \\
                             \end{aligned} \right\}
                          }
                   }
		{\translateSrcInt{\srctypecheckcom{\transtypingcontext}{\Load{\var{x}}{\type{\tau}{l_1}}{\type{\ptr{\tau}}{l_2} \var{y}}}}
		                          {\withconstraints{\mode, i'}{\constr}}}

\inferrule*[Lab=TR-Br-Conditional]
		           {\translateSrc{e}{e'} \\
                             \TFenv(f) = (\mode_f, \ocstat_f) \\\\
                             \TBenv(b_i) = (\mode_i, \_) \\
                             i' = \BrCond{\type{\ione}{l}e'}{ \var{b_1}^{\mu_1}}{ \var{b_2}^{\mu_2}} \\\\
                             \constr =
                          {\left\{\!\begin{aligned} 
                            & \pc \sqcup l \ne \public \Rightarrow \mode \ne 0 \\
                            &  \mu_1 = \mu_2 = \mu \\
                             & \mu_f \ne 0 \Rightarrow \mu = \mu_f \\
                            \end{aligned} \right\}
                          }
		 }
		{\translateSrcInt{\srctypecheckcom{\transtypingcontextbody}{\BrCond{\type{\ione}{l}e}{ \var{b_1}}{ \var{b_2}}}}
		                          \withconstraints{\mode, i'}{\constr}}
\end{mathpar}
}
  \caption{Translation of load and conditional branch. \label{fig:trloadrw}}
  \vspace{-15pt}
\end{figure}

Rule~\rulename{TR-Load}  emits
constraints ensuring that the mode of a private pointer $y$ is enclave,
i.e., $l_2 \ne \public \Rightarrow \mode' \ne 0$. The constraint
$\mode' \ne 0 \Rightarrow \mode = \mode'$ further requires that an
enclave pointer is accessed only by instructions executing in the same
enclave. Notably, these constraints do not prevent a host pointer from
being accessed within an enclave, which accurately reflects the memory
access controls enforced by Intel SGX. The translation for store is analogous.

Rule \rulename{TR-Br-Conditional} emits constraints related to the condition and target block labels. Constraint $\pc \sqcup l \ne \public \Rightarrow \mu \ne 0 $ ensures that if the condition is private, then the instruction executes inside an enclave. The constraint $\mu_1 = \mu_2 = \mu$ requires that the current and the target blocks execute in the same mode. The latter constraint is not strictly necessary but simplifies the enclave placement in the presence of complex control-flow.

\subsection{Step Two: Program Transformation}\label{sec:steptwo}

\begin{figure}[ht]
\scalebox{0.9}{
\begin{mathpar}
\inferrule*[lab=ERw-Eexit]
                 {\Bmode{b} \neq (\mode, \_) \\ \mode\neq 0}
                 {\translatePostProcess{\configthree{\TBenv}{\mode}{\Br{\var{b}}}}
                                                     {\code{eexit}}}

  \inferrule*[lab=ERw-Eenter]
                 {\Bmode{b} \neq (\mode, \_) \\ \mode = 0 \\\\
                  b' = \modepdom{b} \\
                  \translatePostProcess{\configthree{\TBenv}{\mode}{\Br{{b'}}}}
                                                     {\mathit{rest}}}
                 {\translatePostProcess{\configthree{\TBenv}{\mode}{\Br{{b}}}}
                                                     {\Seq{\code{eenter}~{\var{b}}}{\mathit{rest}}}}

\end{mathpar}
}
\caption{Rewrite rules for enclave transition \label{fig:enclavetrans}}
\vspace{-10pt}
\end{figure}

The second phase inserts enclave instructions at mode-transition
points. We abstract this step as
$\code{postprocess}(p^-, \TTypeEnv^-, \constr)$, where $p^-$,
$\TTypeEnv^-$, and $\constr$ are produced by the first phase. It first
splits basic blocks so that each contains instructions of a single
execution mode, and then rewrites unconditional branches as
$\code{eenter}$ or $\code{eexit}$ instructions.

To preserve control flow, we identify a \emph{mode-exit} block that is
reachable from all enclave blocks and insert branches to it after each
$\code{eenter}$. Intuitively, the mode-exit block acts as a join point
that ensures all enclave executions return to a consistent continuation
in the host.

The rewriting rules \rulename{RW-Eenter} and
\rulename{RW-Eexit} (Figure~\ref{fig:enclavetrans}) ensure correct
continuation across mode transitions. Finally, the program is rewritten
to insert $\preserve$ instructions in $\eexit$ blocks, and the resulting
definitions are propagated to their uses, yielding a structurally
well-formed \targetlang program, ensuring well-structured control flow
across enclave boundaries.

\newcommand{\mytikzmark}[1]{%
  \tikz[overlay,remember picture,baseline] \coordinate (#1) at (0,0) {};}

\newcommand{\highlight}[4]{%
  \draw[fill={green!40},opacity=0.25]%
    ([xshift=-3pt, yshift=#3]#1) rectangle ([yshift=#4]#2);%
}

\newcommand{\highlightclear}[4]{%
  \draw[opacity=0.25]%
    ([xshift=-3pt, yshift=#3]#1) rectangle ([yshift=#4]#2);%
}

\begin{figure*}[t]
\begin{subfigure}[c]{0.46\textwidth}
  \scalebox{0.83}{
\begin{lstlisting}[style=mystyle, language=llvm, mathescape=true]

define i32$^\bot$ @verifyOtp(i32$^\bot$ %submitted, i64$^\bot$ %ctr) $(\bot)$ {
entry:
  %submitted_loc = alloca i32$^\bot$
  %ctr_loc = alloca i64$^\bot$
  %expected_loc = alloca i32$^\top$
  %status_loc = alloca i32$^\top$
  store i32$^\bot$ %submitted, ptr$^\bot$ %submitted_loc
  store i64$^\bot$ %ctr, ptr$^\bot$ %ctr_loc
  ;; expected = expectOtp(ctr)
  %ctr_1 = load i64$^\bot$, ptr$^\bot$ %ctr_loc
  %expected = call i32$^\top$ @expectedOtp(i64$^\bot$ %ctr_1)
  store i32$^\top$ %expected, ptr$^\top$ %expected_loc
  ;; status = ct_compare(expected, submitted)
  %expected_1 = load i32$^\top$, ptr$^\top$ %expected_loc
  %submitted_1 = load i32$^\bot$, ptr$^\bot$ %submitted_loc
  %status = call i32$^\top$ @ct_compare(i32$^\top$ %expected_1, i32$^\bot$ %submitted_1)
  store i32$^\top$ %status, ptr$^\top$ %status_loc
  ;; return declassify_i32(status)
  %status_1 = load i32$^\top$, ptr$^\top$ %status_loc
  %out = declassify i32$^\bot$, i32$^\top$ %status_1
  ret i32$^\bot$ %out
}
\end{lstlisting}
  }
  \caption{\code{verifyOTP()} in \srclang \label{fig:login-llvm-tl}}
\end{subfigure}
\vspace{-10pt}
~
\begin{subfigure}[c]{0.52\textwidth}
  \centering
  \scalebox{0.83}{
\begin{lstlisting}[style=mystyle, language=llvm, mathescape=true, escapechar=~]

define i32$^\bot$ @verifyOtp(i32$^\bot$ %sbmt, i64$^\bot$ %ctr) $(\bot)$ {
~\mytikzmark{hl3Start}~entry.host: ;; $\mu = 0$ (host)
  $\colorbox{highlightcolor}{ecreate}$ $\mu=1$, label %entry.enclave.0
  $\colorbox{highlightcolor}{ecreate}$ $\mu=1$, label %entry.enclave.1
  %submitted_loc = alloca i32$^\bot$~\colorbox{highlightcolor}{, $\mu=0$}~
  %ctr_loc = alloca i64$^\bot$~\colorbox{highlightcolor}{, $\mu=0$}~
  $\colorbox{highlightcolor}{eenter}$ label %entry.enclave.0
  store i32$^\bot$ %submitted, ptr$^{\bot}_{\mu=0}$ %submitted_loc
  store i64$^\bot$ %ctr, ptr$^\bot_{\mu=0}$ %ctr_loc
  %ctr_1 = load i64$^\bot$, ptr$^\bot_{\mu=0}$ %ctr_loc
  $\colorbox{highlightcolor}{eenter}$ label %entry.enclave.1
  ;; return declassify_i32(status)
  ret i32$^\bot$ %out                           ~\mytikzmark{hl3End}~

~\mytikzmark{hl1Start}~entry.enclave.0: ;; $\mu = 1$ (enclave 1)
  %expected_loc = alloca i32$^\top$~\colorbox{highlightcolor}{, $\mu=1$}~
  %status_loc = alloca i32$^\top$~\colorbox{highlightcolor}{, $\mu=1$}~
  $\colorbox{highlightcolor}{eexit}$                                   ~\mytikzmark{hl1End}~

~\mytikzmark{hl2Start}~entry.enclave.1; ;; $\mu = 1$ (enclave 1)
  ;; expected = expectOtp(ctr)
  %expected = call i32$^\top$ @expectedOtp(i64$^\bot$ %ctr_1)
  store i32$^\top$ %expected, ptr$^\top_{\mu=1}$ %expected_loc
  ;; status = ct_compare(expected, submitted)
  %expected_1 = load i32$^\top$, ptr$^\top_{\mu=1}$ %expected_loc
  %submitted_1 = load i32$^\bot$, ptr$^\bot_{\mu=0}$ %submitted_loc
  %status = call i32$^\top$ @ct_compare(i32$^\top$ %expected_1, i32$^\bot$ %submitted_1)
  store i32$^\top$ %status, ptr$^\top_{\mu=1}$ %status_loc
  ;; declassify_i32(status)
  %status_1 = load i32$^\top$, ptr$^\top_{\mu=1}$ %status_loc
  %out = declassify i32$^\bot$, i32$^\top$ %status_1
  $\colorbox{highlightcolor}{eexit}$                                 ~\mytikzmark{hl2End}~
}
\end{lstlisting}
  }
  \caption{\code{verifyOTP()} in \targetlang \label{fig:login-siren}}
\end{subfigure}%
  \begin{tikzpicture}[remember picture, overlay]
     \highlight{hl1Start}{hl1End}{2pt}{-2pt}
     \highlight{hl2Start}{hl2End-|hl1End}{12pt}{26pt}
     \highlightclear{hl3Start}{hl3End-|hl1End}{-22pt}{-10pt}
  \end{tikzpicture}
\caption{Translation of HOTP.}
\end{figure*}

\noindent
\paragraph{HOTP: End-to-End Translation.}

Figure~\ref{fig:login-llvm-tl} shows the original \srclang code for initiatelogin(). First, we generate mode constraints. Each line gets a mode $\mu$, and each line that uses a load, store or, alloca instruction gets a mode $\mu'$ for its associated pointer. There are four pointers, allocated on lines~4--7. The first two, \code{submitted\_loc} and \code{ctr\_loc}, are public pointers and may go in either the host or the enclave. The second two, \code{expected\_loc} and \code{status\_loc}, are private and must go in the enclave. The constraints require that each pointer have the same associated $\mu'$ each time it is used (i.e. if it goes in the host, it must always be in the host). Additionally, each instruction must satisfy $\mu' \neq 0 \Rightarrow \mu=\mu'$, requiring an instruction to execute in the pointer's enclave whenever the pointer is not in host memory.

Lines~12--20 directly involve secret registers and therefore require enclave execution. Lines 6 and 7 allocate private pointers, so they must also run in an enclave. Other instructions are constrained by $\pc\neq\bot\Rightarrow\mu\neq0$ and $\mu_f\neq0\Rightarrow\mu=\mu_f$. Here, the first is vacuous because $\pc=\bot$ throughout; the second matters only if the function mode $\mu_f$ is nonzero, in which case all instructions must execute in enclave $\mu_f$.

Next, we choose any enclave assignment satisfying these constraints. One simple allocation is to put lines 6--7, and 12--20 in an enclave, and the rest in the host. Then, since the lines allocating the public pointers will run in host mode, we must put those in the host as well, and the private pointers must go in the enclave because they are private.

The translated program, shown in Figure~\ref{fig:login-siren}, introduces enclave blocks \code{entry.enclave.0} and \code{entry.enclave.1}, along with $\ecreate$, $\eenter$, and $\eexit$ instructions to ensure proper host-enclave translations. Execution begins in host block \code{entry}, enters \code{entry.enclave.0}, which runs in enclave mode, then returns to \code{entry} after the first $\eenter$ and continues in host mode. It then enters \code{entry.enclave.0}, which runs in enclave mode, and then returns to to \code{entry} after the second $\eenter$, and runs in host mode until the end of the program. This is only one valid solution; other satisfying assignments yield different but equally correct and secure \targetlang programs.

\subsection{Type-preserving Translation}\label{sec:transpreserve}
 Heap translation $\translateSrcHeap{\TMenv}{\Hspace}{\THspace}$, the last missing piece, works by mapping each \srclang location to a \targetlang sub-heap consistent with its mode and security level in the target memory specification. In particular, all private \srclang sub-heaps are translated to enclave sub-heaps in \targetlang, ensuring that sensitive data resides only within enclave memory regions.

The translation is type preserving: translating a well-typed
\srclang configuration yields a well-typed \targetlang configuration.
This property is formalized in
Theorem~\ref{thm:transtypepreserve}
(\nameref{thm:transtypepreserve}). 

  \begin{theorem}[Type Preserving Translation]\label{thm:transtypepreserve}
 Let a well-formed and well-typed \srclang $p$ be translated to \targetlang $p^{e}$ program such that the following hold.
\begin{enumerate}
\item The \srclang configuration is well-typed; i.e., $\Menv, \TypeEnv \tc  \scfgp{}$.
\item The \srclang program $p$ is translated to the \targetlangminus program $p^{-}$; i.e., $\translateSrc{\srctypecheckcom{\TypeEnv, \TTypeEnv^{-}}{p}}{\configtwo{p^{-}}{\constr}}$ (see \rulename{TR-Program}).
\item The \targetlang memory specification $\TMenv$ is well-formed; i.e., $\vdash \TMenv$.
\item The \srclang heap $\Hspace$ is translated to \targetlang heap $\THspace$; i.e., $\translateSrcHeap{\TMenv}{\Hspace}{\THspace}$.
\item The \srclang variable map $\Vspace$ is translated to \targetlang variable map $\TVspace$; i.e., $\translateSrcHeap{\TMenv}{\Vspace}{\TVspace}$.
\item Post-processing \targetlangminus program $p^{-}$ yields $p'$; i.e.,  $\code{postprocess}(p^{-}, \TTypeEnv^{-}, \constr) = \langle p^e, \TTypeEnv \rangle $.
\end{enumerate}
Then the translated program configuration is well-typed; i.e., $\TMenv, \TTypeEnv \tc  \tcfgp{p=p^e}$.
\end{theorem}
Proof outline is presented in  $\S 5.3$ of the accompanying technical report. Note that the well-typed \targetlang configuration is secure against \pubattacker and \eattacker attackers (see Theorems~\ref{thm:siren-ni-high} and~\ref{thm:siren-ni-low}), and since Theorem~\ref{thm:transtypepreserve} guarantees that the translated program is well-typed, we have that the translation enforces security against both \pubattacker and \eattacker. Also, note that since $p$ satisfies gradual release, so does $p^e$. We leave the mechanized
verification of Theorem~\ref{thm:transtypepreserve} to future work.

Theorem~\ref{thm:transtypepreserve} does not state that the translation is semantics-preserving;
establishing this formally requires careful reasoning about
semantics-preserving control flow and is left for future work.
Nevertheless, our evaluation provides empirical evidence that the
translation preserves semantics.

  \section{Implementation}\label{sec:impl}
  
Our implementation consists of two major components: \infertool, which infers a permissive, complete security policy from sparsely annotated IR, and \splitacronym (pronounced ``splitter''), which consumes the annotated IR and generates an enclave--host partition.

\subsection{\infertool}
\infertool infers a complete security policy in five phases. First, it parses and links the IR files into a single module for whole-program analysis. Second, it constructs the call graph and identifies strongly connected components (SCCs). Third, it processes SCCs bottom-up, solving each independently and reducing it to a summary of its interface with the rest of the program (e.g., that information flows from an argument to the return value). Fourth, it performs a global solve over function signatures and SCC summaries, incorporating seed annotations to determine each function's argument, return, and $\pc$ labels. Omitting function bodies from this solve enables \infertool to scale to large programs such as OpenSSL. Finally, it re-solves each SCC using these concrete labels to annotate individual instructions.

During the final phase, \infertool employs a \emph{repair loop} to resolve conflicting constraints. For example, if a private function \code{foo} is called from a public function \code{bar}, \infertool presents the developer with the UNSAT core, who can then \emph{declassify} \code{foo}'s return value, \emph{promote} \code{bar} to private, or \emph{clone} \code{foo} into public and private versions with separate signatures.

\textbf{Security Annotation Effort.}
Manually annotating programs beyond a few hundred IR instructions is impractical. Instead, developers provide a small number of \emph{seed} annotations marking where secrets enter the program (e.g., the private-key argument to \code{RSA\_private\_decrypt}); \infertool propagates these labels according to the \srclang typing rules to infer a complete security policy.

Although preliminary and non-optimizing, \infertool annotates all our benchmarks, including OpenSSL (${\sim}$400k instructions), in under 10 minutes, requiring at most one manual annotation per distinct confidentiality level. Many benchmarks require only a single annotation specifying the initial security context. More complex policies may require additional annotations, but we expect this effort to remain small relative to program size. Future work will support a larger class of instructions (e.g., indirect calls), improve automation and usability, including source-level annotations and IDE integration.

\subsection{\splitacronym}
\splitacronym is an LLVM 20 new-pass-manager module pass that partitions annotated LLVM IR produced by \infertool into OpenEnclave host and enclave components. It faithfully implements the constraint-based translation of Section~\ref{sec:translation}, deriving per-instruction placement constraints and solving them using a built-in max-flow min-cut solver. The solver models each LLVM instruction as a node, with the source and sink representing host and enclave placement, respectively. Weighted edges encode placement and transition costs; the minimum cut assigns instructions to the host or enclave while minimizing the selected objective.

\splitacronym supports four primary objective functions: \emin, which minimizes TCB size; \tmin, which minimizes enclave--host transitions; \objboth, which minimizes both \emin and \tmin; and \objdata, which minimizes data transferred across the enclave--host boundary.

By targeting LLVM IR, \splitacronym is agnostic to the source language (e.g., C, C++, or Rust) and can be extended to additional TEE backends, such as ARM TrustZone and RISC-V enclaves. The current implementation targets partitioning to a single enclave: each instruction has a boolean mode variable, enabling forced-set propagation or max-flow min-cut rather than a general integer solver. Partitioning to multiple enclaves is an interesting future work.

\splitacronym forces code accessing secret data into the enclave. However, SGX enclaves cannot directly execute host-only operations such as system calls, network or file I/O, host-library calls, or freeing host-allocated memory; TEEs such as Intel TDX that support conventional OS interactions can simplify this constraint. \splitacronym therefore uses \emph{host pinning}: host-only in-module code, unless security-forced into the enclave, is outlined into a host helper invoked through an \ocall\xspace that marshals its operands. External calls are not automatically converted to \ocall s; instead, functions containing them are pinned to the host, where the corresponding libraries are linked. For data transfer, \splitacronym flattens \texttt{argv} into a contiguous representation and deep-copies flat buffers of known length; nested pointers are avoided across the boundary by co-locating code with their data.

\textbf{Limitations.}
\splitacronym uses LLVM's CodeExtractor, which requires single-entry regions. When a fine-grained partition cannot be outlined or its boundary data marshalled, \splitacronym places the enclosing function in the enclave, yielding potentially coarser partitions. It deep-copies flat, known-size, non-opaque data across the enclave boundary; for nested, unsized, or handle-like data, it co-locates code and data or emits hints for OpenEnclave handling. Unsupported LLVM instruction classes (e.g., exceptions, globals and atomics) are treated permissively. Finally, \splitacronym does not enforce coarse-grained memory safety; we leave this to future work studying the interplay between memory safety and information flow control.

  \section{Evaluation} \label{sec:eval}
  
\paragraph{Benchmarks and Experimental Setup.}
We evaluate \splitacronym on benchmarks spanning cryptographic, data-processing, machine-learning, and systems workloads. For each benchmark, we report the total number of LLVM IR instructions (\#IR) and \emph{Priv\%$_\mathrm{o}$}, the percentage of instructions whose own inferred security label is private. Our suite includes \code{BFS}, a breadth-first graph traversal; \code{BTree}, a B-tree index; \code{ChaCha20}, a stream cipher; \code{LibreSSL} and \code{OpenSSL}, cryptographic and TLS workloads; \code{HOTP}, an OTP-authentication application; \code{LZAV}, a Lempel--Ziv compressor; \code{Memcached}, an in-memory key--value store; \code{MonteCarlo}, a Monte Carlo computation; \code{Quicksort}, an in-memory sorting workload; and \code{SVM Predict}, \code{SVM Scale}, and \code{SVM Train}, which exercise SVM prediction, data scaling, and training, respectively. \code{BFS}, \code{Memcached}, \code{OpenSSL}, and the SVM workloads are drawn from SGXGauge. The benchmarks range from small kernels to large applications: for example, \code{LZAV} contains 1,455 LLVM IR instructions, of which 89.48\% are labeled private, whereas \code{OpenSSL} contains 425,953 instructions but only 0.16\% are labeled private. This diversity allows us to evaluate \splitacronym across programs with substantially different sizes and security footprints. All experiments were conducted on Intel SGX hardware provisioned through IBM Cloud, running Ubuntu 20.04 with 8 vCPUs and 40\,GB of memory (at most 5\,GB is used).

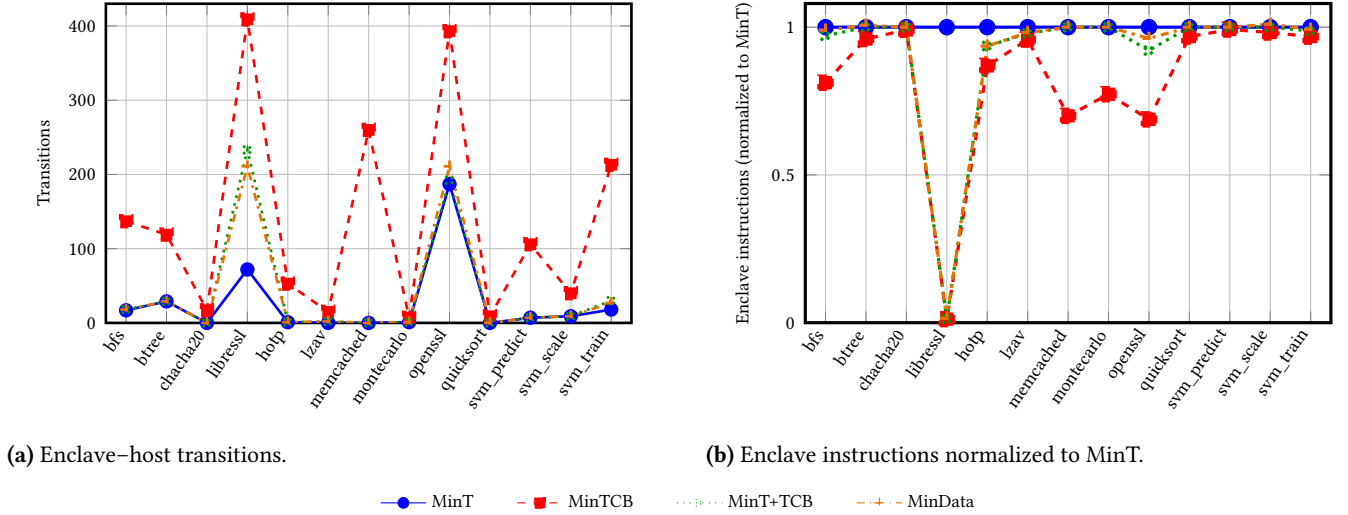
\begin{figure*}[t]
\centering

\begin{subfigure}[t]{0.48\textwidth}
\centering
\begin{tikzpicture}
\begin{axis}[
    width=\linewidth,
    height=5.8cm,
    ylabel={Transitions},
    xmin=0.5, xmax=13.5,
    ymin=0, ymax=430,
    xtick={1,...,13},
    xticklabels={
        bfs,
        btree,
        chacha20,
        libressl,
        hotp,
        lzav,
        memcached,
        montecarlo,
        openssl,
        quicksort,
        svm\_predict,
        svm\_scale,
        svm\_train
    },
    x tick label style={
        rotate=55,
        anchor=east,
        font=\scriptsize
    },
    y tick label style={font=\scriptsize},
    ylabel style={font=\scriptsize},
    grid=major,
    line width=1pt,
    mark size=2.2pt,
]

\addplot[
    blue,
    mark=*,
    solid
]
coordinates {
    (1, 17) (2, 29) (3, 0) (4, 72)
    (5, 1) (6, 0) (7, 0) (8, 1)
    (9, 187) (10, 0) (11, 7) (12, 9) (13, 18)
};

\addplot[
    red,
    mark=square*,
    dashed
]
coordinates {
    (1, 137) (2, 119) (3, 17) (4, 409)
    (5, 53) (6, 15) (7, 260) (8, 7)
    (9, 393) (10, 9) (11, 106) (12, 40) (13, 213)
};

\addplot[
    green!60!black,
    mark=x,
    dotted
]
coordinates {
    (1, 18) (2, 29) (3, 0) (4, 233)
    (5, 1) (6, 2) (7, 0) (8, 1)
    (9, 192) (10, 0) (11, 7) (12, 9) (13, 30)
};

\addplot[
    orange!90!black,
    mark=+,
    dashdotted
]
coordinates {
    (1, 18) (2, 29) (3, 0) (4, 211)
    (5, 1) (6, 2) (7, 0) (8, 1)
    (9, 211) (10, 0) (11, 7) (12, 9) (13, 27)
};

\end{axis}
\end{tikzpicture}
\caption{Enclave--host transitions.}
\label{fig:obj-transitions}
\end{subfigure}
\hfill
\begin{subfigure}[t]{0.48\textwidth}
\centering
\begin{tikzpicture}
\begin{axis}[
    width=\linewidth,
    height=5.8cm,
    ylabel={Enclave instructions (normalized to MinT)},
    xmin=0.5, xmax=13.5,
    ymin=0, ymax=1.08,
    xtick={1,...,13},
    xticklabels={
        bfs,
        btree,
        chacha20,
        libressl,
        hotp,
        lzav,
        memcached,
        montecarlo,
        openssl,
        quicksort,
        svm\_predict,
        svm\_scale,
        svm\_train
    },
    x tick label style={
        rotate=55,
        anchor=east,
        font=\scriptsize
    },
    y tick label style={font=\scriptsize},
    ylabel style={font=\scriptsize},
    grid=major,
    line width=1.2pt,
    mark size=2.5pt,
]

\addplot[
    blue,
    mark=*,
    solid
]
coordinates {
    (1,1.000) (2,1.000) (3,1.000) (4,1.000)
    (5,1.000) (6,1.000) (7,1.000) (8,1.000)
    (9,1.000) (10,1.000) (11,1.000) (12,1.000) (13,1.000)
};

\addplot[
    red,
    mark=square*,
    dashed
]
coordinates {
    (1, 0.812) (2, 0.960) (3, 0.991) (4, 0.012)
    (5, 0.869) (6, 0.958) (7, 0.700) (8, 0.773)
    (9, 0.688) (10, 0.968) (11, 0.992) (12, 0.982) (13, 0.968)
};

\addplot[
    green!60!black,
    mark=x,
    dotted
]
coordinates {
    (1, 0.971) (2, 1.000) (3, 1.000) (4, 0.013)
    (5, 0.936) (6, 0.979) (7, 1.000) (8, 1.000)
    (9, 0.921) (10, 1.000) (11, 1.000) (12, 1.000) (13, 0.987)
};

\addplot[
    orange!90!black,
    mark=+,
    dashdotted
]
coordinates {
    (1, 0.990) (2, 1.006) (3, 1.000) (4, 0.015)
    (5, 0.936) (6, 0.982) (7, 1.000) (8, 1.000)
    (9, 0.962) (10, 1.000) (11, 1.000) (12, 1.010) (13, 0.993)
};

\end{axis}
\end{tikzpicture}
\caption{Enclave instructions normalized to MinT.}
\label{fig:obj-enclave-size}
\end{subfigure}

\medskip

\begin{tikzpicture}
\begin{axis}[
    hide axis,
    xmin=0, xmax=1,
    ymin=0, ymax=1,
    legend columns=4,
    legend style={
        draw=none,
        font=\scriptsize,
        /tikz/every even column/.append style={column sep=0.5cm}
    }
]

\addlegendimage{blue,mark=*,solid}
\addlegendentry{MinT}

\addlegendimage{red,mark=square*,dashed}
\addlegendentry{MinTCB}

\addlegendimage{green!60!black,mark=x,dotted}
\addlegendentry{MinT+TCB}

\addlegendimage{orange!90!black,mark=+,dashdotted}
\addlegendentry{MinData}

\end{axis}
\end{tikzpicture}

\caption{Comparison of partitioning objectives across benchmarks:
(a) number of enclave--host transitions and
(b) number of enclave instructions normalized to MinT for each benchmark.}
\label{fig:objective-comparison}
\end{figure*}

\paragraph{Impact of Optimization Objectives.}
Figure~\ref{fig:objective-comparison} compares the partitions produced by \tmin, \emin, \tmin+\emin, and \objdata in terms of enclave--host transitions and enclave size. Figure~\ref{fig:obj-transitions} shows the expected trade-off between minimizing transitions and minimizing the TCB. \tmin generally produces the fewest boundary crossings, while \emin may introduce additional transitions to move more code out of the enclave. This trade-off is particularly visible for larger applications. \code{LibreSSL} exhibits the most pronounced difference: despite having 275,942 LLVM IR instructions, only 137 (${\sim}0.05\%$) are labeled private, and different placement objectives can therefore make substantially different decisions about the large amount of public code surrounding the small security-sensitive core. \code{Memcached}, \code{OpenSSL}, and \code{SVM\_Train} similarly show a considerable degree of variation between \emin and the transition-oriented objectives. In contrast, several smaller benchmarks and the other SVM workloads show less variation in transitions, indicating that their security and placement constraints leave fewer meaningful optimization choices.

Figure~\ref{fig:obj-enclave-size} reports enclave instructions normalized to the \tmin solution for each benchmark, making the relative effect of each objective visible despite large differences in benchmark size. As expected, \emin generally favors smaller enclaves, whereas \tmin may retain additional public code inside the enclave when doing so avoids boundary crossings. The combined \objboth objective typically lies between these two extremes, while \objdata favors placements that reduce data movement rather than enclave size directly. The largest differences occur in benchmarks where relatively sparse private instructions are embedded within substantially larger programs, giving \splitacronym greater freedom to trade enclave size against transitions. Conversely, when most of the program is private, or when the placement constraints effectively determine the partition, the four objectives converge to similar enclave sizes. Overall, the results demonstrate that the optimal partition depends on the deployment objective: reducing the TCB can increase enclave--host communication, while minimizing transitions may require retaining substantially more code inside the enclave.

\begin{figure}[ht]
\centering
\begin{tikzpicture}
\begin{axis}[
    width=0.48\textwidth,
    height=5.0cm,
    ylabel={Adjusted slowdown ($\times$)},
    ymode=log,
    ymin=0.5, ymax=50,
    xmin=0.5, xmax=13.5,
    xtick={1,...,13},
    xticklabels={
        bfs,btree,chacha20,libressl, hotp,lzav,
        memcached,montecarlo,openssl,quicksort,
        svm\_predict,svm\_scale,svm\_train
    },
    x tick label style={
        rotate=60,
        anchor=east,
        font=\tiny
    },
    y tick label style={font=\tiny},
    ylabel style={font=\scriptsize},
    grid=major,
    line width=0.8pt,
    mark size=1.5pt,
    legend style={
        at={(0.5,1.01)},
        anchor=south,
        legend columns=4,
        draw=none,
        font=\tiny,
        column sep=3pt
    }
]

\addplot[blue,mark=*,solid] coordinates {
(1,2.17) (2,0.87) (3,41.10) (4,11.69) (5,42.21)
(6,41.46) (7,0.99) (8,5.47) (9,0.88) (10,41.11)
(11,1.22) (12,5.38) (13,3.54)};
\addlegendentry{MinT}

\addplot[red,mark=square*,dashed] coordinates {
(1,1.51) (2,0.88) (3,40.69) (4,10.03) (5,42.32)
(6,41.02) (7,0.99) (8,5.64) (9,1.30) (10,40.80)
(11,1.22) (12,5.36) (13,3.61)};
\addlegendentry{MinTCB}

\addplot[green!60!black,mark=x,dotted] coordinates {
(1,2.18) (2,0.89) (3,40.06) (4,11.56) (5,43.15)
(6,41.94) (7,1.01) (8,5.55) (9,1.35) (10,42.92)
(11,1.22) (12,5.38) (13,3.55)};
\addlegendentry{MinT+TCB}

\addplot[orange!90!black,mark=+,dashdotted] coordinates {
(1,2.22) (2,0.98) (3,41.07) (4,11.37) (5,41.92)
(6,40.60) (7,1.00) (8,5.52) (9,1.74) (10,40.25)
(11,1.22) (12,5.35) (13,3.60)};
\addlegendentry{MinData}

\end{axis}
\end{tikzpicture}
\caption{Adjusted slowdown across partitioning objectives.}
\label{fig:slowdown}
\end{figure}
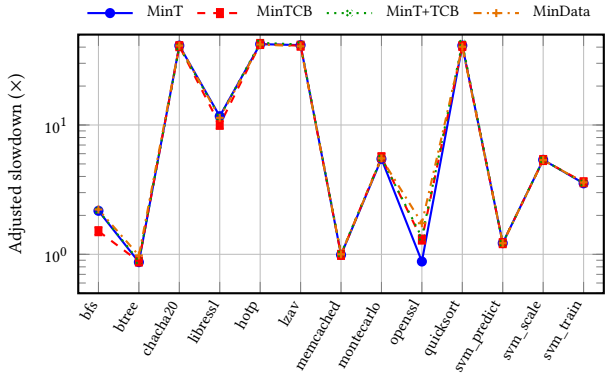

Figure~\ref{fig:slowdown} reports the adjusted runtime slowdown across
partitioning objectives. The selected objective has relatively little
impact on performance for most benchmarks, although OpenSSL is a notable
outlier, ranging from $0.88\times$ under \tmin to $1.74\times$ under \objdata.
BTree, Memcached, and SVM Predict remain close to native
performance, while BFS incurs a moderate ${\sim} 2\times$ overhead.
Short-running workloads such as ChaCha20, HOTP, LZAV, and Quicksort
show approximately $39$--$42\times$ slowdown, suggesting that fixed
enclave execution costs dominate their native runtimes.
LibreSSL uses arbitrary-precision arithmetic for RSA keys through BIGNUM structures containing deeply nested pointers that cannot be readily marshalled across the enclave boundary. The corresponding operations are therefore pinned to the host, inducing frequent enclave--host transitions and a substantial slowdown of ${\sim}12\times$.

Interestingly, the SVM benchmarks demonstrate \objdata in action. $\mathsf{svm\_scale}$ performs substantial computation but invokes logging, which must be pinned to the host because it performs I/O. Separating the computation from logging induces heavy enclave--host data marshalling, resulting in a $120\times$ slowdown. The objective \objdata instead co-locates $\mathsf{svm\_scale}$ with logging on the host, minimizing boundary data transfer and substantially reducing the slowdown.

  \section{Related Work}\label{sec:rel}
  
\noindent \textbf{IFC for Low-level Languages.}
While prior IFC work is extensive~\cite{volpano, sm-jsac, pottier, hunt:flowsensitive, bastys2022secwasm, manivannan2010jvmifc, stefan2015ifc, porter2014finegrainedifc, roy2009laminar, vassena2019fine, denning1976lattice, goguen1982security}, we focus on low-level languages with unstructured control flow, where permissive tracking often relies on post-dominator–like notions.
\citet{webkitjs} use dynamic $\pc$ stacks with post-dominators as declassification points; SIFTAL~\cite{siftal} tracks convergence via linear continuations; \citet{aldous2015} use abstract interpretation to compute influence regions. In contrast, our approach is static and type-based, enabling comparable expressiveness with simpler reasoning.
Prior work on security-type-preserving compilation~\citet{sectypebarthe} enforces IFC via type preservation across compilation from a structured languages to assembly, but does not leverage post-dominator–based reasoning to handle implicit flows in unstructured control flow.
 
\noindent \textbf{Program Partitioning for TEEs.}
Since we have extensively compared our work with \citet{gcoopsla2016}, we focus  on other related works. We also exclude works that partition their programs for distributed settings~\cite{hydra, dflate}.
Civet~\cite{civet}, Glamdring~\cite{glamdring}, and Cadote~\cite{cadote} perform automatic partitioning of Java, C, and Rust programs, but provide no formal security guarantees.
\citet{rubinov} uses taint tracking for Android partitioning but lacks implicit flow support and automatic enclave porting.
HasTEE~\cite{hastee} targets noninterference in Haskell, yet relies on programmer-specified monads and does not support fine-grained IFC or enclave inference. By contrast, our approach provides automatic, low-level partitioning with formal noninterference guarantees and practical compiler support. Moreover, none of these works offer support for  partitioning trade-offs.

Privtrans~\cite{privtrans} and PtrSplit~\cite{ptrsplit} perform automatic
partitioning for privilege separation, with PtrSplit supporting general
pointers. In contrast, we enforce noninterference against low-level
attackers, providing formal guarantees and fine-grained host–enclave
partitioning under unsafe memory.

Viaduct~\cite{viaduct}, similar to Jif/split~\cite{spp,zheng2003replication}, enables secure computation among mutually distrustful principals by selecting trusted hosts and using cryptographic mechanisms when necessary. Unlike our work, these systems do not target TEEs or address low-level challenges of LLVM IR, such as unstructured control flow and pointer arithmetic. While Viaduct also uses constraint optimization, it minimizes the cost of cryptographic mechanisms, whereas we minimize the cost of enclave usage.




  \bibliographystyle{ACM-Reference-Format}
  \bibliography{mybib}


  \fi 

  \fi 
  
  \ifappendix
  \appendix
  \tableofcontents  
  \input{llvm-enclaves-appendix}
  \bibliography{mybib}
  \bibliographystyle{abbrvnat}

  \fi


\end{document}